\documentclass[11pt]{article} 

\usepackage[letterpaper,margin=1in]{geometry}

\usepackage[letterpaper,margin=1in]{geometry}

\newif\ifred
\redtrue

\usepackage[T1]{fontenc}
\usepackage{microtype}

\usepackage{titlesec}
\titleformat*{\paragraph}{\bfseries}

\usepackage{amsthm,amsmath,amssymb,amsfonts,amssymb,mathtools}
\usepackage{amsmath,amssymb,amsfonts,amssymb,mathtools}
\usepackage{booktabs}
\usepackage{array}
\usepackage{xcolor}
\usepackage{empheq}
\usepackage{dsfont}
\usepackage{mathrsfs}
\usepackage{graphicx}
\usepackage{enumitem}
\usepackage{aliascnt} 
\usepackage{xspace}
\usepackage{bbm, bm}

\newcommand{\llbracket}{\mathopen{[\![}}
\newcommand{\rrbracket}{\mathclose{]\!]}}

\usepackage{caption}
\usepackage{tikz}
\usetikzlibrary{patterns}
\usetikzlibrary{patterns}
\usetikzlibrary{arrows,shapes,automata,backgrounds,petri,positioning}
\usetikzlibrary{shadows}
\usetikzlibrary{calc}
\usetikzlibrary{spy}
\usetikzlibrary{angles, quotes}

\usetikzlibrary{matrix}
\usepackage{pgf,pgfplots}
\usepackage{pgfmath,pgffor}
\pgfplotsset{compat=1.17}

\usepackage{framed}

 \usepackage{algorithm}

\usepackage[noend]{algpseudocode}

\definecolor[named]{ACMBlue}{cmyk}{1,0.1,0,0.1}
\definecolor[named]{ACMYellow}{cmyk}{0,0.16,1,0}
\definecolor[named]{ACMOrange}{cmyk}{0,0.42,1,0.01}
\definecolor[named]{ACMRed}{cmyk}{0,0.90,0.86,0}
\definecolor[named]{ACMLightBlue}{cmyk}{0.49,0.01,0,0}
\definecolor[named]{ACMGreen}{cmyk}{0.20,0,1,0.19}
\definecolor[named]{ACMPurple}{cmyk}{0.55,1,0,0.15}
\definecolor[named]{ACMDarkBlue}{cmyk}{1,0.58,0,0.21}

\usepackage[colorlinks,citecolor=blue,linkcolor=magenta,bookmarks=true]{hyperref}
 \usepackage[nameinlink]{cleveref}

\creflabelformat{ineq}{#2{\upshape(#1)}#3}
\crefname{sub}{Subsection}{Subsection}
\creflabelformat{Subsection}{#2{\upshape(#1)}#3}
\crefname{sdp}{SDP}{SDP}
\creflabelformat{sdp}{#2{\upshape(#1)}#3}
\crefname{lp}{LP}{LP}
\creflabelformat{lp}{#2{\upshape(#1)}#3}
\usepackage{aliascnt}
\usepackage{cleveref}
\crefname{ineq}{Inequality}{Inequality}
\creflabelformat{ineq}{#2{\upshape(#1)}#3}
\crefname{sub}{Subsection}{Subsection}
\creflabelformat{Subsection}{#2{\upshape(#1)}#3}
\crefname{sdp}{SDP}{SDP}
\creflabelformat{sdp}{#2{\upshape(#1)}#3}
\crefname{lp}{LP}{LP}
\creflabelformat{lp}{#2{\upshape(#1)}#3}



\newtheorem{theorem}{Theorem}[section]

\newtheorem{lemma}{Lemma}[section]
\newtheorem{informal theorem}[theorem]{Theorem (informal statement)}

\newtheorem{proposition}[theorem]{Proposition}
\newtheorem{corollary}[theorem]{Corollary}

\newtheorem{remark}[theorem]{Remark}

\newtheorem{definition}[theorem]{Definition}
\newtheorem*{definition*}{Definition}

\newcommand{\lp}{\left}
\newcommand{\rp}{\right}
\newcommand{\p}[1]{\left( #1\right)}
\newcommand\norm[1]{\left\| #1 \right\|}
\newcommand\snorm[2]{\left\| #2 \right\|_{#1}}

\DeclareMathOperator*{\pr}{\mathrm{Pr}}
\DeclareMathOperator*{\E}{\mathbf{E}}

\newcommand{\X}{\mathcal X }
\newcommand{\calC}{\mathcal C }

\newcommand{\EX}{\mathsf{EX}}
\newcommand{\QEX}{\mathsf{QEX}}

\newcommand{\R}{\mathbb{R}}

\newcommand{\Z}{\mathbb{Z}}
\newcommand{\N}{\mathbb{N}}
\newcommand{\eps}{\epsilon}
\newcommand{\dtv}{d_{\mathrm TV}}
\providecommand{\dHam}{d_{\mathrm Ham}}
\providecommand{\NS}{\operatorname{NS}}
\newcommand{\poly}{\mathrm{poly}}

\newcommand{\polylog}{\mathrm{polylog}}

\newcommand{\ket}[1]{\lvert #1\rangle}
\newcommand{\bra}[1]{\langle #1\rvert}
\newcommand{\braket}[2]{\langle #1\mid #2\rangle}
\newcommand{\inner}[2]{\langle #1, #2\rangle}

\newcommand{\sgn}{\mathrm{sign}}
\newcommand{\calN}{{\cal N}}

\newcommand{\D}{D}

\newcommand{\Ind}{\mathds{1}}
\newcommand{\1}{\Ind}

\newcommand{\sphere}{\mathbb{S}}

\newcommand{\QHT}{\mathsf{QHT}}

\newcommand{\abs}[1]{\lp| #1 \rp|}

\renewcommand\Pr{\pr}

\begin{document}

\title{Fast Quantum Algorithms for Learning Linear Threshold Functions}

\author{
Aleksandrs Krivcenko\thanks{QuSoft, CWI, the Netherlands, and Télécom Paris, Institut Polytechnique de Paris, France. Supported by the Unitary Foundation microgrant ``Designing quantum query algorithms on a laptop'', and the France Excellence Quantum Scholarship program awarded by IFNL\@.   \nolinkurl{krivcenko@telecom-paris.fr} and \nolinkurl{aleksandrs.krivcenko@gmail.com}}
\and
Tuyen Nguyen\thanks{University of Technology Sydney, Ultimo, NSW, Australia. Work partially done while visiting QuSoft, CWI, the Netherlands. Supported by a scholarship from the Sydney Quantum Academy, PHDR06031. {\tt Tuyen.Q.Nguyen@student.uts.edu.au}}
\and
Ronald de Wolf\thanks{QuSoft, CWI and University of Amsterdam, the Netherlands. Partially supported by the Dutch Research Council (NWO) through Gravitation-grant Quantum Software Consortium, 024.003.037. {\tt rdewolf@cwi.nl}}
}
\date{}

\maketitle

\begin{abstract}
    Linear threshold functions (LTFs) are $\{-1,1\}$-valued functions
$f_{w,\theta}(x)=\sgn(\inner{w}{x}-\theta)$,
where $w\in\mathbb{R}^n$ has unit norm, $\theta\in\mathbb{R}$,
and $x\in\mathbb{R}^n$ or $\{-1,1\}^n$.
If $\theta=0$, the LTF is homogeneous, denoted $f_w=f_{w,0}$.
LTFs are fundamental linear classifiers.
We give three results: one for general LTFs and two for homogeneous LTFs.
\begin{enumerate}
    \item
    With real-domain membership queries evaluating $f_{w,\theta}(x)$
    at any chosen $x\in\mathbb{R}^n$, we learn an LTF with
    classification error at most $\eps$ under the standard Gaussian
    distribution using $O(\log(n/\eps))$ queries and
    $\widetilde{O}(n)$ other gates.
    Classically, $\Omega(n\log(1/\eps))$ queries are required,
    giving an exponential quantum query speedup.

    \item
    A homogeneous LTF $f_w$ over $\{-1,1\}^n$ whose $k$ nonzero
    weights are equal and positive is Majority on the support of $w$.
    We identify this hidden support exactly using $O(\log k)$ queries,
    exponentially improving Belovs'
    $O(k^{1/4})$ upper bound.
    Our algorithm uses $O(n\log n\log k)$ gates and has
    $O(\log n\log k)$ non-oracle depth.
    Classical algorithms require $\Omega(k\log(n/k))$ queries.

    \item
    Suppose a unitary $U$ prepares discretized Gaussian quantum
    examples, corresponding to
    $\int_x\sqrt{\gamma_n(x)}\ket{x}\ket{f_w(x)}\,dx$.
    This is weaker than membership queries.
    Using the efficient quantum Hermite transform circuit of Jain et al., we learn homogeneous LTFs with classification
    error at most $\eps$ under the Gaussian distribution using
    $O(n^{1/4}/\sqrt{\eps})$ applications of $U$ and $U^\dagger$
    and $\widetilde{O}\!\left(
        \frac{n^{5/4}}{\sqrt{\eps}}+\frac{n}{\eps^2}
    \right)$ other gates.
    Classically, $\Omega(n/\eps)$ examples are needed,
    giving a quartic quantum speedup in $n$.
\end{enumerate}
We show that the quantum query complexities in all three cases are optimal by proving matching lower bounds.
\end{abstract}

\setcounter{page}{0}

\thispagestyle{empty}

\newpage

\tableofcontents

\newpage

\section{Introduction}

\subsection{Linear threshold functions}

Linear threshold functions (LTFs) are $\pm 1$-valued functions of the form $f_{w,\theta}(x)=\sgn(\inner{w}{x}-\theta)$, where the \emph{weight} vector $w\in\mathbb{R}^n$ is a unit vector, the \emph{threshold} $\theta\in\mathbb{R}$, and $x$ typically comes from $\mathbb{R}^n$ or $\{-1,1\}^n$~\cite{Muroga1971}. If $\theta = 0$, the function is called \emph{homogeneous} and we write \(f_w:=f_{w,0}\).

Such functions are among the most fundamental objects in machine learning, as they linearly separate positive and negative examples by a hyperplane (or equivalently, they define a halfspace); in the homogeneous case $\theta=0$, the separating hyperplane passes through the origin. 
They can also be used to model individual neurons in neural networks. 
The study of learning LTFs dates back to the Perceptron algorithm~\cite{Rosenblatt1958} and has strongly influenced the development of many foundational learning methods, including support vector machines~\cite{vapnik1997support} and AdaBoost~\cite{freund1997decision}.

The problem of learning LTFs has previously been considered in several settings, depending on the type of access one has to~$f_{w,\theta}$. For example, one can try to learn an approximation of~$f_{w,\theta}$ from multiple labeled examples of the form $(x,f_{w,\theta}(x))$, with either a known or an unknown distribution on~$x$. One can also try to learn $f_{w,\theta}$ by evaluating it at inputs of one's choice, using so-called
``membership queries''. Under the standard Gaussian distribution, if we learn a unit vector $\tilde w$ satisfying
$\norm{w-\tilde w}_2=O(\eps)$ and a threshold $\tilde\theta$ satisfying $|{\theta-\tilde\theta}|=O(\eps)$, then the corresponding LTFs disagree on at most an $O(\eps)$ fraction of inputs, as shown in~\Cref{lem:gaussian-disagreement}. For homogeneous LTFs, this relation is stronger: the Euclidean error in the weight vector equals (up to constant factors) classification error under any rotationally-invariant distribution that
assigns zero mass to the origin, such as the Gaussian distribution, or the Haar measure on the sphere. For general LTFs, however, the converse need not hold. We therefore formulate our learning guarantees for general LTFs (\Cref{ssec:Rn}) directly in terms of classification error under the Gaussian distribution. 

Classically, the query complexity of learning LTFs with \emph{membership queries} has been studied extensively. With this form of access, a lower bound follows from basic information theory: an $\eps$-packing of $\mathbb S^{n-1}$ contains $(1/\eps)^{\Omega(n)}$ distinguishable weight vectors, while each classical query reveals only one bit. Consequently, identifying the target halfspace to error $\eps$ requires $\Omega\!\left(n\log(1/\eps)\right)$ queries in the worst case, even when restricting to homogeneous LTFs~\cite{kulkarni1993active}. A strictly weaker access model is the standard \emph{example} or \emph{passive-learning} model, in which the learner does not choose the inputs but instead receives independent labeled examples $(x,f_{w, \theta}(x))$, where $x$ is sampled from a distribution $\mathcal{D}$.\footnote{When the distribution $\mathcal D$ is known and can be sampled efficiently, membership-query access can simulate such examples simply by sampling $x\sim\mathcal D$ and querying $f_{w, \theta}(x)$. In contrast, example access does not, in general, allow the learner to evaluate $f_{w, \theta}$ at inputs of its choice. A general separation for Boolean functions is discussed in~\cite{arunachalam2017guest}.} For LTFs in $n$ dimensions, Long~\cite{long1995sample} established a lower bound of $\Omega(n/\eps)$ passive examples under the uniform distribution on the unit ball, and a matching upper bound was subsequently obtained~\cite{Long2003}. For homogeneous LTFs, the same scaling transfers to the standard Gaussian distribution by radial symmetry~\cite{diakonikolas2021agnostic}. Hence, for constant success probability, passive learning requires $\Theta(n/\eps)$ examples, which is an exponentially worse $\eps$-dependence than the logarithmic dependence available in the stronger membership-query access.

On the quantum side, Kapoor, Wiebe, and Svore~\cite{kapoor2016quantum} studied quantum algorithms for perceptron learning, obtaining quadratic
improvements in the dependence on the number of training examples and on the
margin parameter in their respective settings. Most relevant to our work is Belovs~\cite{belovs2015symmetricjuntas} on learning symmetric juntas: $n$-bit Boolean functions that depend only on an unknown subset of $k$ input bits, on which it acts as a known symmetric $k$-bit function~$h$. When $h$ is Majority, which is itself a homogeneous Boolean threshold function, Belovs showed that the unknown set of $k$ relevant variables can be identified using $O(k^{1/4})$ quantum membership queries, compared with $\Omega(k\log(n/k))$ randomized classical queries~\cite{bshouty1996oracles}, yielding a quartic separation. This problem can be viewed as a restricted instance of learning a homogeneous LTF over the Boolean domain~$x\in\{-1,1\}^n$, where the weight vector $w$ assigns positive equal weight to the coordinates in an unknown support of size~$k$. In this setting, learning $w$ is equivalent to identifying the hidden support of the corresponding Majority-junta.

Quantum learning from quantum examples (superpositions of the form $\sum_x\sqrt{{\cal D}(x)}\ket{x,f(x)}$) has also been studied since the work of Bshouty and Jackson~\cite{bshouty1995learning}. While quantum examples can provide advantages for particular fixed distributions, Arunachalam and de Wolf~\cite{arunachalam_optimal_2017, arunachalam2017guest} showed that in the distribution-free example model, the quantum and classical complexities coincide up to constant factors. These general results, however, do not preclude stronger quantum advantages in structured settings, such as learning homogeneous LTFs under a fixed Gaussian distribution.

\subsection{Overview of our results}

We give three positive results about efficiently learning LTFs from various kinds of quantum access (one for general LTFs and two for homogeneous LTFs) in terms of their query and gate complexities, with essentially matching lower bounds.
Our results are summarized in Table~\ref{tab:query-complexities}, and are described in more detail in the next three subsections.

\begin{table}[htb]
    \centering
    \begin{tabular}{|l|c|c|c|c|}
\hline
& \multicolumn{2}{c|}{\textbf{Classical}}
& \multicolumn{2}{c|}{\textbf{Quantum}} \\
\cline{2-5}
\multicolumn{1}{|c|}{}
& \emph{Upper}
& \emph{Lower}
& \emph{Upper}
& \emph{Lower} \\
\hline
\shortstack[l]{Real membership\\[-1pt]
query (\Cref{sec:quantum-membership})}
& \shortstack{
    $\widetilde O\!\left(n\log(1/\eps)\right)$\\[-1pt]
    \cite{hopkins2020point}}
& \shortstack{
    $\Omega\!\left(n\log(1/\eps)\right)$\\[-1pt]
    \cite{kulkarni1993active}\\[-1pt]
    \footnotesize see also~\Cref{thm:classical-main}}
& \shortstack{
    $O\!\left(\log(n/\eps)\right)$\\[-1pt]
    \Cref{thm:quantum-ltf-learning}}
& \shortstack{
    $\Omega\!\left(\log(n/\eps)\right)$\\[-1pt]
    \Cref{thm:quantum-ltf-lower bound}}
\\
\hline
\shortstack[l]{Boolean membership\\[-1pt]
query (\Cref{sec:boolean})}
& \shortstack{
    $O\!\left(k\log(n/k)\right)$\\[-1pt]
    \Cref{prop:majority-classical-upper-bound}}
& \shortstack{
    $\Omega\!\left(k\log(n/k)\right)$\\[-1pt]
    \cite{bshouty1996oracles}}
& \shortstack{
    $O(\log k)$\\[-1pt]
    \Cref{thm:majority-upper-bound}\\[-1pt]
    \footnotesize prev.\ $O(k^{1/4})$~\cite{belovs2015symmetricjuntas}}
& \shortstack{
    $\Omega(\log k)$\\[-1pt]
    \Cref{thm:majority-lower-bounds}}
\\
\hline
\shortstack[l]{Example query\\[-1pt]
(\Cref{sec:example})}
& \shortstack{
    $O(n/\eps)$\\[-1pt]
    \cite{Long2003}}
& \shortstack{
    $\Omega(n/\eps)$\\[-1pt]
    \cite{BalcanLong2013}\\[-1pt]
    \footnotesize see also~\Cref{thm:classical-example-main}}
& \shortstack{
    $O\!\left(n^{1/4}/\sqrt{\eps}\right)$\\[-1pt]
    \Cref{thm:quantum-example-learning}}
& \shortstack{
    $\Omega\!\left(n^{1/4}/\sqrt{\eps}\right)$\\[-1pt]
    \Cref{thm:gaussian-assouad-lower-bound}}
\\
\hline
\end{tabular}
    \caption{Overview of our main results in comparison to what was known before. Here $n$ is the
    dimension of $w$ and $x$, $\eps$ is the target error, $k\geq2$ is the support size or an upper
    bound on it, and the notation $\widetilde O(\cdot)$ hides polylogarithmic factors in $n$ and
    $1/\eps$.\protect\footnotemark}
    \label{tab:query-complexities}
\end{table}
\footnotetext{For Boolean membership queries, all nonzero entries of each candidate vector $w$
are assumed to have the same positive value. The classical upper bound $O(k\log(n/k))$ holds in
the regime $\liminf_{k\to \infty} n/k > 1$. The quantum upper bound holds for every $n\geq k$.
The restriction $n\geq2k$ applies only to the lower bound under the promise $|A|=k$; see
\Cref{thm:eqk}. Under the promise $1\leq |A|\leq k$, the lower bound holds for every $n\geq k$;
see \Cref{thm:majority-lower-bounds}.}

\subsection{Learning general LTFs from real membership queries}\label{ssec:Rn}
Membership-query access to the target function means that the learner may choose the point at which the function is evaluated. Classically, a membership oracle returns $f(x)$ on an arbitrary query $x\in\mathcal X$, where $\cal X$ is typically $\mathbb{R}^n$ or $\{0,1\}^n$ or $\{-1,1\}^n$. For $\pm1$-valued functions, define $\bar f(x):=(1-f(x))/2\in\{0,1\}$. The coherent quantum membership oracle is the unitary
\begin{equation}
    O_f:\ket{x}\ket{b}
    \longmapsto
    \ket{x}\ket{b\oplus\bar f(x)},
\end{equation}
for Boolean labels, or a phase oracle
\begin{equation}
\label{eq:phase-oracle}
    O_f^{\pm}\ket{x}
    =
    f(x)\ket{x},
    \qquad
    f(x)\in\{\pm1\},
\end{equation}
which may be obtained from the bit-output oracle by phase kickback (these two types of oracles are equivalent if we can apply them in a controlled manner).  The crucial distinction with classical membership queries is that $O_f$ may be applied to a superposition of query points, allowing the learner to process many function values coherently in superposition.  

Suppose we are given membership-query access to an unknown LTF $f_{w,\theta}:\mathbb{R}^n\to\{-1,1\}$. We present a quantum algorithm that learns $f_{w,\theta}$ to classification error at most~$\eps$ under the standard Gaussian distribution using only $O(\log(n/\eps))$ membership queries and $\widetilde{O}(n)$ additional gates. In contrast, any classical algorithm requires $\Omega(n\log(1/\eps))$ membership queries in the worst case, even when restricted to homogeneous LTFs~\cite{kulkarni1993active}. Our result therefore establishes an exponential quantum--classical separation in query complexity for this important learning problem. Moreover, we prove a matching quantum lower bound of $\Omega(\log(n/\eps))$ membership queries, demonstrating that our algorithm is indeed optimal.

\paragraph{The Algorithm.} The main idea of our algorithm is inspired by the quantum convex optimization techniques of~\cite{van_Apeldoorn_2020,chakrabarti2018QuantumConvexOpt}, which use quantum membership queries to a convex set~$K$ to efficiently implement a ``separation query'' for~$K$. Central to this approach is the \emph{height function}, which characterizes the boundary of a convex set along a fixed direction. The key observation is that the gradient of the height function determines a supporting hyperplane of~$K$. In our setting, where~$K$ is a halfspace, the height function is affine, and its gradient is proportional to the unknown normal vector of the halfspace. Then, using the quantum gradient estimation algorithm~\cite{Jordan_2005} allows us to directly recover the target halfspace. Therefore, our exponential separation directly comes from the efficiency of implementing the height function query using quantum membership queries~\cite{van_Apeldoorn_2020} and the advantage for gradient estimation from Jordan's algorithm~\cite{Jordan_2005}.   

Concretely, for a halfspace $f_{w,\theta}(x)=\sgn(\inner{w}{x}-\theta)$, consider the convex set 
\begin{equation}
    K_{w,\theta}:=\{x\in\mathbb R^n:\inner{w}{x}\leq\theta\}.
\end{equation}
Its boundary is a hyperplane $\partial K_{w,\theta}=\{x\in\mathbb R^n:\inner{w}{x}=\theta\}$. Fix a reference direction $p\in\R^n$ satisfying $\langle w,p\rangle\neq0$. We do not assume that the orientation of $p$ relative to $w$ is known. For every $z\in\R^n$, the affine line $z+\R p:=\{z+tp:t\in\R\}$ intersects $\partial K_{w,\theta}$ at a unique point. We define the \emph{height function} $h_p(z)$ to be the unique scalar satisfying
\begin{equation*}
    z+h_p(z)p\in\partial K_{w,\theta}.
\end{equation*}
Equivalently,
\begin{equation}
    \langle w,z\rangle
    +
    h_p(z)\langle w,p\rangle
    =
    \theta.
\end{equation}
Solving for the height gives
\begin{equation*}
    h_p(z)
    =
    h_p(0)
    -
    \frac{\langle w,z\rangle}{\langle w,p\rangle}, \quad \mbox{for}\quad  h_p(0):=\frac{\theta}{\langle w,p\rangle}.
\end{equation*}
Thus, $h_p$ is exactly an \emph{affine} function of $z$, and its gradient is
\begin{equation*}
    \nabla h_p(z)
    =
    -\frac{w}{\langle w,p\rangle},
    \qquad
    \text{for every }z\in\R^n.
\end{equation*}
The threshold $\theta$ therefore affects only the additive offset, while the gradient is independent of both $z$ and $\theta$ and remains collinear with the unknown normal vector $w$. In particular,
\begin{equation*}
    \frac{\nabla h_p}{\|\nabla h_p\|}
    =
    -\sgn(\langle w,p\rangle)\,w.
\end{equation*}
Hence, estimating the gradient determines the hidden vector $w$, while estimating the offset $h_p(0)$ allows us to recover its offset~$\theta$.

Our algorithm therefore proceeds in two stages: first, locating the separating hyperplane along a suitable reference direction, and second, recovering its normal vector.

For the first stage, conditioned on a good choice of $p$ and on the nontrivial regime where $\theta$ is bounded, the line $\mathbb{R}p$ intersects the separating hyperplane at the unique point $h_p(0)p$. In particular, for any scalar $r$, a membership query at $rp$ determines whether $r$ lies above or below $h_p(0)$, provided that the orientation $\sgn(\langle w,p\rangle)$ is known. We can therefore estimate $h_p(0)$ to the desired precision by binary search, using only logarithmically many membership queries, as shown  in~\Cref{sec:estimate-offset}.

For the second stage, we exploit the fact that membership queries can be used to construct a comparison oracle for the height function $h_p$. Indeed, by the definition of $h_p$, we have
\begin{equation*}
    \langle w,z+tp\rangle-\theta
    =
    \langle w,p\rangle
    \bigl(t-h_p(z)\bigr).
\end{equation*}
Consequently, once the orientation of $p$ is known, a membership query at $z+tp$ determines whether $t$ lies above or below $h_p(z)$. As detailed in~\Cref{sec:height-oracle}, this allows us to implement a coherent binary-search procedure that approximates $h_p(z)$ over a superposition of grid points and encodes the resulting height values into quantum phases, following~\cite{van_Apeldoorn_2020}. We then apply Jordan's quantum gradient-estimation algorithm~\cite{Jordan_2005}, which exploits the quantum Fourier transform to estimate all coordinates of $\nabla h_p$ simultaneously. Since $h_p$ is affine, its gradient is constant and can, in principle, be recovered using a constant number of queries to an ideal phase oracle. In our setting, however, the phase oracle is only approximate, as the height values are obtained through finite-precision coherent binary search. To account for these approximation errors, we employ the robust gradient-estimation analysis of~\cite{Gily_n_2019}, which yields logarithmic membership-query complexity, shown in~\Cref{sec:gradient-estimation}. Finally, normalizing the estimated gradient and correcting its sign using the orientation of $p$ yields an estimate of the normal vector $w$. The previously estimated value $h_p(0)$ then determines $\theta$. Combining these two stages, we recover the parameters of a general LTF using only logarithmically many membership queries.

The following theorem formalizes this approach, showing that general LTFs can be learned using only logarithmically many membership queries and $\widetilde O(n)$ additional quantum gates.  For a probability distribution $\mathcal D$ over a domain $\mathcal X$
and two Boolean functions $f,g:\mathcal X\to\{\pm1\}$,
we denote their disagreement probability under $\mathcal D$ by
$
    d_{\mathcal D}(f,g)
    :=
    \Pr_{x\sim \mathcal  D}[f(x)\neq g(x)].
$
In particular, $d_{\gamma_n}(f,g)$ denotes the
disagreement probability under the standard Gaussian
distribution $\gamma_n$ on $\mathbb{R}^n$.

\begin{theorem}[Learning general LTFs with Membership Query;~\Cref{thm:quantum-ltf-learning}]\label{thrm:1.1}
For $n\geq 2$, let $f_{w,\theta}(x)=\sgn(\langle w,x\rangle-\theta)$, where
$w\in\sphere^{n-1}$ and $\theta\in\R$, and suppose that the
learner has coherent membership-query access to $f_{w,\theta}$.
For every $0<\eps\leq1/2$, there exists a quantum algorithm
that outputs a hypothesis $h:\R^n\to\{-1,+1\}$ such
that $d_{\gamma_n}(h,f_{w, \theta})\leq \eps$
with probability at least $2/3$. The algorithm uses $$O\left(\log\frac{n}{\eps}\right)$$ membership queries, and $\widetilde O(n)$ additional gates.
\end{theorem}

\paragraph{Matching Lower Bound.} Given the fact that Jordan's gradient estimation algorithm could yield a constant query complexity if we could obtain the ideal phase oracle, it is natural to ask whether there might be an algorithm for this problem with a number of queries that is independent of~$n$. However, we rule this out by proving that our query complexity is  optimal in both parameters $n$ and $\eps$, even for the special case of homogeneous LTFs.
The proof uses a continuous-parameter version of the multiplicative adversary method~\cite{spalek2007multiplicativequantumadversary}. 

\begin{theorem}[Quantum membership-query lower bound;~\Cref{thm:quantum-ltf-lower bound}]
There is a universal constant $\eps_0>0$ such that, for every $n\geq2$ and $0<\eps\leq\eps_0$, any quantum learner that, for every $w\in \sphere^{n-1}$, outputs a classical hypothesis $h$ with $d_{\gamma_n}(h,f_w)\leq\eps$ with probability at least $2/3$ requires
\begin{equation}
    \Omega\!\left(\log\frac{n}{\eps}\right)
\end{equation}
membership queries. This holds even for adaptive algorithms making superposition queries at arbitrary real points and outputting improper hypotheses (i.e., an $h$ that is not itself an LTF).
\end{theorem}

\subsection{Learning a \texorpdfstring{weight-$k$ vector~$w$}{weight-k vector w} from Boolean membership queries}
\label{boolean}

% If the domain of a homogeneous LTF $f_w$ is $\{-1,1\}^n$ and the weight vector $w\in\{0,1/\sqrt{k}\}^n$ has support~$k$, then $f_w$ is the Majority function on the $k$ bits in the support of $w$.

We now consider a special case of learning homogeneous LTFs. Suppose the weight vector is 
\[
w=\frac{\mathbf 1_A}{\sqrt{|A|}}
\]
for some hidden subset $A \subseteq[n]$ of size $k$, and one has access to the Boolean oracle defined in \Cref{eq:phase-oracle}. Then, $f_w$ is the Majority function on the $k$ bits $x_A$ in the support $A$. 

Belovs~\cite{belovs2015symmetricjuntas} gave a bounded-error quantum algorithm for this case that identifies the hidden support \emph{exactly} (and hence learns $f_w$) using $O(k^{1/4})$ membership queries. 
In contrast, the classical algorithm requires $\Omega(k\log(n/k))$ queries for this (in particular, $\Omega(n)$ queries if $\Omega(n)=k\leq n/2$)~\cite{bshouty1996oracles}. He also gave several algorithms for other types of $k$-juntas, which compute an OR or an ExactHalf function on the $k$ bits rather than Majority, and showed matching poly$(k)$ lower bounds for those problems. However, for the case of Majority, an exponential gap was left open between the best known lower and upper bounds.

We show here that this gap can be closed. We provide a quantum algorithm that (with bounded error probability) identifies the support of $w$ exactly using only $O(\log k)$ membership queries, giving an \emph{exponential} quantum speedup in query complexity for the natural problem of learning this special class of LTFs.
Our algorithm also works if the actual support size of $w$ is unknown and $k$ is only an upper bound on it.
For the promise $1\leq|A|\leq k$, we show that this $O(\log k)$ bound is optimal for every $n\geq k$; under the exact promise $|A|=k$, it is optimal when $n\geq2k$. This fully answers an open question from~\cite[Section 8]{belovs2015symmetricjuntas}.

\paragraph{The Algorithm.}

This logarithmic query-complexity upper bound looks quite similar to our logarithmic upper bound in~\Cref{ssec:Rn}, but it is proved quite differently. We cannot use those gradient-based techniques here, since now $f_w$ has discrete domain $\{-1,1\}^n$.
Our upper bounds are proven by constructing explicit feasible solutions to the dual adversary semidefinite program (SDP), whose optimal value is known to characterize quantum query complexity up to constant factors~\cite{HoyerLeeSpalek2007,reichardt2011span,lee2011stateconversion}: such a feasible solution can be converted into an algorithm. The dual adversary SDP for the Majority-learning problem is the following minimization program.

\paragraph{Dual adversary bound for the quantum query complexity of identifying $A$:}  \Cref{eq:junta-dual-adversary}

\[
\begin{aligned}
\underset{\{X_S\}}{\operatorname{minimize}}\quad
&\max_{A\in\mathcal A}\sum_{S\subseteq[n]}
X_S\llbracket A,A\rrbracket\\[3pt]
\textnormal{subject to}\quad
&\sum_{\substack{S\subseteq[n]\\ f_A(S)\ne f_B(S)}}
X_S\llbracket A,B\rrbracket=1
\quad(A\ne B)\\[3pt]
&X_S\succeq0
\quad(S\subseteq[n])
\end{aligned}
\]
Here, $\mathcal A$ is the promise (in our case it consists of the sets $A$ such that $|A|=k$ or $1\leq |A| \leq k$). $L_{\mathcal A}$ is the associated problem of learning the hidden support $A$. Since $f_w$ and $O^\pm_{f_w}$ only depend on $A$, we refer to them as $f_A$ and $O^\pm_A$. Instead of querying Hamming strings $x \in \{\pm1\}^n$, we query the associated subsets $S := \{i:x_i=1\} \subseteq[n]$. The SDP runs over matrices $X_S$, that for each $S\subseteq [n]$ represent query $2\mathbf{1}_{\{i \in S\}}-1$, with entries for each hidden support pair $A,B \in \mathcal A$.

Its optimal value is $\operatorname{Adv}^{\pm}(L_{\mathcal A})$,  while \emph{any} feasible solution of this SDP gives an upper bound on the associated quantum query complexity. However, computing the optimum for large $n$ and $k$ quickly becomes numerically intractable. 
Our feasible solutions  were found using the \emph{intersection-free} restriction of the dual adversary SDP. 
It consists of restricting the feasible set of the SDP to matrices of the form
\[
X_S\llbracket A,B\rrbracket
=\frac{1}{2^n}Y\llbracket |A\cap S|,|B\cap S|\rrbracket
\qquad Y\succeq0\,
\]
for the promise $|A|=k$, where $Y \in \mathbb{R}^{(k+1)\times (k+1)}$. Since it is a restriction of the adversary SDP, it can only give an upper bound:
\[
\operatorname{Adv}^{\pm}(L_{\mathcal A})
=
\operatorname{OPT}(\text{dual adversary SDP})
\leq
\operatorname{OPT}(\text{intersection-free SDP})\,.
\]
This restriction has the advantage that its feasible points are \emph{low-rank} and are parametrized by exponentially fewer variables than the original SDP. Numerical optimization on the intersection-free SDP revealed the following feasible solutions, which only work under the promise $|A|=k$ if $k$ is odd. 

Let $Y = W_k W_k^\intercal$ with
\[
W_k\llbracket a\rrbracket
=\sum_{t=0}^{k}K_t^{(k)}(a)
\bigl(P_k(t)+(-1)^tQ_k(t)\bigr),
\]
where $P_k,Q_k$ are the polynomials from \Cref{def:mod4}, and $K_t^{(k)}$ are the binary Krawtchouk polynomials.
This gives us a feasible solution with cost
$\sum_{t=0}^{k}\binom{k}{t}(P_k(t)^2+Q_k(t)^2)=O(\log k)$.

However, for the general promise $1\leq |A|\leq k$, this rank-one construction is no longer feasible. We solve this problem by generalizing the previous solutions by taking 

\[
X_S\llbracket A,B\rrbracket
=\frac{1}{2^n}\left\langle
W_m(|A\cap S|),W_\ell(|B\cap S|)
\right\rangle,
\]
for odd support sizes $|A|=2m+1$ and $|B|=2\ell+1$. For each such $m$, $W_m$ is defined as
\[
W_m(a)=\sum_{t=0}^{2m+1}K_t^{(2m+1)}(a)
\bigl(p_m(t)+(-1)^tq_m(t)\bigr),
\qquad
\]
with $p_m,q_m$ being the vector-valued polynomials from \Cref{sec:new-pol}.
These matrices then form a feasible solution of the dual adversary SDP with objective value $O(\log k)$. 

To turn the feasible adversary solutions into an algorithm with a good gate complexity rather than just a good query complexity, we follow the transducer formalism of~\cite{belovs2023taming}; a similar construction has been done in~\cite{dubus2026optimaltransducers}. We first consider the feasible solution for $|A|=k$ with $k$ odd, as it is rank-one. 
For every support $A$ of this promise, we define the \emph{catalyst} as the state of the \emph{query space} $\mathcal Q := \mathbb{C}^{2^n}$ such that
\[
\ket{w_A}:=\frac{1}{\sqrt{2^{n+1}}}
\sum_{S\subseteq[n]}W_k\llbracket |A\cap S|\rrbracket\ket S.
\]
The feasibility constraints of the dual adversary SDP ensure the existence of an input-independent unitary $U$ acting on the direct sum $\mathcal V \oplus \mathcal Q =\mathbb{C}^{2^n} \oplus \mathbb{C}^{2^n} \cong \mathbb{C}^2 \otimes \left(\mathbb{C}^2\right)^{\otimes n}$, where $\mathcal V$ is the \emph{output space}, that satisfies
\[
U\left(\ket0\oplus O_A^\pm\ket{w_A}\right)
=\ket A\oplus\ket{w_A}
\]
Hence, one query to the controlled oracle $O_A^{\pm}$ followed by the unitary $U$ maps $\ket{0} \oplus \ket{w_A}$ to 
$\ket{A} \oplus \ket{w_A}$. Note that this does not solve the majority-learning problem, as it needs access to the input-dependent catalyst $\ket{w_A}$, and one cannot simply omit the catalyst as its squared norm is by definition the objective value $O(\log k)$, leading to an unbounded error. However, one can use $K = O(\log k)$ controlled applications of this subroutine to construct an algorithm on a bigger space 
$\mathbb C^ K \otimes \mathcal V \oplus \mathcal Q$, that maps $\ket{0} \oplus \frac{1}{\sqrt{K}} \ket{w_A}$ to $\ket{A} \oplus \frac{1}{\sqrt{K}} \ket{w_A}$. The new catalyst has bounded norm, and therefore can be omitted. This gives us the final bounded-error algorithm that uses $K=O(\log k)$ controlled queries to the oracle.
 
The analysis of the non-oracle gate complexity and depth of our algorithm boils down to the gate implementation of the unitary $U$. Fortunately, one can implement it with $O(n \log n)$ one-qubit and two-qubit gates, and $O(\log n)$ depth, as follows.

\begin{figure}[H]
\centering
\begin{tikzpicture}[
    x=1cm,
    y=.78cm,
    wire/.style={line width=.45pt},
    gate/.style={draw,fill=white,rounded corners=1pt,
        minimum width=1.15cm,minimum height=.56cm,align=center,font=\small},
    block/.style={gate,minimum width=1.30cm}
]
    \def\xend{12.6}
    \node[anchor=east] at (0,0)  {$D:$};
    \node[anchor=east] at (0,-1) {$Q:$};
    \node[anchor=east] at (0,-2) {$W:\ket{0^\ell}$};
    \foreach \y in {0,-1,-2}
        \draw[wire] (.1,\y)--(\xend,\y);

    \fill (1.15,0) circle (1.6pt);
    \draw[wire] (1.15,0)--(1.15,-1);
    \node[gate] at (1.15,-1) {$H^{\otimes n}$};
    \node[block,minimum height=1.38cm] at (3.15,-1.5)
        {$C$};
    \node[gate] (xgate) at (4.85,0) {$X$};
    \fill (4.85,-2) circle (1.6pt);
    \draw[wire] (xgate.south)--(4.85,-2);
    \node[below,font=\scriptsize] at (4.85,-2) {$0$};
    \node[gate] (zgate) at (6.70,0) {$Z$};
    \fill (6.70,-2) circle (1.6pt);
    \draw[wire] (zgate.south)--(6.70,-2);
    \node[below,font=\scriptsize] at (6.70,-2) {$t_2$};
    \node[gate,minimum width=1.25cm] (ukgate) at (8.55,0)
        {$U_kZ^{k_2}$};
    \fill (8.55,-2) circle (1.6pt);
    \draw[wire] (ukgate.south)--(8.55,-2);
    \node[below,font=\scriptsize] at (8.55,-2) {$k$};
    \node[block,minimum height=1.38cm] at (10.10,-1.5)
        {$C^\dagger$};
    \fill (11.75,0) circle (1.6pt);
    \draw[wire] (11.75,0)--(11.75,-1);
    \node[gate] at (11.75,-1) {$H^{\otimes n}$};

    \node[anchor=west] at (\xend,-2) {$\ket{0^\ell}$};
\end{tikzpicture}
\caption*{\Cref{fig:majority-unitary-circuit}: Circuit implementation of $U$ for the rank-one $|A|=k$ construction.}
\end{figure}

The implementation first uses a Hadamard transform controlled by the qubit $D$ that selects between the output space $\mathcal V$ and the query space $\mathcal Q$. It is followed by the application of
$C$, which is the subroutine developed by Zi, Yuan, and Nie~\cite{zi2026exacthamming}. It coherently computes the Hamming weight $t$ of the $n$-qubit register $Q$, and stores its binary representation on the ancillary register $W = W_{\ell} \otimes \cdots \otimes W_1$ with $\ell = \lceil\log_2(n+1)\rceil$ qubits. The three gates in the middle are controlled by the different 
Hamming weights (respectively $0$, every $t \equiv 2,3 \pmod 4$ and $k$). Here, $t_2$ and $k_2$ denote the second least significant bit of $t$ and $k$, while the third gate $U_k$ is defined in~\Cref{sec:majority-unitary-implementation}.

In~\Cref{sec:U-unknown}, we derive a very similar implementation of $U$ for the whole $1 \leq |A| \leq k$ promise. The final algorithm uses $O(\log k)$ oracle calls, $O(n\log n\log k)$ non-oracle gates, and has non-oracle depth $O(\log n\log k)$, see \Cref{thm:majority-unknown-time-efficient-implementation}.

For known even support size $|A|=2m$, one can add one known dummy variable and run the rank-one
circuit for the known odd size $2m+1$, using the reduction in~\Cref{sec:majority-unknown-size}.

% While our algorithms only work on odd support sizes, one can always solve for even support $|A| = 2m$; we can add one known dummy variable and run our algorithm for odd size $2m+1$, as done in~\Cref{sec:majority-unknown-size}.
% which were later translated into an analytical construction.
% Note that for such efficient quantum query algorithms coming from the dual adversary semidefinite program (SDP), we typically cannot give a good upper bound on their gate complexity.

% The feasible solutions in question were found using an \emph{intersection-free} restriction of the dual adversary SDP. This restriction presents the advantage that its feasible points are \emph{low-rank} and are parametrized by exponentially fewer variables than the original SDP. Numerical optimization on the intersection-free SDP led to the discovery of feasible solutions with objective value $O(\log k)$, which were later translated into an analytical construction. 

\paragraph{Matching Lower Bound.}
For completeness, we prove a matching lower bound by giving a feasible solution for the \emph{primal} adversary SDP:
\[
\begin{aligned}
\underset{\Gamma}{\operatorname{maximize}}\quad
&\norm{\Gamma}\\[3pt]
\textnormal{subject to}\quad
&\norm{\Gamma\circ\Delta_S}\leq1
\quad(S\subseteq[n])\\[3pt]
&\Gamma=\Gamma^{\mathsf T},
\qquad \Gamma\llbracket A,A\rrbracket=0
\quad(A\in\mathcal A),
\end{aligned}
\]
where $\Delta_S\llbracket A,B\rrbracket=\mathbf 1_{\{f_A(S)\ne f_B(S)\}}$. We restrict the promise to a known superset $U$ of size $2m$, and subsets $A,B \subseteq U$ of size $m$. We define the primal feasible solution
\[
\Gamma\llbracket A,B\rrbracket
=
\begin{cases}
\displaystyle
\frac{1}{|A\setminus B|\binom{m}{|A\setminus B|}^{2}}&A\ne B,\\[6pt]
0&A=B.
\end{cases}
\]
This matrix has operator norm $\norm{\Gamma}=\Theta(1+\log m)$ and
$\norm{\Gamma\circ\Delta_S}=O(1)$ for every $S$.  Taking
$m=\lfloor k/2\rfloor$ gives the $\Omega(\log k)$ lower bound for $1\leq|A|\leq k$. 

For the exact promise $|A|=k$, we take $m=\min\{k,n-k\}$. We fix disjoint known sets $C,U$ with $|C|=k-m$ and $|U|=2m$, and consider hidden supports $C\cup A$ for each $A \subseteq U$ of size $m$. Both of these families give the bound $\Omega(\log k)$ for every $n\geq k$ if $1\leq|A|\leq k$, and for $n\geq2k$ if~$|A|=k$.

\subsection{Learning from quantum examples using the Hermite transform}\label{ssec:introHermitealgo}

 Following the standard learning models (see e.g.~\cite{arunachalam2017guest}), let $\mathcal C$ be a concept class over $\mathcal X$, let $f\in\mathcal C$
be the unknown target concept, and let $\mathcal D$ be a distribution over $\mathcal X$. The classical example oracle $\EX(f,\mathcal D)$ returns, on each invocation, an independent labeled example $(x,f(x))$ with $x\sim\mathcal D$. A quantum analogue is a unitary quantum example oracle $\QEX(f,\mathcal D)$ satisfying
\begin{equation}
    \QEX(f,\mathcal D)\ket{0}
    :=
    \sum_{x\in\mathcal X}
    \sqrt{\mathcal D(x)}\,
    \ket{x}\ket{\bar f(x)}.\footnotemark
\end{equation}
\footnotetext{
For a $\{\pm1\}$-valued function $f$, we use
$\bar f(x):=(1-f(x))/2\in\{0,1\}$, so that
$f(x)=(-1)^{\bar f(x)}$.
}$\QEX(f, \mathcal D)$ is a state-preparation oracle: only its action on the designated all-zero input state is specified, and its behavior on other states doesn't matter.
We additionally allow coherent access to the inverse oracle $\QEX(f,\mathcal D)^\dagger$, with each
application of either $\QEX(f,\mathcal D)$ or $\QEX(f,\mathcal D)^\dagger$ counted as one query, similar to~\cite{kothari2023mean, van2023quantum, Salmon2024Provable}. 
Salmon, Strelchuk, and Gur~\cite{Salmon2024Provable} showed that this unitary-example-generation model allows for more quantum advantage than the setting where one is given only copies of the example state~\cite{arunachalam_optimal_2017} (allowing inverses is necessary for this advantage~\cite{isogai2026advantagesamplecomplexityquantum}).

For continuous instance spaces such as $\mathcal X=\R^n$, the input must in practice be represented with finite precision. We therefore write $\QEX_M(f,\mathcal D)$ for the quantum example oracle associated with a finite discretization of the domain, indexed by a precision parameter $M$. If $\mathcal X_M$ denotes the resulting finite grid and $\mathcal D_M$ the induced discrete distribution, then
\begin{equation}
    \QEX_M(f,\mathcal D)\ket{0}
    =
    \sum_{x\in\mathcal X_M}
    \sqrt{\mathcal D_M(x)}\,
    \ket{x}\ket{\bar f(x)}.
\end{equation}
The particular choice of discretization and the resulting approximation error will be specified when we discuss the implementation of the algorithm. It does not affect the resulting quantum advantage, as discussed in~\Cref{rem: continuous-vs-discretized}. The classical lower bound, and hence the quartic separation in $n$, is preserved under our discretization.

Generally, $\EX(f,\mathcal D)$ reveals one randomly drawn labeled example per query, whereas $\QEX_M(f,\mathcal D)$ and its inverse permit coherent processing and uncomputation of the labeled distribution. In both models, the queried points are determined by $\mathcal D$, rather than chosen arbitrarily by the learner, so this constitutes a weaker access model than coherent membership queries.%
\footnote{When the sampling distribution $\mathcal D$ is known and its coherent superposition can be prepared, a quantum membership query can generate a quantum example by preparing $\sum_x\sqrt{\mathcal D(x)}\ket{x}\ket{0}$ and applying $O_f$; the converse simulation is not possible in general with only a small number of quantum examples~\cite{arunachalam2017guest}.}

Suppose we are given access to a (discretized) quantum example oracle for a homogeneous LTF $f_w$ under the standard Gaussian distribution $\gamma_n$:
\[    \QEX(f_w,\gamma_n)\ket{0}
    =
    \int_{\R^n}
    \sqrt{\gamma_n(x)}\,
    \ket{x}\ket{\bar{f}_w(x)}\,dx.
\]
with $\bar{f}_w(x):=\frac{1-f_w(x)}{2}\in\{0,1\}$, where $f_w(x)=(-1)^{\bar{f}_w(x)}$. We give a quantum algorithm that learns $f_w$ to error at most $\eps$ under $\gamma_n$ using $O(n^{1/4}/\sqrt{\eps})$ applications of $\QEX(f_w,\gamma_n)$ and its inverse, together with $\widetilde O\!\left(n^{5/4}/\sqrt{\eps}+n/\eps^2\right)$ additional gates. 
As mentioned before, this learning task, classically, would require $\Omega(n/\eps)$ examples under Gaussian distribution~\cite[Theorem~13]{BalcanLong2013}, so we obtain a quartic improvement in terms of~$n$ when going from classical examples to a quantum example oracle. Unlike our membership-query result from~\Cref{ssec:Rn}, this algorithm currently applies only to homogeneous LTFs; extending the algorithm to general LTFs with nonzero threshold remains an open problem.

\paragraph{The Algorithm.} Our algorithm builds on the efficient \emph{quantum Hermite transform} (QHT) of Jain, Iyer, Somma, Bao, and Jordan~\cite{JainIyerSommaBaoJordan2025}. This answers an open question from~\cite{JainIyerSommaBaoJordan2025} and, to the best of our knowledge, gives the first application of the Hermite transform to a natural machine-learning problem. \cite{JainIyerSommaBaoJordan2025} already had an application to a Gaussian variant of Goldreich--Levin, but our use of the QHT is conceptually different from theirs. There, the QHT is primarily used to identify a small number of individual large Hermite coefficients. For halfspaces, however, the relevant information is spread across many coefficients of the same degree. The key observation that we use is that these coefficients are highly structured: by the \emph{Hermite ridge identity} (\Cref{lem:ridge-identity}), they collectively encode the unknown direction~$w$. Our algorithm uses the QHT to isolate this structured high-degree component and then recovers $w$ from it using quantum state tomography.  Since
$\Theta(n/\eps^2)$ copies suffice to reconstruct an $n$-dimensional pure state to error $O(\eps)$, we choose a degree of $K=\Theta(n/\eps^2)$. The probability of obtaining this informative high-degree band of $[K, 2K]$ is $\Theta(K^{-1/2})$,  so amplitude amplification can access it using only $O(K^{1/4})$ example queries. This gives the query complexity $O(n^{1/4}/\sqrt{\eps})$.

To illustrate the main ideas, we first describe the quantum algorithm in an idealized continuous setting. We defer the details of its finite-dimensional implementation, approximation errors, and gate complexity to~\Cref{sec:example}. Throughout this overview, we introduce registers as needed and omit ancillas that have been returned to fixed states for notational simplicity.

We begin by applying the quantum example oracle to prepare the state
\begin{equation*}
    \QEX(f_w,\gamma_n)\ket{0}
    =
    \int_{\R^n} \sqrt{\gamma_n(x)}
    \ket{x}\ket{\bar{f}_w(x)}dx.
\end{equation*}
Applying a Hadamard gate to the label register (i.e., the last qubit) gives
\begin{align*}
    (I\otimes H)\QEX(f_w,\gamma_n)\ket{0}
    =
    \frac{1}{\sqrt2}
    \int_{\R^n}
    \sqrt{\gamma_n(x)}
    \ket{x}
    \bigl(
        \ket{0}
        +
        (-1)^{\bar{f}_w(x)}\ket{1}
    \bigr)\, dx
   % \nonumber\\ &
   =
    \frac{1}{\sqrt2}
    \left(
        \ket{G}\ket{0}
        +
        \ket{\Phi_w}\ket{1}
    \right),
\end{align*}
where
\[
    \ket{G}
    :=
    \int_{\R^n}\sqrt{\gamma_n(x)}\ket{x}\, dx,
    \qquad
    \ket{\Phi_w}
    :=
    \int_{\R^n}
    \sqrt{\gamma_n(x)}
    f_w(x)\ket{x} \, dx.
\]
Thus, the desired phase-encoded function state appears coherently in the branch in which the label register is \(\ket1\). Measuring and postselecting this register would succeed with probability exactly \(1/2\) (and we know when we succeeded). However, for the algorithm below we keep the register coherent, since this allows the subsequent degree projection to be amplified unitarily.

We next apply the inverse quantum Hermite transform to the input register. Using the convention that \(\QHT^{-1}\) maps the Gaussian position
representation to Hermite coefficients, we obtain
\begin{equation}
    (\QHT^{-1}\otimes I)
    \frac{1}{\sqrt2}
    \left(
        \ket{G}\ket0
        +
        \ket{\Phi_w}\ket1
    \right)
    =
    \frac{1}{\sqrt2}
    \left(
        \ket{F_0}\ket0
        +
        \ket{F_w}\ket1
    \right),
    \label{eq:flagged-hermite-state}
\end{equation}
where the first branch \(\ket{F_0}:=\QHT^{-1}\ket G\) will be irrelevant, while
\begin{equation*}
    \ket{F_w}
    :=
    \sum_{\alpha\in\mathbb N^n}
    \widehat f_w(\alpha)\ket{\alpha},
    \qquad
    \widehat f_w(\alpha)
    :=
    \E_{x\sim\gamma_n}
    \bigl[
        f_w(x)h_\alpha(x)
    \bigr].
\end{equation*}

To understand the structure of $\ket{F_w}$, write the one-dimensional Hermite expansion of the sign function as $\sgn(t)=\sum_{k\geq0}a_k h_k(t)$. Since $f_w(x)=\sgn(\langle w,x\rangle)$, the Hermite ridge identity of~\Cref{lem:ridge-identity} gives
$
    h_k(\langle w,x\rangle)
    =
    \sum_{|\alpha|=k}
    \sqrt{\frac{k!}{\alpha!}}\,
    w^\alpha h_\alpha(x)
$, where $w^\alpha=\prod_{i=1}^n w_i^{\alpha_i}$ is an $n$-variate monomial. 
Thus,
$
    \widehat f_w(\alpha)
    =
    a_{|\alpha|}
    \sqrt{\frac{|\alpha|!}{\alpha!}}\,
    w^\alpha
$. Defining quantum states
\begin{equation*}
    \ket{C_k(w)}
    :=
    \sum_{|\alpha|=k}
    \sqrt{\frac{k!}{\alpha!}}\,
    w^\alpha\ket{\alpha},
\end{equation*}
we obtain the orthogonal decomposition
\begin{equation}
    \ket{F_w}
    =
    \sum_{k\geq0}
    a_k\ket{C_k(w)}.
    \label{eq:condensate-decomp}
\end{equation}
The key observation is that $\ket{C_k(w)}$ provides a compressed representation of $k$ identical copies of the hidden direction state $\ket{w}:=\sum_{i=1}^n w_i\ket{i}$. To make this connection precise, we introduce the following $w$-independent linear map:
\begin{equation}
    V_k\ket{\alpha}
    :=
    \sqrt{\frac{\alpha!}{k!}}
    \sum_{\substack{i=(i_1,\ldots,i_k)\in[n]^k:\\
    |\{r:i_r=j\}|=\alpha_j \ \forall j\in[n]}}
    \ket{i_1}\otimes\cdots\otimes\ket{i_k},
    \label{eq:overview-Vk}
\end{equation}
for every $\alpha\in\mathbb N^n$ with $|\alpha|=k$, where $\alpha!:=\prod_{j=1}^n\alpha_j!$. The map $V_k$ is an isometry from the degree-$k$ subspace into the symmetric subspace of $(\mathbb C^n)^{\otimes k}$. To see this, observe that the number of strings $i=(i_1,\ldots,i_k)\in[n]^k$ with occupation vector $\alpha$ is exactly the multinomial coefficient $k!/\alpha!$. Since the corresponding computational-basis states are mutually orthogonal, the prefactor $\sqrt{\alpha!/k!}$ ensures that $V_k\ket{\alpha}$ is normalized. Moreover, distinct occupation vectors correspond to disjoint sets of strings, implying that $\langle V_k\alpha\mid V_k\beta\rangle=0$ whenever $\alpha\neq\beta$. Thus, $V_k$ preserves inner products, as required.

Crucially, $V_k$ is independent of the unknown direction $w$ and therefore requires no additional queries to $\QEX$. Furthermore, it admits an efficient quantum circuit implementation using the symmetrization procedure of~\cite{liu2025lowdepthquantumsymmetrization}. Moreover, a direct calculation gives
\begin{equation*}
    V_k\ket{C_k(w)}
    =
    \ket w^{\otimes k}.
\end{equation*}
Hence, once a sufficiently high-degree component of $\ket{F_w}$ has been isolated, learning $w$ reduces to pure-state tomography. By the
tomography guarantee of~\cite{odonnell2015efficientquantumtomography, Haah_2023}, $K=\Theta(n/\eps^2)$ copies suffice to recover $w$ to $O(\eps)$ Euclidean error, up to the global sign, with constant success probability.

It remains to obtain such a high-degree component efficiently. For each $k\geq0$, let $\Pi_k$ denote the projector onto total Hermite degree $k$. From~\Cref{eq:condensate-decomp}, $\Pi_k\ket{F_w}=a_k\ket{C_k(w)}$ so a direct projection onto degree $k$ succeeds with probability $|a_k|^2$. One could therefore use amplitude amplification to obtain the degree-\(k\) component using $O\!\left(1/|a_k|\right)$ applications of $\QEX(f_w,\gamma_n)$ and its inverse. However, this direct projection onto a single degree $k=\Theta(K)$ is not optimal. Recall that $\{a_k\}$ is the set of Hermite coefficients of the $\sgn$ function. From~\Cref{lem:sign-band}, for an odd $k$ one has $|a_k|^2=\Theta(k^{-3/2})$, so amplifying a single useful degree would require $\Theta(K^{3/4})$ applications of $\QEX(f_w,\gamma_n)$ and $\QEX(f_w,\gamma_n)^\dagger$. The key observation is that tomography does not require a particular degree: every $k\geq K$ already contains at least $K$ copies of $\ket w$. We therefore enlarge the good subspace to the entire degree band
\begin{equation*}
    B_K
    :=
    \{K,K+1,\ldots,2K\},
    \qquad
    \Pi_{B_K}
    :=
    \sum_{k=K}^{2K}\Pi_k.
\end{equation*}
The weight of this band is $p_K:=\|\Pi_{B_K}\ket{F_w}\|^2=\sum_{k=K}^{2K}|a_k|^2=\Theta(K^{-1/2})$, where only the odd degrees contribute. The joint label-and-band event has probability $p_K/2$. Consequently, amplitude amplification prepares a normalized state
\begin{equation*}
    \ket{F_w^{B_K}}
    :=
    \frac{1}{\sqrt{p_K}}
    \sum_{k=K}^{2K}
    a_k\ket{C_k(w)}
\end{equation*}
using only $O(K^{1/4})=O(n^{1/4}/\sqrt{\eps})$ applications of $\QEX(f_w,\gamma_n)$ and its inverse. Finally, for each degree-$k$ component, we append the multi-index state $\ket{\alpha}$ with $2K-k$ copies of an additional blank symbol $\ket{\perp}$, orthogonal to $\ket{1},\ldots,\ket{n}$, so that all components have length $2K$. We then apply a single symmetrization circuit $V_{2K}$ and measure which output registers contain the blank symbol, without disturbing the nonblank labels. Since $k\ge K$, at least $K$ nonblank registers remain, containing copies of $\ket{w}$ that can be used for tomography.

To connect this idealized algorithm with its finite-dimensional realization, we must carefully choose the discretization parameters and control the resulting approximation errors. The quantum Hermite transform of~\cite{JainIyerSommaBaoJordan2025} operates on a finite grid and approximates the continuous transform only up to a prescribed Hermite degree cutoff $D$, provided that the grid size $M$ is sufficiently large. Since our algorithm only requires the degree band $B_K=[K,2K]$, we choose $D=2K+1$ to ensure that all relevant Hermite degrees are included.

An additional subtlety arises in our setting because the sign function is discontinuous along the decision boundary
$\{x\in\R^n:\langle w,x\rangle=0\}$. In particular, discretization cells intersecting this hyperplane may contain points with opposite labels, introducing an additional approximation error that is not captured by the $\QHT$ approximation alone. We address this issue by bounding the Gaussian measure of the cells intersecting the decision boundary and combining this bound with the approximation guarantees of the finite $\QHT$. As shown in~\Cref{sec:discrete-qht}, choosing $M$ sufficiently large, but still polynomial in $n$ and $1/\eps$, ensures that the total approximation error is small relative to the amplitude $\Theta(K^{-1/4})$ of the degree band $B_K$. Consequently, the probability of obtaining the desired component remains $\Theta(K^{-1/2})$, preserving the $O(K^{1/4})$ query complexity of amplitude amplification.

The remaining operations also admit efficient quantum-circuit implementations. In particular, the coherent projection onto the high-degree band $B_K$ can be realized by reversibly computing the total Hermite degree $|\alpha|$ and checking whether it belongs to $B_K$. This allows us to implement the reflection about the desired label-and-degree subspace required for amplitude amplification, as detailed in~\Cref{sec:high-degree-projection}. Moreover, the particle-number padding and the isometry $V_{2K}$ can be implemented efficiently using reversible occupation-number encoding and quantum symmetrization techniques~\cite{liu2025lowdepthquantumsymmetrization}, as described in~\Cref{sec:isometry-map}. Together with pure-state tomography from~\cite{Haah_2023}, these constructions complete the finite-dimensional implementation of our algorithm.

Combining the above ingredients yields the following theorem.
\begin{theorem}[Learning homogeneous LTFs with Example Oracle;~\Cref{thm:quantum-example-learning}]\label{thm: quantum-example-learning_intro}
Assume $n\geq2$; the case $n=1$ is trivial and can be handled separately. Let $f_w(x):=\sgn(\langle w,x\rangle)$ for some $w\in\mathbb \sphere^{n-1}$, and let $\eps\in(0,1/2)$. Set $K:=\left\lceil c_1 n/\eps^2\right\rceil$, $D:=2K+1$, and $M:=2^{\left\lceil\log_2\!\left(c_2n^{13/4}K^{49/16}\right)\right\rceil}$, where $c_1,c_2>0$ are sufficiently large universal constants. Given coherent access to $\QEX_M(f_w,\gamma_n)$ and its inverse, there exists a quantum algorithm that outputs a unit vector
$\widetilde w\in\mathbb R^n$ such that
$
    \|\widetilde w-w\|\leq\eps
$
with probability at least $2/3$. The algorithm makes
\[
    O\!\left(
        \frac{n^{1/4}}{\sqrt{\eps}}
    \right)
\]
queries to $\QEX_M(f_w,\gamma_n)$ and its inverse, using $\widetilde O\!\left(\frac{n^{5/4}}{\sqrt{\eps}}+\frac{n}{\eps^2}\right)$ additional gates, where the $\widetilde O$ notation hides factors polylogarithmic in $n$ and $1/\eps$, together with $\operatorname{poly}(n,1/\eps)$ classical post-processing time.
\end{theorem}

\paragraph{Matching Lower Bound.} While~\Cref{thrm:1.1} achieves an exponential quantum improvement under the stronger membership-query model, the above algorithm achieved a quartic improvement in the dimension-dependence over classical passive learning. This naturally raises the question  whether the query complexity can be improved further, potentially recovering the exponential advantage achievable with membership queries. We answer this question negatively by establishing a matching lower bound, showing that our algorithm's query complexity is optimal in both $n$ and $\eps$. Consequently, the restriction to quantum-example oracle access fundamentally limits the achievable quantum advantage to a polynomial rather than an exponential separation from classical learning.
\begin{theorem}[Quantum example-oracle lower bound;~\Cref{thm:gaussian-assouad-lower-bound}]\label{thm: gaussian-assouad-lower-bound_intro}
For every $n\geq2$, $0<\eps\leq1/(8\pi)$, any quantum learner that for every $w\in \sphere^{n-1}$, outputs a classical measurable classifier $h$ satisfying
$d_{\gamma_n}(h,f_w)\leq\eps$ with probability at least $2/3$ requires
\begin{equation}
    \Omega\!\left(\frac{n^{1/4}}{\sqrt\eps}\right)
\end{equation}
queries to $\QEX(f_w, \gamma_n)$ (and/or its inverse). This holds even for adaptive algorithms that output improper hypotheses (i.e., an $h$ that is not itself a homogeneous LTF).
\end{theorem}
The $\Omega(1/\sqrt{\eps})$ dependence follows readily from a simple two-point method argument by constructing two halfspaces whose disagreement probability under the Gaussian distribution is $\Theta(\eps)$, so that any sufficiently accurate learner must distinguish between them. Since the corresponding quantum example oracles can be chosen to differ in operator norm by at most $O(\sqrt{\eps})$, a standard quantum hybrid argument shows that $\Omega(1/\sqrt{\eps})$ queries are necessary. However, obtaining dependence on the dimension $n$ requires a more delicate analysis. To this end, we develop a novel connection between Assouad's hypercube method~\cite{yu1997festschrift} and Gaussian surface area in studying the complexity of geometric concept classes~\cite{KOS08}. The key idea is to construct a family of concepts indexed by hidden bit strings, such that accurate learning requires recovering most of the bits, while neighboring concepts obtained by perturbing a single coordinate remain difficult to distinguish. For Gaussian halfspaces, we bound the disagreement between neighboring concepts using Gaussian surface area and noise sensitivity under perturbation~\cite{KOS08}. Combining this geometric bound with a hybrid argument and Assouad’s lemma~\cite{yu1997festschrift} yields the matching lower bound $\Omega(n^{1/4}/\sqrt{\eps})$. This lower-bound technique may also be of independent interest. We leave open whether a similar technique can be extended to other classes with controlled Gaussian surface area, such as those studied in~\cite{KOS08}. The lower bound of~\Cref{thm: gaussian-assouad-lower-bound_intro} also holds for the finite-dimensional oracle $\QEX_M(f_w,\gamma_n)$ with the grid size $M$ chosen in~\Cref{thm: quantum-example-learning_intro}. See~\Cref{rem: lower_bound_discrete} for further details.
\section{Preliminaries}
\subsection{Basic notation}
Throughout, $n$ denotes the ambient dimension and $[n]:=\{1,\ldots,n\}$. For an integer $D>0$, we write $[D]_0:=\{0,\ldots,D-1\}$. Let $w=(w_1,\ldots,w_n)$ be a real vector and let $\alpha=(\alpha_1,\ldots,\alpha_n)$ be a vector of nonnegative integers. Whenever these expressions are well defined, we write
\[
w^\alpha:=\prod_{i=1}^n w_i^{\alpha_i},
\qquad
\alpha!:=\prod_{i=1}^n\alpha_i!,
\qquad
|\alpha|:=\sum_{i=1}^n\alpha_i.
\]
Throughout this paper, we adopt the convention $\mathbb N=\mathbb N_0:=\{0,1,2,\ldots\}$. For multi-indices $\alpha\in\mathbb N^n$, we use the convention $0^0:=1$.

We will also write $x$ for an input vector in the domain under consideration and $w\in\sphere^{n-1}$ for the unknown unit vector. The notation
$\langle\cdot,\cdot\rangle$ denotes the standard inner product. For a vector $v$, $\norm{v}$ denotes its Euclidean norm; for a matrix or linear operator $M$, $\norm{M}$ denotes its operator norm. We use the convention
\[
\sgn(t):=
\begin{cases}
1,&t\geq0,\\
-1,&t<0
\end{cases}.
\]
For $w\in\sphere^{n-1}$ and $\theta\in\mathbb R$, we write $f_{w,\theta}(x):=\sgn(\langle w,x\rangle-\theta)$, and abbreviate $f_w:=f_{w,0}$.

For $x\in\{-1,1\}^n$, we write $|x|:=|\{i\in[n]:x_i=1\}|$ for its Hamming weight, and $\bar x$ the vector in $\{-1,1\}^n$ such that for each $i\in[n]$, $\bar x_i = -x_i$. For $A\subseteq[n]$, $|A|$ denotes its cardinality and $\mathbf 1_A\in\{0,1\}^n$
its associated indicator vector; equivalently, $|A|=|\mathbf 1_A|$. We write
$\binom{[n]}{k}$ for the family of all subsets of $[n]$ of size $k$. Finally,
$x_A\in\{-1,1\}^A$ is the restriction of $x$ to the coordinates in~$A$.

\subsection{Adversary bound}

In the special case of \Cref{boolean}, the problem becomes an instance of a discrete \emph{function evaluation} problem. The query complexity of this class of problems is well studied~\cite{hamoudi2025brief}, and many techniques are known for proving lower and upper bounds. One of the main lower bound techniques is the \emph{adversary bound}, first introduced in~\cite{ambainis2002quantum}.
It was generalized in~\cite{HoyerLeeSpalek2007} and was later shown to characterize the complexity of bounded-error quantum query complexity up to constant factors~\cite{reichardt2011span,lee2011stateconversion}.

The adversary bound can be phrased as a \emph{semidefinite program} (SDP), which admits a dual formulation that gives upper bounds on query complexity.
A feasible solution to the dual program can be 
turned into a quantum query algorithm whose query complexity matches the objective value of the SDP, up to constant factors, using
the span-program/transducer framework~\cite{reichardt2011span,belovs2023taming}.
Note that this construction does not, in general, give any guarantees on the space or time complexity of the resulting algorithm.

For Majority-junta learning, we use the standard dual adversary bound given in \Cref{eq:junta-dual-adversary}. 
Its optimal value characterizes bounded-error quantum query complexity up to constant factors, for both Boolean and non-Boolean alphabets \cite{lee2011stateconversion}.
By weak duality, any feasible solution of the primal adversary SDP gives a lower bound on quantum query complexity, whereas a feasible solution of the dual adversary SDP gives an upper bound.

We first define the junta learning problem, for which the adversary bound yields a tight query-complexity characterization. For each $k\geq1$, fix a known non-constant symmetric Boolean function $h_k:\{-1,1\}^k\to\{-1,1\}$. We also write $h_k(t)$ for its value on inputs of Hamming weight $t\in\{0,\ldots,k\}$.

\begin{definition}
\label{def:junta}
Let $1\leq k\leq n$ and $A\subseteq[n]$ with $|A|=k$. The symmetric
$k$-junta with hidden set $A$ is the function
$f_A:\{-1,1\}^n\to\{-1,1\}$ defined by
$f_A(x)=h_k(x_A)$. Equivalently,
\[
f_A(S)=h_k(|A\cap S|),
\]
where $S:=\{i\in[n]:x_i=1\}$.
\end{definition}

\begin{definition}
Let $\mathcal A$ be any promise of nonempty subsets of $[n]$. For every
$A\in\mathcal A$, let $f_A$ be the symmetric $|A|$-junta from the preceding
definition. The symmetric-junta learning problem with promise $\mathcal A$ is,
given oracle access to $f_A$ for an unknown $A\in\mathcal A$, to output $A$
with error probability at most $1/3$.
\end{definition}

The inputs to this learning task are functions. Thus, it is the partial
function-evaluation problem
\[
L_{\mathcal A}:\{f_A:A\in\mathcal A\}\to\mathcal A,
\qquad L_{\mathcal A}(f_A)=A.
\]
For every $S\subseteq[n]$, define the matrix $\Delta_S$ whose rows and columns are indexed by
$\mathcal A$, by
$\Delta_S\llbracket A,B\rrbracket:=\mathbf 1_{\{f_A(S)\ne f_B(S)\}}$.
We define both the primal and dual adversary bounds for this learning problem.
All matrices below are indexed by 
$\mathcal A$.

\paragraph{Primal adversary bound for $L_{\mathcal A}$:}

\begin{equation}
\label{eq:junta-primal-adversary}
\begin{aligned}
\underset{\Gamma}{\operatorname{maximize}}\quad
&\norm{\Gamma}\\[3pt]
\textnormal{subject to}\quad
&\norm{\Gamma\circ\Delta_S}\leq1
\quad(S\subseteq[n])\\[3pt]
&\Gamma=\Gamma^{\mathsf T},
\qquad \Gamma\llbracket A,A\rrbracket=0
\quad(A\in\mathcal A)
\end{aligned}
\end{equation}

\paragraph{Dual adversary bound for $L_{\mathcal A}$:}

\begin{equation}
\label{eq:junta-dual-adversary}
\begin{aligned}
\underset{\{X_S\}}{\operatorname{minimize}}\quad
&\max_{A\in\mathcal A}\sum_{S\subseteq[n]}
X_S\llbracket A,A\rrbracket\\[3pt]
\textnormal{subject to}\quad
&\sum_{\substack{S\subseteq[n]\\ f_A(S)\ne f_B(S)}}
X_S\llbracket A,B\rrbracket=1
\quad(A\ne B)\\[3pt]
&X_S\succeq0
\quad(S\subseteq[n])
\end{aligned}
\end{equation}

Here $\circ$ denotes the entrywise product. The two programs have the same
optimal value, which characterizes bounded-error quantum query complexity up
to constant factors \cite{reichardt2011span,lee2011stateconversion}.

\subsection{Hermite polynomials}\label{ssec:hermitedefs}
\paragraph{Hermite functions and polynomials.} Let $\gamma_n$ denote the standard Gaussian measure on $\R^n$, with density $\gamma_n(x)=(2\pi)^{-n/2}e^{-\norm{x}^2/2}$. We
work in $L^2(\gamma_n)$, the vector space of all functions $f:\R^n\to\R$ such that $\E[f^2]<\infty$ under the Gaussian measure~$\gamma_n$. This is an inner product space under $\inner{f}{g}=\E_{x\sim\gamma_n}[f(x)g(x)]$. It has a complete orthonormal basis given by the \emph{Hermite polynomials}. For $n=1$, the probabilists' Hermite polynomials are defined by 
\[
\operatorname{He}_k(t):=(-1)^k e^{t^2/2}\frac{d^k}{dt^k}e^{-t^2/2}, 
\mbox{ with normalized versions }h_k(t):=\frac{\operatorname{He}_k(t)}{\sqrt{k!}}.
\]
For general $n$, a basis for $L^2(\gamma_n)$ is formed by products of these polynomials, one for each coordinate: 
\[
h_\alpha(x)=\prod_{i=1}^n h_{\alpha_i}(x_i)\mbox{ for }\alpha\in\N^n.
\]
Thus, every $f\in L^2(\gamma_n)$ has a unique Hermite expansion
\[
    f(x)=\sum_{\alpha\in\N^n}\hat f(\alpha)h_\alpha(x),
    \qquad
    \hat f(\alpha)=\int f(x)h_\alpha(x)\,d\gamma_n(x).
\]
The above convention is related to the \emph{physicists' Hermite polynomials}, defined by $H_k(t):=(-1)^k e^{t^2}\frac{d^k}{dt^k}e^{-t^2}$. The two conventions differ only by a rescaling of the argument and normalization. In particular, $\operatorname{He}_k(t)=2^{-k/2}H_k\!\left(\frac{t}{\sqrt{2}}\right)$. Hence, our normalized probabilists' Hermite polynomial is $h_k(t)=\frac{1}{\sqrt{2^k k!}}\,H_k\!\left(\frac{t}{\sqrt{2}}\right)$. The rescaling reflects the different Gaussian weights underlying the two conventions: $e^{-t^2/2}$ for the probabilists' convention and $e^{-t^2}$ for the physicists' convention. Throughout this work, we primarily use the probabilists' convention and convert between the two when needed.

\paragraph{Gaussian halfspaces.} We first show that the disagreement probability of two Gaussian halfspaces can be bounded in terms of the Euclidean error between their parameters.

\begin{lemma}[Gaussian disagreement of LTFs]
\label{lem:gaussian-disagreement}
For all $w_1,w_2\in\sphere^{n-1}$ and $\theta_1,\theta_2\in\R$,
\begin{equation}
    \Pr_{x\sim\gamma_n}
    \left[
        f_{w_1,\theta_1}(x)
        \neq
        f_{w_2,\theta_2}(x)
    \right]
    \leq
    \frac12\|w_1-w_2\|
    +
    \frac{|\theta_1-\theta_2|}{\sqrt{2\pi}}.
\end{equation}
Moreover, in the homogeneous case $\theta_1=\theta_2=0$, the disagreement probability admits the exact expression
\begin{equation}
    \Pr_{x\sim\gamma_n}
    \left[
        f_{w_1}(x)
        \neq
        f_{w_2}(x)
    \right]
    =
    \frac{1}{\pi}
    \arccos\!\bigl(\langle w_1,w_2\rangle\bigr).
    \label{eq:homogeneous-gaussian-disagreement}
\end{equation}
\end{lemma}
\begin{proof}
By the triangle inequality for disagreement probability,
\[
\Pr\!\left[f_{w_1,\theta_1}(x)\neq f_{w_2,\theta_2}(x)\right]
\leq
\Pr\!\left[f_{w_1,\theta_1}(x)\neq f_{w_2,\theta_1}(x)\right]
+
\Pr\!\left[f_{w_2,\theta_1}(x)\neq f_{w_2,\theta_2}(x)\right].
\]
The second term is easy to control because the two classifiers have the same normal vector \(w_2\) and differ only in their thresholds. Since \(w_2\) is a unit vector and \(x\sim\gamma_n\), the scalar random variable \(\langle w_2,x\rangle\) is distributed as \(\mathcal N(0,1)\). Therefore, the two classifiers disagree exactly when \(\langle w_2,x\rangle\) lies between \(\theta_1\) and \(\theta_2\). Hence, this disagreement probability is the integral of the standard Gaussian density over an interval of length \(|\theta_1-\theta_2|\). Since the Gaussian density is everywhere at most \(1/\sqrt{2\pi}\), this contribution is at most \(|\theta_1-\theta_2|/\sqrt{2\pi}\).

It remains to bound the first term, where the two classifiers have the same threshold $\theta_1$ but different normals. Let
$Z_i=\langle w_i,x\rangle$ and $\rho:=\langle w_1,w_2\rangle=\cos\alpha$, where $\alpha\in[0,\pi]$ is the angle between $w_1$ and $w_2$.
If $\rho=-1$, then $w_2=-w_1$ and $\alpha=\pi$, so trivially
\[
    \Pr\!\left[
        f_{w_1,\theta_1}(x)\neq f_{w_2,\theta_1}(x)
    \right]
    \leq 1
    =
    \frac{\alpha}{\pi}
    =
    \frac12\|w_1-w_2\|.
\]
Hence, it remains to consider $\rho>-1$. In this case, setting $a:=\sqrt{\frac{1+\rho}{2}}>0$, $b:=\sqrt{\frac{1-\rho}{2}}$, we can write the correlated Gaussian pair as
\[
    Z_1=aA+bB,
    \qquad
    Z_2=aA-bB,
\]
where $A,B\sim\mathcal N(0,1)$ are independent. For a common threshold $\theta_1$, disagreement occurs exactly when $(Z_1-\theta_1)(Z_2-\theta_1)<0$, or equivalently $|aA-\theta_1|<b|B|$. Conditioning on \(B\), this means that \(A\) must lie in an interval of fixed length centered at \(\theta_1/a\). Because the standard Gaussian density is symmetric and decreases away from the origin, an interval of fixed length has the largest Gaussian mass when centered at \(0\). Thus, the disagreement probability is maximized at \(\theta_1=0\), giving
\[
\Pr[f_{w_1,\theta_1}(x)\neq f_{w_2,\theta_1}(x)]
\leq
\Pr[f_{w_1,0}(x)\neq f_{w_2,0}(x)]
=
\frac{\alpha}{\pi}.
\]
Finally, since \(\|w_1-w_2\|=2\sin(\alpha/2)\) and \(\sin(\alpha/2)\geq \alpha/\pi\), we obtain \(\alpha/\pi\leq \|w_1-w_2\|/2\). Combining the two contributions gives the claimed bound.

For the homogeneous case, let
$\alpha:=\arccos(\langle w_1,w_2\rangle)\in[0,\pi]$ be the angle between
$w_1$ and $w_2$. By rotational invariance of the Gaussian distribution,
the disagreement probability depends only on the projection of $x$ onto
$\operatorname{span}\{w_1,w_2\}$. In this two-dimensional plane, the two
homogeneous halfspaces disagree on two opposite wedges, each of angle
$\alpha$. Since the direction of a standard Gaussian vector is uniform
on the circle, their total Gaussian measure is
\[
    \frac{2\alpha}{2\pi}
    =
    \frac{\alpha}{\pi}
    =
    \frac{1}{\pi}
    \arccos\!\bigl(\langle w_1,w_2\rangle\bigr),
\]
which proves~\Cref{eq:homogeneous-gaussian-disagreement}.
\end{proof}

We now recall two basic facts about the Hermite expansion of homogeneous halfspaces.
\begin{lemma}[Hermite ridge identity]\label{lem:ridge-identity}
For every \(k\geq0\),
\begin{equation*}
  h_k(\langle w,x\rangle)
  =\sum_{\abs\alpha=k}
    \sqrt{\frac{k!}{\alpha!}}\,w^\alpha h_\alpha(x).
\end{equation*}
\end{lemma}
\begin{proof}
    Denote $z := \inner{w}{x}$, then we can write 
    \[
        x = zw + x_{\perp}, \quad \text{such that } \inner{w}{x_{\perp}} = 0 
    \]
    Since $\norm{w} = 1$, we have $\norm{x}^2 = z^2 + \norm{x_{\perp}}^2$ and
    \begin{equation}
        e^{\norm{x}^2/2}  \frac{d^k}{dz^k}e^{-\norm{x}^2/2} = e^{\norm{x}^2/2} e^{-\norm{x_{\perp}}^2/2} \frac{d^k}{dz^k} e^{-z^2/2} = e^{z^2/2} \frac{d^k}{dz^k} e^{-z^2/2}. 
    \end{equation}
    Thus, by definition, we have 
    \begin{equation}
        h_k(z) = \frac{(-1)^k}{\sqrt{k!}} e^{\norm{x}^2/2}  \frac{d^k}{dz^k}e^{-\norm{x}^2/2}.\label{temp1}
    \end{equation}
    On the other hand, the operator
    \begin{align}
        \frac{d^k}{dz^k} &= \left( \sum_{i=1}^n w_i \cdot \frac{d}{dx_i} \right)^k = \sum_{|\alpha| = k} \frac{k!}{\alpha!} w^\alpha \prod_{i=1}^n \frac{d^{\alpha_i}}{d x_i^{\alpha_i}} \label{temp2}
    \end{align}
    Substituting~\Cref{temp2} into~\Cref{temp1}, we obtain
    \begin{align}
         h_k(z) &= \frac{(-1)^k}{\sqrt{k!}} e^{\sum_i x^2_i/2} \sum_{|\alpha| = k} \frac{k!}{\alpha!} w^\alpha \prod_{i=1}^n \frac{d^{\alpha_i}}{d x_i^{\alpha_i}} e^{\sum_i -x^2_i/2} \\
         &= \sum_{|\alpha| = k} \sqrt{\frac{k!}{\alpha!}} w^\alpha \prod_{i=1}^n \frac{(-1)^{\alpha_i}}{\sqrt{\alpha_i!}} e^{x^2_i/2} \frac{d^{\alpha_i}}{d x_i^{\alpha_i}} e^{-x^2_i/2}\\
         &=\sum_{|\alpha| = k} \sqrt{\frac{k!}{\alpha!}} w^\alpha h_{\alpha}(x).
    \end{align}
\end{proof}

\begin{lemma}[Hermite mass of sign function]
\label{lem:sign-band}
Consider the function $g(z):=\sgn(z)$ with Hermite decomposition $g(z)=\sum_r a_r h_r(z)$. For every $r\geq0$, we have
\begin{enumerate}[label=(\roman*)]
    \item $a_{2r}=0$.
    \item $|a_{2r+1}|^2=\frac{2}{\pi}\frac{\binom{2r}{r}}{4^r(2r+1)}=\Theta((r+1)^{-3/2})$.
\end{enumerate}
\end{lemma}
\begin{proof}
By~\cite[Theorem~2.11]{davis2024hermite}, the positive-half-line indicator has the Hermite expansion
\begin{equation*}
    \mathbf{1}_{\{z>0\}}
    =
    \frac12
    +
    \frac{1}{\sqrt{2\pi}}
    \sum_{r=0}^{\infty}
    \frac{(-1)^r}{(2r+1)2^r r!}
    \operatorname{He}_{2r+1}(z).
\end{equation*}
Since $g(z)=2\mathbf{1}_{\{z>0\}}-1$ almost everywhere, uniqueness of the Hermite expansion gives, for every $r\geq0$,
\begin{equation*}
    a_{2r}=0,
    \qquad
    a_{2r+1}
    =
    (-1)^r\sqrt{\frac{2}{\pi}}\,
    \frac{\sqrt{(2r+1)!}}{(2r+1)2^r r!}.
\end{equation*}
Consequently,
\begin{equation*}
    |a_{2r+1}|^2
    =
    \frac{2}{\pi}
    \frac{(2r+1)!}{(2r+1)^2 4^r(r!)^2}
    =
    \frac{2}{\pi}
    \frac{\binom{2r}{r}}{4^r(2r+1)}.
\end{equation*}
Finally, Stirling's formula yields $\binom{2r}{r}=\Theta(4^r/\sqrt r)$ for $r\geq1$. Thus $|a_{2r+1}|^2=\Theta((r+1)^{-3/2})$ for all $r\geq0$.
\end{proof}
\begin{remark} Under our convention $\sgn(0)=1$. On the other hand, every odd
Hermite polynomial satisfies $h_{2r+1}(0)=0$, and therefore the Hermite series in \Cref{lem:sign-band} evaluates to $0$ at
$z=0$. Thus, the Hermite expansion of $\sgn$ should not be interpreted as a pointwise identity at the origin, but rather under the
continuous Gaussian measure, where $\gamma_1(\{0\})=0$. When we later pass to a finite grid, however, the decision boundary may carry positive discrete mass. This contribution is handled separately in
\Cref{prop:halfspace-boundary-error}.
\end{remark}

\section{Learning LTFs over $\mathbb{R}^n$ with a Membership Oracle}\label{sec:quantum-membership}

In this section, we present a quantum algorithm for learning a general linear threshold function $f_{w,\theta}(x)=\sgn(\inner{w}{x}-\theta)$ from real-domain membership queries. Given target accuracy $\eps>0$, our goal is to yield an $O(\eps)$-accurate output under the standard Gaussian distribution. By~\Cref{lem:gaussian-disagreement}, it suffices to recover a unit vector $\widetilde w$ and threshold $\widetilde\theta$ such that $\norm{\widetilde w-w}=O(\eps)$ and
$|{\widetilde\theta-\theta}|=O(\eps)$. Classically, even the homogeneous case requires $\Omega(n\log(1/\eps))$ membership queries in the worst case~\cite{kulkarni1993active} (see also~\Cref{thm:classical-main}),
while the known upper bound is $\widetilde O(n\log(1/\eps))$~\cite{hopkins2020point}. We show that coherent quantum membership access reduces the query complexity exponentially,
to $O(\log(n/\eps))$, while using $\widetilde O(n)$ additional gates.

\subsection{Quantum algorithm}

We now develop each step outlined in~\Cref{ssec:Rn} in detail. We begin in~\Cref{sec:estimate-offset} by constructing a well-conditioned, finite-precision reference direction $p$ and estimating the boundary location $h_p(0)$. Next, in~\Cref{sec:height-oracle}, we construct a coherent finite-precision oracle for evaluating the height function $h_p$ over the grid used by our quantum algorithm. Building on this oracle, we implement Jordan's quantum gradient-estimation procedure in~\Cref{sec:gradient-estimation}. Finally, in~\Cref{sec:membership-putting-together}, we combine these components to recover $w$ and $\theta$ and establish the overall query and gate complexities.

\subsubsection{Estimating offset}\label{sec:estimate-offset}
In this section, we describe how to choose a good direction $p$ and how to estimate the offset $h_p(0)$ cheaply using binary search.

We first show that one can choose a finite-precision reference direction $p$ such that $\abs{\inner{w}{p}}$ is bounded above and below by universal constants. In particular, we construct $p$ so that each coordinate can be stored using only $O(\log n)$ bits. This guarantees that the associated height function is well-conditioned and can be implemented using reversible fixed-point arithmetic.
\begin{lemma}[Finite-precision reference direction]
\label{lem:finite-reference}
For any $n\geq 2$, there is a randomized classical preprocessing procedure that outputs a reference vector $p\in\mathbb{R}^n$, represented using $O(\log n)$ bits
per coordinate, such that for every fixed $w\in\sphere^{n-1}$,
\begin{equation*}
    \Pr\left[
        \frac1{10}
        \leq
        |\langle w,p\rangle|
        \leq
        3
    \right]
    \geq
    \frac45.
\end{equation*}
\end{lemma}

\begin{proof}
Let $p^\star\sim\mathcal N(0,I_n)$. For every fixed $w\in\sphere^{n-1}$, we have $\langle w,p^\star\rangle \sim\mathcal N(0,1)$. Hence,
\begin{equation}
    \Pr\left[
        \frac18
        \leq \abs{\inner{w}{p^\star}}
        \leq \frac{29}{10}
    \right]
    >
    0.89.
    \label{eq:ideal-overlap-event}
\end{equation}

We truncate each coordinate of $p^{\star}$ within $[-R, R]$ where $R:=\sqrt{2\log(200n)}$. As $p^{\star}_i \sim \calN(0,1)$, union bound yields
\begin{align*}
    \Pr\left[
        \max_{i\in[n]}|p_i^\star|>R
    \right]
    &\leq \sum_{i\in [n]} \Pr[|p^\star_i| > R] \leq
    2n e^{-R^2/2}
    \leq
    \frac1{100}.
\end{align*}
Thus, except with probability at most $1/100$, truncating each coordinate to the interval $[-R,R]$ does not modify $p^\star$ at all.

We now choose the number of bits $b_0:=\left\lceil\log_2(100\sqrt n)\right\rceil$ used to represent $p^{\star}_i$. Define $\Delta:=2^{-b_0}$, round each truncated coordinate to the nearest multiple of $\Delta$, and denote the resulting vector by $p$. On the event that no truncation occurs, we have
\begin{align}
\norm{p-p^\star}&\leq\frac{\sqrt n\,\Delta}{2}.
\end{align}
By the choice of $b_0$, we have $\frac{1}{200\sqrt{n}} < \Delta \leq \frac{1}{100\sqrt{n}}$, then the error $\norm{p-p^\star} \leq 1/200$. Hence, because $\norm{w}=1$,
\begin{equation*}
    \left|
        \langle w,p\rangle
        -
        \langle w,p^\star\rangle
    \right|
    \leq
    \norm{p-p^\star}
    \leq
    \frac1{200}.
\end{equation*}
Consequently, whenever~\Cref{eq:ideal-overlap-event} holds and no truncation occurs, $|\langle w,p\rangle|\geq\frac18-\frac1{200}>\frac1{10}$, and 
$|\langle w,p\rangle|\leq\frac{29}{10}+\frac1{200}<3$. Hence,
\begin{align}
    \Pr\left[
        \frac1{10}
        \leq|\langle w,p\rangle|
        \leq3
    \right]
    &\geq
    0.89-\frac1{100}>
    \frac45.
\end{align}
Finally, since $R=O(\sqrt{\log n})$ and $\Delta^{-1}=O(\sqrt n)$, a coordinate of $p$ requires only $O\!\left(\log\frac{R}{\Delta}\right)=O(\log n + \log \log(n)) = O(\log n)$ bits to represent.
\end{proof}

Now, we are ready to handle the offset estimation. We begin with a simple observation: if the threshold $\theta$ is sufficiently far from the origin, then the corresponding halfspace is already close to a constant function under the standard Gaussian distribution. In this regime, there is therefore no need to recover either the normal vector $w$ or the threshold $\theta$ accurately. Moreover, the appropriate constant label can be determined using a single membership query at the origin. We formalize this observation in the following lemma.
\begin{lemma}
\label{lem:large-threshold-constant}
For $n\geq 2$, let $f_{w,\theta}(x):=\sgn(\langle w,x\rangle-\theta)$,
where $w\in\sphere^{n-1}$ and $\theta\in\R$, and let $x\sim\gamma_n$.
For every $0<\eps<1$, if $|\theta|
    \geq
    \sqrt{2\log\frac1\eps}$, then a single membership query suffices to output a constant function $c:\R^n\to\{-1,+1\}$ satisfying
\begin{equation}
    \Pr_{x\sim\gamma_n}
    \left[
        f_{w,\theta}(x)\neq c(x)
    \right]
    \leq
    \eps.
\end{equation}
\end{lemma}

\begin{proof}
Since $\|w\|=1$ and $x\sim\gamma_n$, the random variable
$z:=\langle w,x\rangle$ is distributed as $\calN(0,1)$.
Because the assumed lower bound on $|\theta|$ implies $\theta\neq0$,
a membership query at the origin returns
\[
    f_{w,\theta}(0)
    =
    \sgn(-\theta)
    =
    -\sgn(\theta).
\]
We therefore output the constant function
$c(x)\equiv f_{w,\theta}(0)$. Its disagreement probability is
\begin{equation}
    \Pr[f_{w,\theta}(x)\neq c(x)]
    =
    \begin{cases}
        \Pr[z\geq|\theta|], & \theta>0,\\
        \Pr[z<-|\theta|], & \theta<0,
    \end{cases}
    \leq
    e^{-\theta^2/2},
\end{equation}
where the last inequality follows from the standard Gaussian tail bound.
If $|\theta|\geq\sqrt{2\log(1/\eps)}$, then
$e^{-\theta^2/2}\leq\eps$, which proves the claim.
\end{proof}
\Cref{lem:large-threshold-constant} shows that sufficiently large thresholds are easy once this regime can be certified. We therefore focus on the non-trivial regime of $|\theta|\leq \sqrt{2\log(1/\eps)}$. In the full learning algorithm, an endpoint test will first determine which regime applies. Accordingly, we state the following boundary-search theorem directly under the above bounded-threshold assumption.
\begin{theorem}[Estimating the boundary offset]
\label{thm:estimate-boundary-offset}
For $n\geq 2$, let $0<\eps\leq 1/2$ and $0<\zeta<1$. Condition on the event
$\frac1{10}\leq|\langle w,p\rangle|\leq3$.
Let $T_\eps:=\sqrt{2\log(1/\eps)}$ and
$L_\eps:=2^{\lceil\log_2(20T_\eps)\rceil}$, and suppose that
$|h_p(0)|\leq L_\eps$.
Then there is a deterministic algorithm that uses
$O(\log(L_\eps/\zeta))$ membership queries and returns
$s=\operatorname{sign}(\langle w,p\rangle)$ and a value
$\widetilde h_p(0)$ satisfying
\begin{equation}
    \bigl|
        \widetilde h_p(0)-h_p(0)
    \bigr|
    \leq
    \zeta.
\end{equation}
Moreover, $\widetilde h_p(0)$ can be represented exactly in
signed fixed-point form using $O(\log(L_\eps/\zeta))$ bits.
\end{theorem}

\begin{proof}
Recall that $h_p(0)=\theta/\langle w,p\rangle$.
By assumption, $h_p(0)\in[-L_\eps,L_\eps]$.
Restricting the target function to the line $\R p$, we have
\begin{align}
    f_{w,\theta}(rp)
    &=
    \sgn\!\left(r\langle w,p\rangle-\theta\right)
    \nonumber\\
    &=
    \sgn\!\left(
        \langle w,p\rangle
        \bigl(r-h_p(0)\bigr)
    \right).
    \label{eq:offset-line-threshold}
\end{align}
To determine the orientation, we query at $2L_\eps p$.
Since $|h_p(0)|\leq L_\eps$, we have
$2L_\eps-h_p(0)\geq L_\eps>0$.
Applying~\Cref{eq:offset-line-threshold} with $r=2L_\eps$
therefore gives
\[
    f_{w,\theta}(2L_\eps p)
    =
    \sgn\!\left(
        \langle w,p\rangle
        \bigl(2L_\eps-h_p(0)\bigr)
    \right)
    =
    \sgn(\langle w,p\rangle).
\]
Thus, a single membership query determines
$s:=\operatorname{sign}(\langle w,p\rangle)$.
This remains valid when $h_p(0)$ is either endpoint of the
search interval, since $2L_\eps$ lies strictly to the right
of every point in $[-L_\eps,L_\eps]$.

Once $s$ is known, every subsequent membership query on the
line $\R p$ becomes a comparison query for the unknown value
$h_p(0)$.
Indeed, for any $r\neq h_p(0)$,
\begin{equation}
    s\,f_{w,\theta}(rp)
    =
    \sgn(r-h_p(0)).
\end{equation}
Equivalently, define
\begin{equation}
    c(r)
    :=
    \frac{
        1-sf_{w,\theta}(rp)
    }{2}.
\end{equation}
Then, away from the boundary,
\begin{equation}
    c(r)
    =
    \begin{cases}
        1, & r<h_p(0),\\
        0, & r>h_p(0).
    \end{cases}
    \label{eq:offset-comparison}
\end{equation}
At $r=h_p(0)$, our convention $\sgn(0)=+1$ gives
$c(r)=(1-s)/2$.
This boundary case is handled by retaining closed intervals.

We now perform binary search over the initial interval
$I_0=[-L_\eps,L_\eps]$.
Suppose after $j-1$ rounds we have a closed interval
$I_{j-1}=[\ell_{j-1},u_{j-1}]$ containing $h_p(0)$, and let
$r_j:=(\ell_{j-1}+u_{j-1})/2$ be its midpoint.
We query $f_{w,\theta}(r_jp)$ and compute $c(r_j)$.
If $c(r_j)=1$, we retain the right half
$I_j=[r_j,u_{j-1}]$; if $c(r_j)=0$, we retain the left half
$I_j=[\ell_{j-1},r_j]$.

If $r_j\neq h_p(0)$, then~\Cref{eq:offset-comparison}
ensures that the retained half contains $h_p(0)$.
If $r_j=h_p(0)$, either comparison outcome retains a closed
child interval containing $r_j$ as an endpoint.
Hence the invariant $h_p(0)\in I_j$ holds after every round,
including when a query lands exactly on the boundary.

After $q$ rounds, the interval length is
$|I_q|=2L_\eps/2^q$.
Let $\widetilde h_p(0)$ be the midpoint of $I_q$. Then
\begin{equation}
    \bigl|
        \widetilde h_p(0)-h_p(0)
    \bigr|
    \leq
    \frac{L_\eps}{2^q}.
\end{equation}
Since $0<\eps\leq1/2$, we have
$T_\eps\geq\sqrt{2\log2}>1$, and hence $L_\eps\geq2$.
Together with $0<\zeta<1$, this implies
$L_\eps/\zeta>2$.
Therefore it suffices to choose
\begin{equation}
    q
    :=
    \left\lceil
        \log_2\frac{L_\eps}{\zeta}
    \right\rceil.
\end{equation}
The algorithm uses one membership query to determine $s$
and one query in each round. Its total membership-query complexity is therefore
\begin{equation}
    1+q
    =
    O\!\left(
        \log\frac{L_\eps}{\zeta}
    \right).
\end{equation}

We finally bound the exact bit complexity of the returned
estimate. By definition,
\begin{equation}
    L_\eps
    =
    2^{r_\eps},
    \qquad
    r_\eps
    :=
    \left\lceil
        \log_2(20T_\eps)
    \right\rceil.
\end{equation}
After $q$ rounds, $I_q$ is one of the $2^q$ equal
subintervals of $[-L_\eps,L_\eps]$.
Hence, for some $J\in\{0,\ldots,2^q-1\}$, its midpoint is
\begin{equation}
    \widetilde h_p(0)
    =
    -L_\eps
    +
    (2J+1)\frac{L_\eps}{2^q}
    =
    -2^{r_\eps}
    +
    (2J+1)2^{r_\eps-q}.
    \label{eq:offset-dyadic-form}
\end{equation}
Since $\zeta<1$ and $L_\eps=2^{r_\eps}$, we have
$q=r_\eps+\lceil\log_2(1/\zeta)\rceil>r_\eps$.
Therefore~\Cref{eq:offset-dyadic-form} requires at most
$r_\eps+O(1)$ bits for the signed integer part and
$q-r_\eps$ fractional bits.
The total number of bits is consequently
\begin{equation}
    r_\eps+O(1)+(q-r_\eps)
    =
    q+O(1)
    =
    O\!\left(
        \log\frac{L_\eps}{\zeta}
    \right),
\end{equation}
as claimed.
\end{proof}
\subsubsection{Implementing a height-function query}\label{sec:height-oracle}
In this section, we give a coherent implementation of $h_p$ using membership queries to $O_{f_{w, \theta}}$. We condition throughout on the nontrivial case of the boundary-offset estimation procedure, so that we have already obtained the sign $s:=\sgn(\langle w,p\rangle)$ and an estimate $\widetilde h_p(0)$ of $h_p(0)$. Before going into the details, we first introduce a representation of qubit strings suitable for reversible fixed-point arithmetic and for the finite-grid implementation of quantum gradient estimation.

For $b\in\mathbb N$, we label the $b$-qubit computational basis state $\ket{j}$, where $j\in\{0,\ldots,2^b-1\}$, by $g_b(j)$, where
\begin{equation*}
    g_b(j)
    :=
    \frac{j}{2^b}
    -\frac12
    +2^{-b-1}.
\end{equation*}
Equivalently, we define the one-dimensional grid $G_b:=\left\{\frac{j}{2^b}-\frac12+2^{-b-1}:j=0,\ldots,2^b-1\right\}\subset\left(-\frac12,\frac12\right)$. By~\cite[Claim 19]{Gily_n_2019}, the map $j\mapsto g_b(j)$ is a bijection, so we use $\ket{j}$ and $\ket{g_b(j)}$ interchangeably.

\begin{theorem}[Coherent implementation of height function]
\label{lem:coherent-height}
For $n\geq 2$, let $0< \eps \leq 1/2$. Condition on the event $\frac1{10}\leq|\langle w,p\rangle|\leq3$, assume that the procedure of~\Cref{thm:estimate-boundary-offset} run with accuracy $\zeta=\eps/100$, returns $s= \sgn(\inner{w}{p})$ and an estimate $\widetilde h_p(0)$ satisfying $|\widetilde h_p(0)-h_p(0)|\leq\frac{\eps}{100}$. Let $\rho=2^{-\lceil\log_2(40\sqrt n)\rceil}$ and define $H_p(x):=h_p(\rho x)$. Then, for every accuracy parameter $0<\eta<1/8$, there is a reversible
circuit that maps $\ket{x}\ket{0}\longmapsto\ket{x}\ket{\widetilde H_p(x)}$
such that
\[
    |\widetilde H_p(x)-H_p(x)|
    \leq
    \eta,
\]
for every $x\in G_b^n$. The circuit uses $O\!\left(\log\frac1{\eta}\right)$ queries to $O_{f_{w, \theta}}$ and
\[
    O\!\left(
        n\left(b+\log n+\log\frac1\eps+\log\frac1\eta\right)^2
        \log\frac1\eta
    \right)
\]
additional gates.
\end{theorem}

\begin{proof}
Recall that every $x\in G_b^n$ satisfies $|x_i|\leq 1/2$ for all $i\in[n]$. Hence $\norm{x}\leq\frac{\sqrt n}{2}$. Moreover, by the definition of $\rho$, we have $\rho\leq\frac{1}{40\sqrt n}$ and $\rho\norm{x}\leq 1/80$.  Consequently,
\begin{align*}
    |H_p(x)-h_p(0)|
    &\leq
    \frac{\rho\norm{x}}
         {|\langle w,p\rangle|}
    \leq
    \frac{1/80}{1/10}
    =
    \frac18.
\end{align*}
Thus, using $|\widetilde h_p(0)-h_p(0)|\leq\frac{\eps}{100}$, we obtain
\begin{align}
    |H_p(x)-\widetilde h_p(0)|
    &\leq
    |H_p(x)-h_p(0)|
    +
    |h_p(0)-\widetilde h_p(0)|
    \nonumber\\
    &\leq
    \frac18+\frac{\eps}{100}
    < \frac18+\frac{1}{16} = 
    \frac3{16}
    <
    \frac14,
\end{align}
where we used $\eps/100 < 1/16$ for every $\eps\in(0,1/2]$. Therefore, choosing the interval
$
    I
    :=
    \left[
        \widetilde h_p(0)-\frac14,
        \widetilde h_p(0)+\frac14
    \right]
$
suffices to contain $H_p(x)$ at every point in the entire hypergrid. We next construct a coherent comparison oracle for the unknown height.
Recall that, for every $t\in\R$,
\begin{align*}
    \langle w,\rho x+tp\rangle-\theta
    &=
    \langle w,p\rangle
    \bigl(t-H_p(x)\bigr).
\end{align*}
Since $s:=\sgn(\langle w,p\rangle)$ is already known from~\Cref{thm:estimate-boundary-offset}, define $c_p(x,t)
    :=
    \frac{1-sf_{w, \theta}(\rho x+tp)}{2}$. Then, away from the boundary,
\begin{equation}
    c_p(x,t)
    =
    \begin{cases}
        1, & t<H_p(x),\\
        0, & t>H_p(x).
    \end{cases}
    \label{eq:comparison-direction}
\end{equation}
Thus, $c_p(x,t)$ determines on which side of the candidate $t$ the boundary height $H_p(x)$ lies. Using it as a comparison oracle in a
binary search over $I$ yields an approximation of $H_p(x)$.

We now show that this comparison can be implemented coherently and reversibly using one membership query. Encode the $\{\pm1\}$-valued
label by
\[
    \bar f_{w, \theta}(y)
    :=
    \frac{1-f_{w, \theta}(y)}{2}
    \in\{0,1\},
\]
and similarly define the known classical bit
\[
    \sigma
    :=
    \frac{1-s}{2}
    \in\{0,1\}.
\]
Then
\begin{equation*}
    c_p(x,t)
    =
    \frac{1-sf_{w, \theta}(\rho x+tp)}{2}
    =
    \bar f_{w, \theta}(\rho x+tp)\oplus\sigma.
\end{equation*}

Let $B$ be a comparison qubit and $Q$ a query-address register. Define the reversible address-preparation circuit
\begin{equation}
    A_{\mathrm{addr}}:
    \ket{x,t}\ket0_Q
    \longmapsto
    \ket{x,t}\ket{\rho x+tp}_Q.
    \label{eq:address-preparation}
\end{equation}
We use the standard coherent membership oracle in the form
\begin{equation*}
    O_{f_{w, \theta}}:
    \ket{y}\ket b
    \longmapsto
    \ket{y}
    \ket{b\oplus\bar f_{w, \theta}(y)}.
\end{equation*}
Starting from an arbitrary basis state $\ket{x,t}\ket b_B\ket0_Q$, we obtain
\begin{align*}
    \ket{x,t}\ket b_B\ket0_Q
    &\xmapsto{A_{\mathrm{addr}}}
    \ket{x,t}\ket b_B\ket{\rho x+tp}_Q
    \nonumber\\
    &\xmapsto{O_{f_{w, \theta}}}
    \ket{x,t}
    \ket{b\oplus\bar f_{w, \theta}(\rho x+tp)}_B
    \ket{\rho x+tp}_Q
    \nonumber\\
    &\xmapsto{X^\sigma}
    \ket{x,t}
    \ket{b\oplus c_p(x,t)}_B
    \ket{\rho x+tp}_Q
    \nonumber\\
    &\xmapsto{A_{\mathrm{addr}}^\dagger}
    \ket{x,t}
    \ket{b\oplus c_p(x,t)}_B
    \ket0_Q.
\end{align*}
Hence one call to $O_{f_{w, \theta}}$ implements the reversible comparison
\begin{equation*}
    U_{\mathrm{cmp}}:
    \ket{x,t}\ket b
    \longmapsto
    \ket{x,t}
    \ket{b\oplus c_p(x,t)}.
\end{equation*}
In particular, $U_{\mathrm{cmp}}$ is its own inverse.

We now implement the binary search reversibly. Let \(C:=C_1\cdots C_q\) be a \(q\)-qubit register initialized to \(\ket{0^q}\). For a fixed basis state $\ket{x}$, after $r$ rounds the first $r$ qubits store
\[
    \ket{
        c_p(x,t_1),
        \ldots,
        c_p(x,t_r)
    },
\]
while the remaining $q-r$ qubits remain in $\ket0$. We use
\begin{equation*}
    K_r(x)
    :=
    \sum_{j=1}^r
    c_p(x,t_j)2^{r-j},
    \qquad
    K_0:=0,
\end{equation*}
as notation for the integer encoded by the first $r$ comparison bits. After \(r-1\) rounds, we have
\begin{equation*}
    H_p(x)\in
    I_{r-1}(x)
    :=
    \left[
        \widetilde h_p(0)-\frac14
        +
        K_{r-1}(x)2^{-r},
        \;
        \widetilde h_p(0)-\frac14
        +
        \bigl(K_{r-1}(x)+1\bigr)2^{-r}
    \right].
\end{equation*}
The midpoint of this interval is
\begin{equation*}
    t_r
    :=
    \widetilde h_p(0)-\frac14+
    \left(
        K_{r-1}(x)+\frac12
    \right)2^{-r}.
\end{equation*}
In particular $t_1=\widetilde h_p(0)$,  which is the midpoint of the initial interval $I_0=
    \left[
        \widetilde h_p(0)-\frac14,
        \widetilde h_p(0)+\frac14
    \right]$. At the $r$-th round, we first compute $t_r$ reversibly into a temporary register from the previously stored bits
$C_1,\ldots,C_{r-1}$. We then apply $U_{\mathrm{cmp}}$ to the fresh qubit $C_r$, obtaining $C_r=c_p(x,t_r)$. Finally, we reverse the arithmetic used to compute $t_r$, returning the temporary midpoint register to $\ket0$. Thus, the only information
retained after the $r$-th round is the new comparison bit $C_r$.

Equivalently, the integer represented by the first $r$ comparison bits satisfies $K_r=2K_{r-1}+C_r$, and the interval consistent with these outcomes is
\begin{equation*}
    I_r(x)
    :=
    \left[
        \widetilde h_p(0)-\frac14
        +
        K_r2^{-r-1},
        \;
        \widetilde h_p(0)-\frac14
        +
        (K_r+1)2^{-r-1}
    \right].
\end{equation*}
If $t_r\neq H_p(x)$, the claim $H_p(x)\in I_r(x)$ follows directly from~\Cref{eq:comparison-direction}. And if
$t_r=H_p(x)$, we have $f_{w, \theta}(\rho x+t_rp)=+1$ under our convention $\sgn(0)=+1$. If $s=+1$, then $c_p(x,t_r)=0$ and the left closed child interval is retained; if $s=-1$, then $c_p(x,t_r)=1$ and the right closed child interval is retained. In either case, $t_r=H_p(x)$ remains an endpoint of the retained interval. Then $H_p(x)\in I_r(x)$ holds for every round. Hence, after $q$ rounds, we reversibly compute the midpoint of $I_q(x)$,
\begin{equation*}
    \widetilde H_p(x)
    :=
    \widetilde h_p(0)-\frac14+
    \left(
        K_q+\frac12
    \right)2^{-q-1},
\end{equation*}
into a designated output register.

Reversing the $q$ bisection rounds cleans all work registers and yields
\begin{equation*}
    \ket{x}\ket0
    \longmapsto
    \ket{x}\ket{\widetilde H_p(x)}.
\end{equation*}

Since $|I_q(x)|=2^{-q-1}$ and $\widetilde H_p(x)$ is the midpoint of $I_q(x)$, we have
\begin{equation}
    \left|
        \widetilde H_p(x)-H_p(x)
    \right|
    \leq
    \frac12|I_q(x)|
    =
    2^{-q-2}.
\end{equation}
Therefore, choosing $q
    =
    \Theta\!\left(
        \log\frac1\eta
    \right)$ guarantees $\left|
        \widetilde H_p(x)-H_p(x)
    \right|
    \leq
    \eta$. Each forward round uses one membership query, and its inverse uses one additional membership query. Hence, the forward and reverse bisection
procedures together use
\[
    2q
    =
    O\!\left(
        \log\frac1\eta
    \right)
\]
membership queries.

It remains to account for the gate complexity.
By~\Cref{thm:estimate-boundary-offset},
the estimate $\widetilde h_p(0)$ is a number that can
be represented exactly using $O\!\left(\log\frac1\eps\right)$ bits, for $\zeta = \eps/100$. A coordinate of $x\in G_b^n$ requires $b+1$ fractional bits,
and multiplication by $\rho=2^{-O(\log n)}$ only shifts the binary point by
$O(\log n)$ positions. By~\Cref{lem:finite-reference}, each coordinate $p_i$ can be represented using $O(\log n)$ bits. Moreover, throughout the binary search,
\[
    t_r
    =
    \widetilde h_p(0)-\frac14
    +
    \left(
        K_{r-1}(x)+\frac12
    \right)2^{-r},
\]
so after at most $q$ rounds the correction contributes
only $O(q)$ additional fractional bits.
Since $q=O(\log(1/\eta))$, all intermediate quantities can be
represented exactly using
\begin{equation}
    \ell
    =
    O\!\left(
        b+\log n
        +\log\frac1\eps
        +\log\frac1\eta
    \right)
\end{equation}
bits.

To implement $A_{\mathrm{addr}}$ in~\Cref{eq:address-preparation},
for every $i\in[n]$ we reversibly compute
$\rho x_i+t_rp_i$.
Since $\rho$ is a power of two, multiplication by $\rho$
is only a binary shift.
Computing $t_rp_i$ requires one fixed-point multiplication,
followed by a constant number of additions or subtractions.
Using standard reversible arithmetic, multiplication of
$\ell$-bit numbers costs $O(\ell^2)$ elementary
Toffoli/CNOT gates, whereas addition and subtraction cost
$O(\ell)$ gates.
Hence, computing all $n$ coordinates of the query address
costs $O(n\ell^2)$ non-oracle gates, and
$A_{\mathrm{addr}}^\dagger$ has the same asymptotic cost.

Computing $t_r$ from the already stored comparison bits
requires only $O(\ell)$ additional gates and is therefore a
lower-order contribution.
Consequently, one coherent comparison has non-oracle gate
complexity $O(n\ell^2)$.
Since the forward and reverse procedures contain
$O(\log(1/\eta))$ comparison rounds in total, the gate
complexity is
\begin{equation}
    O\!\left(
        n\ell^2\log\frac1\eta
    \right)
    =
    O\!\left(
        n
        \left(
            b+\log n
            +\log\frac1\eps
            +\log\frac1\eta
        \right)^2
        \log\frac1\eta
    \right),
\end{equation}
as claimed.
\end{proof}

\subsubsection{Gradient estimation}
\label{sec:gradient-estimation}

We now combine the coherent height-evaluation procedure
of~\Cref{lem:coherent-height} with the grid-based formulation of
Jordan's quantum gradient-estimation algorithm in~\cite{Gily_n_2019}.
The main simplification in our setting is that the height function is
exactly affine. Consequently, there is no Taylor-remainder or
local-linearization error. Moreover, the additive offset $h_p(0)$
contributes only an input-independent global phase, so the error
analysis reduces to controlling the finite precision of the implemented
phase oracle and the finite resolution of the Fourier readout.

Particularly, Jordan's algorithm estimates the gradient components in parallel using a single call of an appropriately scaled phase oracle. It first prepares a
uniform superposition over a finite grid and applies the oracle to encode the function values as phases. Applying an inverse quantum Fourier transform to each coordinate register, followed by measurement, produces a coordinate-wise approximation of the gradient components.

To formulate the Fourier readout on our centered grid $G_b$, we define
the Fourier transform of a state $\ket{x}$, for $x\in G_b$, as
$|G_b|=2^b$,
\begin{equation*}
    \mathsf{QFT}_{G_b}:\ket{x}
    \mapsto
    \frac{1}{\sqrt{2^b}}
    \sum_{u\in G_b}
    e^{2\pi i\,2^b xu}\ket{u}.
\end{equation*}
The following lemma shows that centering the grid incurs only a linear
overhead in single-qubit phase gates.

\begin{lemma}[{\cite[Claim~19]{Gily_n_2019}}]
$\mathsf{QFT}_{G_b}$ is the same as the usual quantum Fourier transform
up to composition before and after with a tensor product of $b$
single-qubit unitaries.
\end{lemma}

\begin{proof}
Let $x=g_b(j^{(x)})\in G_b$. Then $\mathsf{QFT}_{G_b}$ acts on $\ket{x}$
as
\begin{align}
    \mathsf{QFT}_{G_b}:\ket{x}
    &\longmapsto
    \frac{1}{\sqrt{2^b}}
    \sum_{u\in G_b}
    e^{2\pi i\,2^b xu}
    \ket{u}
    \\
    &=
    \frac{1}{\sqrt{2^b}}
    \sum_{j^{(u)}\in\{0,\ldots,2^b-1\}}
    e^{
        2\pi i\,2^b
        \left(
            \frac{j^{(x)}}{2^b}-\frac12+2^{-b-1}
        \right)
        \left(
            \frac{j^{(u)}}{2^b}-\frac12+2^{-b-1}
        \right)
    }
    \ket{j^{(u)}}
    \\
    &=
    \frac{1}{\sqrt{2^b}}
    \sum_{j^{(u)}\in\{0,\ldots,2^b-1\}}
    e^{
        2\pi i
        \left(
            \frac{j^{(u)}j^{(x)}}{2^b}
            -
            (j^{(u)}+j^{(x)})
            \left(
                \frac12-2^{-b-1}
            \right)
            +
            \left(
                2^{b-2}-\frac12+2^{-b-2}
            \right)
        \right)
    }
    \ket{j^{(u)}}.
\end{align}
Using the usual quantum Fourier transform,
\begin{equation}
    \mathsf{QFT}_b:\ket{j^{(x)}}
    \longmapsto
    \frac{1}{\sqrt{2^b}}
    \sum_{j^{(u)}\in\{0,\ldots,2^b-1\}}
    e^{2\pi i\,2^{-b}j^{(u)}j^{(x)}}
    \ket{j^{(u)}},
\end{equation}
and the phase unitary
\begin{equation}
    U:\ket{j^{(u)}}
    \longmapsto
    e^{
        2\pi i
        \left(
            -j^{(u)}
            \left(
                \frac12-2^{-b-1}
            \right)
            +
            \frac{
                2^{b-2}-\frac12+2^{-b-2}
            }{2}
        \right)
    }
    \ket{j^{(u)}},
\end{equation}
we obtain
$\mathsf{QFT}_{G_b}=U\cdot\mathsf{QFT}_b\cdot U$.
Writing $j^{(u)}$ in binary shows that $U$ is a tensor product of $b$
single-qubit phase gates.
\end{proof}

We next give the end-to-end complexity of estimating the gradient of
$h_p$ as follows.

\begin{theorem}[Gradient estimation]
\label{lem:gradient-estimation}
For $n\geq 2$, assume the conditions of~\Cref{lem:coherent-height}
with learning-accuracy parameter $0<\eps\leq1/2$,
and let $g:=\nabla h_p=-w/\langle w,p\rangle$.
Fix a gradient-accuracy parameter $0<\eps_g\leq10$
and failure probability $0<\gamma\leq 1/2$.
Set
\begin{equation}
    \eta
    :=
    \frac{\rho\eps_g}{48\pi(3n+1)}.
\end{equation}
Let $H_p(x):=h_p(\rho x)$ for $x\in G_b^n$, where
$b:=\lceil\log_2((3n+1)/(\rho\eps_g))\rceil$. If we have access to oracle
$O_H:\ket{x}\ket0\to\ket{x}\ket{\widetilde H_p(x)}$, such that
\begin{equation}
    \abs{
        \widetilde H_p(x)-H_p(x)
    }
    \leq
    \eta,
\end{equation}
for every $x\in G_b^n$, then we can calculate a vector
$\widetilde g\in\R^n$ such that
$\snorm{\infty}{\widetilde g-g}\leq\eps_g$ with probability at least
$1-\gamma$, using
$O\!\left(\log\frac1\gamma\right)$ queries to $O_H$ and its inverse.
The additional non-oracle gate complexity is
\begin{equation}
    O\!\left(
        \p{n\log^2\frac{n}{\eps_g} + \log\frac{1}{\eps}}
        \log\frac1\gamma
    \right).
\end{equation}
\end{theorem}

\begin{proof}
We adapt the analysis of Jordan's gradient-estimation algorithm
from~\cite[Lemma~20]{Gily_n_2019} to the exactly-affine height function
$h_p$.

Under the condition of~\Cref{lem:coherent-height}, we have
$\frac1{10}\leq|\langle w,p\rangle|\leq3$. Define
\begin{equation*}
    g
    :=
    \nabla h_p
    =
    -\frac{w}{\langle w,p\rangle}.
\end{equation*}
Hence, $\frac13\leq\|g\|\leq10$, and $\|g\|_\infty\leq10$. The restriction $\eps_g\leq10$ only excludes a trivial regime: if
$\eps_g\geq10$, then the zero vector already satisfies
$\|0-g\|_\infty\leq\eps_g$, and no oracle queries are needed.
Moreover, by the choice
$\rho=2^{-\lceil\log_2(40\sqrt n)\rceil}$,
$
    \|\rho g\|_\infty
    \leq
    \frac{1}{40\sqrt n}\|g\|_\infty
    \leq
    \frac{1}{4\sqrt n}
    \leq
    \frac13.
$
As $h_p$ is affine,
\begin{equation}
    H_p(x)
    =
    h_p(\rho x)
    =
    h_p(0)+\rho\langle g,x\rangle.
    \label{eq:affine-scaled-height}
\end{equation}
Choose
$
    b=
    \left\lceil
        \log_2\frac{3n+1}{\rho\eps_g}
    \right\rceil$, and set $
    N:=2^b.
$
Then $b=O(\log(n/\eps_g))$. We use $n$
registers of $b$ qubits each. Their computational basis states are
interpreted directly as points $x\in G_b^n$. Starting from
$\ket0^{\otimes bn}$, Hadamard gates prepare the uniform superposition
\begin{equation*}
    \ket{\Psi_0}
    =
    \frac{1}{N^{n/2}}
    \sum_{x\in G_b^n}
    \ket{x}.
\end{equation*}

Suppose first that $H_p(x)$ could be evaluated exactly. We apply the
phase oracle with scaling $N$,
\begin{equation}
    O_H^{N}:
    \ket{x}
    \longmapsto
    e^{2\pi iN H_p(x)}
    \ket{x}.
\end{equation}
Using~\Cref{eq:affine-scaled-height}, it yields
\begin{align*}
    \ket{\Psi}
    &=
    O_H^{N}\ket{\Psi_0}
    \nonumber\\
    &=
    \frac{1}{N^{n/2}}
    \sum_{x\in G_b^n}
    e^{2\pi iN
        \left(
            h_p(0)+\rho\langle g,x\rangle
        \right)}
    \ket{x}
    \nonumber\\
    &=
    e^{2\pi iN h_p(0)}
    \left(
        \frac{1}{\sqrt N}
        \sum_{x_1\in G_b}
        e^{2\pi iN\rho g_1x_1}
        \ket{x_1}
    \right)
    \otimes\cdots\otimes
    \left(
        \frac{1}{\sqrt N}
        \sum_{x_n\in G_b}
        e^{2\pi iN\rho g_nx_n}
        \ket{x_n}
    \right).
\end{align*}
The prefactor
$e^{2\pi iN h_p(0)}$ is independent of $x$ and is therefore a
global phase. In particular, it has no effect on the subsequent inverse
Fourier transforms or on any measurement probabilities.

Applying the inverse Fourier transform to each register separately therefore gives, up to the same global phase,
\begin{equation*}
    \bigotimes_{i=1}^n
    \left(
        \sum_{k_i\in G_b}
        \left[
            \frac1N
            \sum_{x_i\in G_b}
            e^{
                2\pi iN
                \left(
                    \rho g_i-k_i
                \right)x_i
            }
        \right]
        \ket{k_i}
    \right).
\end{equation*}
Suppose we measure the output of the inverse Fourier transforms and
obtain $k=(k_1,\ldots,k_n)$. We use the more general form of the
standard phase-estimation tail bound: for every $r>1$ and every
$i\in[n]$,
\begin{equation*}
    \Pr\!\left[
        \left|
            k_i-\rho g_i
        \right|
        >
        \frac{r}{N}
    \right]
    \leq
    \frac{1}{2(r-1)}.
\end{equation*}
We choose $r:=3n+1$. Then, for every $i\in[n]$,
\begin{equation}
    \Pr\!\left[
        \left|
            k_i-\rho g_i
        \right|
        >
        \frac{3n+1}{N}
    \right]
    \leq
    \frac{1}{6n}.
\end{equation}
Hence, by a union bound over the $n$ coordinates,
\begin{equation*}
    \Pr\!\left[
        \max_{i\in[n]}
        \left|
            k_i-\rho g_i
        \right|
        >
        \frac{3n+1}{N}
    \right]
    \leq
    \frac16.
\end{equation*}
Thus, with probability at least $5/6$, all coordinates are
simultaneously accurate. On this event, defining
$\widetilde g_i:=\frac{1}{\rho}k_i$, we have
\begin{align}
    |\widetilde g_i-g_i|
    &\leq
    \frac{1}{\rho}\frac{3n+1}{N}
    =
    \frac{3n+1}{\rho N}
    \leq
    \eps_g
\end{align}
for every $i\in[n]$. Therefore, a single application of the ideal phase
oracle produces an $\eps_g$-accurate estimate of the entire gradient in
$\ell_\infty$ norm with probability at least $5/6$.

It remains to account for the approximation errors in our coherent height evaluation.
Let $\ket{\widetilde\Psi}$ denote the phase state obtained by replacing
$H_p(x)$ with $\widetilde H_p(x)$:
\begin{equation}
    \ket{x}
    \longmapsto
    e^{2\pi iN\widetilde H_p(x)}
    \ket{x}.
\end{equation}
This phase can be implemented by computing $\widetilde H_p(x)$ using
$O_H$, applying the corresponding phase rotation, and uncomputing with
$O_H^\dagger$, and hence requires only a constant number of
height-oracle queries.

For the ideal and approximate phase states,
\begin{align}
    \|
        \ket{\widetilde\Psi}-\ket{\Psi}
    \|^2
    &=
    \frac1{N^n}
    \sum_{x\in G_b^n}
    \left|
        e^{2\pi iN\widetilde H_p(x)}
        -
        e^{2\pi iN H_p(x)}
    \right|^2
    \nonumber\\
    &\leq
    \frac1{N^n}
    \sum_{x\in G_b^n}
    \left(
        2\pi N
        |\widetilde H_p(x)-H_p(x)|
    \right)^2
    \nonumber\\
    &\leq
    (2\pi N\eta)^2,
\end{align}
where we used $|e^{ia}-e^{ib}|\leq|a-b|$. Therefore, using that $N
    <
    \frac{2(3n+1)}{\rho\eps_g}$ by the definition of $b$, and $\eta=\frac{\rho\eps_g}{48\pi(3n+1)}$, we have
\begin{equation*}
    \|
        \ket{\widetilde\Psi}-\ket{\Psi}
    \|
    \leq
    2\pi N\eta\leq
    \frac1{12}.
\end{equation*}
For any measurement event, the difference between its probabilities on
two pure states is at most their trace distance, and
\begin{equation}
    \frac12
    \left\|
        \ket{\Psi}\!\bra{\Psi}
        -
        \ket{\widetilde\Psi}\!\bra{\widetilde\Psi}
    \right\|_1
    =
    \sqrt{
        1-
        |\langle\Psi|\widetilde\Psi\rangle|^2
    }
    \leq
    \|
        \ket{\widetilde\Psi}-\ket{\Psi}
    \|.
\end{equation}
Hence the probability of the good event
$\|\widetilde g-g\|_\infty\leq\eps_g$ can decrease by at most $1/12$.
Since the ideal procedure succeeds with probability at least $5/6$,
one approximate run succeeds with probability at least $\frac56-\frac1{12}
    =
    \frac34$.

We now repeat the entire procedure independently $R$ times, where $R$
is an odd integer satisfying $R=\Theta\!\left(\log\frac1\gamma\right)$. Let
$\widetilde g^{(1)},\ldots,\widetilde g^{(R)}$ denote the resulting
gradient estimates. Each whole vector satisfies $\|
        \widetilde g^{(j)}-g
    \|_\infty
    \leq
    \eps_g$
with probability at least $3/4$. By Hoeffding's inequality,
\begin{equation}
    \Pr\!\left[
        \text{at most $R/2$ runs are good}
    \right]
    \leq
    e^{-R/8}.
\end{equation}
Thus, choosing $R
    \geq
    8\log\frac1\gamma$ makes this probability at most $\gamma$. Finally, take the median coordinate-wise over the $R$ output vectors. If more than half of the runs are good then, for every coordinate, more than half of the estimates lie in
$[g_i-\eps_g,g_i+\eps_g]$, and hence their median lies in the same interval.
Consequently, with probability at least $1-\gamma$, we have $\|\widetilde g_{\mathrm{med}}-g\|_\infty\leq\eps_g$. Since each repetition uses only a constant number of calls to $O_H$ and $O_H^\dagger$, the total number of height-oracle queries is
\begin{equation*}
    O\!\left(
        \log\frac1\gamma
    \right).
\end{equation*}

For the gate complexity, our choice $N=\Theta(n/(\rho\eps_g))$ implies $b=\log N=O\!\left(\log\frac{n}{\eps_g}\right)$, where we used
$\rho^{-1}=\poly(n)$. A single run requires $n$ exact inverse Fourier transforms on $b$ qubits, each using $O(b^2)$ gates. The fixed-point height value uses $O(\log(1/\eps)+\log(1/\eta))$ bits. Implementing $e^{2\pi iN\widetilde H_p(x)}$ by bitwise phase rotations, therefore, requires $O(\log(1/\eps)+\log(1/\eta))$ additional gates per run. Including the Fourier transforms, the non-oracle gate complexity per run is
$O(nb^2+\log(1/\eps)+\log(1/\eta))$.
Repeating the procedure
$O(\log(2/\gamma))$ times gives the total quantum gate complexity
\begin{equation}
    O\!\left(
        \p{n\log^2\frac{n}{\eps_g} + \log\frac{1}{\eps}+\log\frac{n}{\eps_g}}
        \log\frac1\gamma
    \right).
\end{equation}
There is also classical post-processing of the measurement outcomes. For each of the $n$ coordinates, we compute the median of the $R$
measured $b$-bit values. Using a linear-time selection algorithm, this
requires $O(R)$ comparisons per coordinate, and hence
\begin{equation}
    O(nRb)
    =
    O\!\left(
        n\log\frac{n}{\eps_g}
        \log\frac1\gamma
    \right)
\end{equation}
classical bit operations. Thus, the overall computational complexity is
\begin{equation}
    O\!\left(
       \p{n\log^2\frac{n}{\eps_g} + \log\frac{1}{\eps}+\log\frac{n}{\eps_g}}
        \log\frac1\gamma +  n\log\frac{n}{\eps_g}
        \log\frac1\gamma
    \right) = O\!\left(
        \p{n\log^2\frac{n}{\eps_g} + \log\frac{1}{\eps}}
        \log\frac1\gamma
    \right).
\end{equation}
\end{proof}

\subsubsection{Putting everything together}
\label{sec:membership-putting-together}

We now combine the boundary-offset estimation procedure of
\Cref{thm:estimate-boundary-offset}, the coherent implementation of the
height function from~\Cref{lem:coherent-height}, and the gradient-estimation
procedure of~\Cref{lem:gradient-estimation} to obtain a quantum learner for
general linear threshold functions. Recall that
\begin{equation}
    f_{w,\theta}(x)
    :=
    \sgn(\langle w,x\rangle-\theta),
    \qquad
    w\in\sphere^{n-1},
    \quad
    \theta\in\R,
\end{equation}
and that, for a reference direction $p$ satisfying
$\langle w,p\rangle\neq0$, the associated height function obeys
\begin{equation*}
    h_p(x)
    =
    h_p(0)
    -
    \frac{\langle w,x\rangle}
         {\langle w,p\rangle},
    \qquad
    h_p(0)
    =
    \frac{\theta}{\langle w,p\rangle},
    \qquad
    g:=\nabla h_p
    =
    -\frac{w}{\langle w,p\rangle}.
\end{equation*}
Thus, recovering $g$, the orientation
$s:=\sgn(\langle w,p\rangle)$, and the boundary offset $h_p(0)$ suffices
to recover both parameters of the target halfspace. Indeed, since
$\|w\|=1$,
\begin{equation*}
    \|g\|
    =
    \frac{1}{|\langle w,p\rangle|},
    \qquad
    w
    =
    -s\frac{g}{\|g\|},
    \qquad
    \theta
    =
    \frac{s\,h_p(0)}{\|g\|}.
\end{equation*}

\begin{theorem}[Learning general LTFs with Membership Query]
\label{thm:quantum-ltf-learning}
For $n\geq 2$, let
$f_{w,\theta}(x)=\sgn(\langle w,x\rangle-\theta)$, where
$w\in\sphere^{n-1}$ and $\theta\in\R$, and suppose that the
learner has coherent membership-query access to $f_{w,\theta}$.
For every $0<\eps\leq1/2$, there exists a quantum algorithm
that outputs a hypothesis $\widehat f:\R^n\to\{-1,+1\}$ such
that
\begin{equation*}
    \Pr_{x\sim\gamma_n}
    \left[
        \widehat f(x)
        \neq
        f_{w,\theta}(x)
    \right]
    \leq
    \eps
\end{equation*}
with probability at least $2/3$.
The algorithm uses $O\left(\log\frac{n}{\eps}\right)$ membership queries, and $\widetilde O(n)$ additional gates.
\end{theorem}

\begin{proof}
We first use the randomized classical preprocessing procedure
of~\Cref{lem:finite-reference} to choose a reference vector $p$.
For every fixed $w\in\sphere^{n-1}$, with probability at least
$4/5$, we have
\begin{equation}
    \frac1{10}
    \leq
    |\langle w,p\rangle|
    \leq
    3.
    \label{eq:final-good-reference}
\end{equation}
This step needs \(\widetilde O(n)\) additional classical bit operations.

We do not assume that the threshold $\theta$ is known. We first test the endpoints to determine whether a constant hypothesis suffices. Let $T_\eps:=\sqrt{2\log(1/\eps)}$ and define
\[
    L_\eps
    :=
    2^{\lceil\log_2(20T_\eps)\rceil}.
\]
The algorithm queries $f_{w,\theta}$ at the two endpoints
$\pm L_\eps p$.
Write $y_-:=f_{w,\theta}(-L_\eps p)$ and
$y_+:=f_{w,\theta}(L_\eps p)$.
We analyze the two possible outcomes conditioned on the
good-reference event in~\Cref{eq:final-good-reference}.

\paragraph{The endpoint labels agree.}
In this case, $|h_p(0)|\geq L_\eps$. Since $|\langle w,p\rangle|\geq1/10$, we obtain
\[
    |\theta|
    =
    |h_p(0)|\,|\langle w,p\rangle|
    \geq
    \frac{L_\eps}{10}
    \geq
    2T_\eps.
\]
Since \(r\mapsto f_{w,\theta}(rp)\) is monotone, equal endpoint labels imply \(f_{w,\theta}(0)=y_-=y_+\). We therefore output the constant hypothesis \(c\equiv y_+\). By~\Cref{lem:large-threshold-constant}, this hypothesis has
Gaussian classification error at most $\eps$, and the
algorithm terminates.

\paragraph{The endpoint labels differ.}
In this case, $h_p(0)\in[-L_\eps,L_\eps]$. The boundary-search procedure~\Cref{thm:estimate-boundary-offset} then estimates $h_p(0)$ and returns $s=\sgn(\langle w,p\rangle)$. We use these quantities together with gradient estimation to recover an $\eps$-accurate hypothesis, as follows.

We condition on the good-reference event in~\Cref{eq:final-good-reference} throughout the remainder of the analysis. In
particular,
\begin{equation}
    \frac13
    \leq
    \|g\|
    =
    \frac1{|\langle w,p\rangle|}
    \leq
    10.
    \label{eq:final-gradient-norm}
\end{equation}

We apply~\Cref{thm:estimate-boundary-offset} with
$\zeta:=\eps/100$.
The procedure returns
\begin{equation}
    s
    =
    \sgn(\langle w,p\rangle)
\end{equation}
together with an estimate $\widetilde h_p(0)$ such that
\begin{equation}
    |\widetilde h_p(0)-h_p(0)|
    \leq
    \frac{\eps}{100}.
    \label{eq:final-offset-error}
\end{equation}
We next estimate the gradient. Define $\Delta:=\frac{\eps}{100(L_\eps+1)}$ and set $\eps_g:=\frac{\Delta}{\sqrt n}$ and $\gamma:=\frac16$. Following~\Cref{lem:gradient-estimation}, define $\eta:=\frac{\rho\eps_g}{48\pi(3n+1)}$ and $H_p(x):=h_p(\rho x)$. Since $\rho=\Theta(n^{-1/2})$ and
$\eps_g=\Theta(\eps/((L_\eps+1)\sqrt n))$, we have
\begin{equation*}
    \log\frac1\eta
    =
    O\!\left(
        \log\frac{n}{\eps}
    \right),
\end{equation*}
where we used $L_\eps=O(\sqrt{\log(1/\eps)})$. By~\Cref{lem:coherent-height}, we can implement coherent value access $O_H:
    \ket{x}\ket0
    \longmapsto
    \ket{x}\ket{\widetilde H_p(x)}$ such that
\begin{equation}
    |\widetilde H_p(x)-H_p(x)|
    \leq
    \eta,
\end{equation}
using $O(\log(1/\eta))=O(\log(n/\eps))$ queries to $O_{f_{w,\theta}}$. We then invoke~\Cref{lem:gradient-estimation}, which returns
$\widetilde g\in\R^n$ satisfying
$\|\widetilde g-g\|_\infty\leq\eps_g$ with probability at
least $1-\gamma=5/6$.
Consequently, with the same probability,
\begin{equation}
    \|\widetilde g-g\|
    \leq
    \sqrt n\,\eps_g
    =
    \Delta.
    \label{eq:final-gradient-l2}
\end{equation}

We now reconstruct the normal vector and the threshold.
If $\widetilde g=0$, the algorithm outputs an arbitrary
constant hypothesis and terminates.
Otherwise, define
\begin{equation*}
    \widehat w
    :=
    -s
    \frac{\widetilde g}{\|\widetilde g\|},
    \qquad
    \widehat\theta
    :=
    \frac{s\,\widetilde h_p(0)}
         {\|\widetilde g\|},
\end{equation*}
and output $\widehat f:=f_{\widehat w,\widehat\theta}$.

For the following error analysis, also condition on the
gradient-estimation event in~\Cref{eq:final-gradient-l2}.
Since $\Delta\leq1/100<1/12$,
\Cref{eq:final-gradient-norm} gives
\begin{equation}
    \|\widetilde g\|
    \geq
    \|g\|-\|\widetilde g-g\|
    \geq
    \frac13-\frac1{12}
    =
    \frac14.
    \label{eq:estimated-gradient-lower bound}
\end{equation}
In particular, $\widetilde g\neq0$ on this event. By the standard normalization inequality,
\begin{equation*}
     \|\widehat w-w\| = \left\|
        \frac{\widetilde g}{\|\widetilde g\|}
        -
        \frac{g}{\|g\|}
    \right\|
    \leq
    \frac{
        2\|\widetilde g-g\|
    }{
        \|g\|
    } \leq 6\Delta,
\end{equation*}
where we used $w=-sg/\|g\|$, together with \Cref{eq:final-gradient-norm,eq:final-gradient-l2}. Moreover,
\begin{equation}
    \left|
        \frac1{\|\widetilde g\|}
        -
        \frac1{\|g\|}
    \right|
    =
    \frac{
        \bigl|
            \|\widetilde g\|-\|g\|
        \bigr|
    }{
        \|\widetilde g\|\|g\|
    }
    \leq
    12\Delta,
    \label{eq:inverse-gradient-norm-error}
\end{equation}
where we used~\Cref{eq:final-gradient-norm},
\Cref{eq:estimated-gradient-lower bound}, and the reverse
triangle inequality.

On the other hand, since $\theta=s h_p(0)/\|g\|$, combining
\Cref{eq:final-offset-error},
\Cref{eq:estimated-gradient-lower bound}, and
\Cref{eq:inverse-gradient-norm-error}, we obtain
\begin{align*}
    |\widehat\theta-\theta|
    &\leq
    \frac{
        |\widetilde h_p(0)-h_p(0)|
    }{
        \|\widetilde g\|
    }
    +
    |h_p(0)|
    \left|
        \frac1{\|\widetilde g\|}
        -
        \frac1{\|g\|}
    \right|\\
    &\leq
    4\frac{\eps}{100}
    +
    12L_\eps\Delta\\
    &\leq
    \frac{4\eps}{100}
    +
    \frac{12L_\eps\eps}
         {100(L_\eps+1)}\\
    &\leq
    \frac{16\eps}{100}.
\end{align*}
Finally, since $\|\widehat w\|=\|w\|=1$,
\Cref{lem:gaussian-disagreement} gives
\begin{align*}
    \Pr_{x\sim\gamma_n}
    \left[
        f_{\widehat w,\widehat\theta}(x)
        \neq
        f_{w,\theta}(x)
    \right]
    &\leq
    \frac12\|\widehat w-w\|
    +
    \frac1{\sqrt{2\pi}}
    |\widehat\theta-\theta|\\
    &\leq
    3\Delta
    +
    \frac{16\eps}{100\sqrt{2\pi}}
    <
    \eps.
\end{align*}
The last inequality uses $\Delta\leq\eps/100$.

Conditioned on the good-reference event
in~\Cref{eq:final-good-reference}, the equal-endpoint branch
succeeds deterministically.
In the differing-endpoint branch, all steps preceding gradient estimation are deterministic, and
\Cref{lem:gradient-estimation} succeeds with probability
at least $5/6$.
Therefore, conditioned on any good reference vector, the 
algorithm succeeds with probability at least $5/6$.
Since the good-reference event has a probability of at least
$4/5$, the overall success probability is at least
$\frac45\cdot\frac56=\frac23$.

We now consider the membership-query complexity.
The endpoint test and the orientation query use only a
constant number of membership queries.
The boundary-offset estimation uses
\begin{equation}
    O\!\left(
        \log\frac{L_\eps}{\zeta}
    \right)
    =
    O\!\left(
        \log\frac1\eps
    \right)
\end{equation}
membership queries, where $\zeta=\eps/100$.

By~\Cref{lem:gradient-estimation}, gradient estimation makes
$O(\log(1/\gamma))=O(1)$ calls to $O_H$ and its inverse.
Each call is implemented using~\Cref{lem:coherent-height},
with membership-query complexity
\begin{equation}
    O\!\left(
        \log\frac1\eta
    \right)
    =
    O\!\left(
        \log\frac{n}{\eps}
    \right).
\end{equation}
Therefore, the total number of membership queries is
\begin{equation}
    O\!\left(
        \log\frac{n}{\eps}
    \right).
\end{equation}

One coherent height query has gate complexity
\begin{equation}
    O\!\left(
        n
        \left(
            b+\log n
            +\log\frac1\eps
            +\log\frac1\eta
        \right)^2
        \log\frac1\eta
    \right)
    =
    O\!\left(
        n\log^3\frac{n}{\eps}
    \right),
\end{equation}
where we used $b=O(\log(n/\eps))$ and
$\log(1/\eta)=O(\log(n/\eps))$.
By~\Cref{lem:gradient-estimation}, with $\gamma=1/6$,
the gradient-estimation algorithm makes $O(1)$ calls to
the coherent height oracle and its inverse.
Its additional non-oracle gates, including the inverse
Fourier transforms and the value-to-phase conversion,
have complexity
\begin{equation}
    O\!\left(
        \left(
            n\log^2\frac{n}{\eps}
            +\log\frac1\eps
        \right)
        \log\frac1\gamma
    \right)
    =
    O\!\left(
        n\log^2\frac{n}{\eps}
    \right).
\end{equation}
Here we used $\log(1/\eps)\leq\log(n/\eps)$ and
$\gamma=1/6$. Therefore, after including the implementation of all
height-oracle calls, the non-oracle quantum gate complexity is
\begin{equation}
    O\!\left(
        n\log^3\frac{n}{\eps}
        +
        n\log^2\frac{n}{\eps}
    \right)
    =
    O\!\left(
        n\log^3\frac{n}{\eps}
    \right)
    =
    \widetilde O(n).
\end{equation}

The classical boundary search uses
$O(\log(1/\eps))$ rounds on numbers with
$O(\log(n/\eps))$ bits per query coordinate, so its
non-oracle computational cost is also $\widetilde O(n)$.
Finally, the output hypothesis can be represented as
$\widehat f(x)=
\sgn(-s\langle\widetilde g,x\rangle-s\widetilde h_p(0))$,
which equals $f_{\widehat w,\widehat\theta}(x)$ whenever
$\widetilde g\neq0$. Together with the preprocessing
of~\Cref{lem:finite-reference}, these operations require
only $\widetilde O(n)$ additional classical bit operations.
Therefore, the total non-oracle gate complexity is
\begin{equation}
    \widetilde O(n).
\end{equation}
Combining the analyses of the two endpoint-test outcomes
completes the proof.
\end{proof}

\subsection{A matching lower bound}
We use a continuous-parameter version of the \emph{multiplicative adversary method}~\cite{spalek2007multiplicativequantumadversary}. At a high level, the method tracks how rapidly a quantum algorithm can acquire information about an unknown target through oracle queries. Suppose that, after $t$ queries, the algorithm has a family of states $\{\ket{\psi_w^t}\}_w$, indexed by the possible hidden targets $w$. The basic question is how quickly oracle queries can make these states sufficiently dependent on, and ultimately distinguishable with respect to, the unknown target. To analyze all possible targets simultaneously, the method introduces an auxiliary reference register indexing the hidden input and chooses a positive adversary operator $\Gamma\succeq I$ acting on this register. For the resulting joint state $\ket{\Psi_t}$ after $t$ queries, one defines the \emph{progress measure} $W_t:=\bra{\Psi_t}(\Gamma\otimes I)\ket{\Psi_t}$. Intuitively, $\Gamma$ assigns different weights to different spectral directions of the hidden-input space, with larger weights corresponding to components that are more informative for solving the problem. The quantity $W_t$ therefore measures how much weight the algorithm's state has accumulated in these informative directions. Importantly, the ``input'' indexed by the reference register is the hidden target parameter $w$, rather than the query point $x$ supplied to the membership oracle.

The lower-bound argument therefore has three standard components:
\begin{enumerate}
    \item[(i)] The initial progress $W_0$ is small.
    \item[(ii)] A single oracle query can increase the progress by at most a bounded multiplicative factor,
    $W_{t+1}\leq C\,W_t$.
    \item[(iii)] Any final state that solves the problem with the desired success probability must have sufficiently large progress $W_T$.
\end{enumerate}
If these three properties hold, then $W_T\leq C^T W_0$. Comparing this upper bound with the progress required for successful learning yields a lower bound on the number of queries $T$.

We now present our adversary construction and establish these three properties. 

\paragraph{Adversary construction.}
Let $G\sim\gamma_n$ be a standard Gaussian vector and set $w=G/\norm{G}$. By the rotational invariance of the Gaussian distribution, $w$ is uniformly distributed on $\sphere^{n-1}$. Moreover, defining
\begin{equation}
    f_G(x):=\sgn(\langle G,x\rangle),
\end{equation}
we have $f_G=f_w$ almost surely. We use $G$ as a Gaussian parameterization of the unknown target, while the learning objective remains to approximate $f_w$. Working with $G$, rather than directly with $w$, allows us to use Hermite polynomials to describe the learner's dependence on the unknown target.

Following the multiplicative adversary framework~\cite{spalek2007multiplicativequantumadversary}, we analyze the learner's executions for all hidden targets simultaneously. After purifying the algorithm and deferring all measurements, let $\ket{\psi_G^t}\in\mathcal H_{\mathrm{alg}}$ denote its state after $t$ queries to $O_{f_G}$. We introduce an inaccessible reference system $\mathcal H_I:=L^2(\gamma_n)$ and collect this family of states into a single joint vector $\ket{\Psi_t}\in\mathcal H_I\otimes\mathcal H_{\mathrm{alg}}\cong L^2(\gamma_n;\mathcal H_{\mathrm{alg}})$ such that
\begin{equation}
    \Psi_t(G):=\ket{\psi_G^t}.
\end{equation}
The Gaussian prior is incorporated into the measure defining $L^2(\gamma_n)$. Since the algorithm's state is normalized for every target, the joint vector is also normalized:
\begin{equation}
    \norm{\Psi_t}^2
    =\int_{\mathbb R^n}\|\psi_G^t\|^2\,d\gamma_n(G)
    =1.
\end{equation}

We next construct the adversary operator $\Gamma$. Let $\{h_\alpha\}_{\alpha\in\mathbb N_0^n}$ be the orthonormal multivariate Hermite basis of $L^2(\gamma_n)$, and write
\begin{equation}
    \Psi_t(G)
    =\sum_{\alpha\in\mathbb N_0^n}
        h_\alpha(G)v_{t,\alpha},
    \qquad
    v_{t,\alpha}\in\mathcal H_{\mathrm{alg}}.
\end{equation}
Equivalently, denoting the corresponding reference-system basis vectors by $\ket{h_\alpha}$,
\begin{equation}
    \ket{\Psi_t}
    =\sum_{\alpha\in\mathbb N_0^n}
        \ket{h_\alpha}_I\otimes v_{t,\alpha},
\end{equation}
with $\sum_\alpha\|v_{t,\alpha}\|^2=1$. The distribution of $\|v_{t,\alpha}\|^2$ across Hermite degrees captures the spectral complexity of the algorithm's dependence on the \emph{unknown parameter} $G$. We therefore choose $\Gamma$ to assign larger weights to higher total Hermite degrees. Let $N$ be the Hermite number operator, $Nh_\alpha=|\alpha|h_\alpha$, and define
\begin{equation}
    \Gamma:=(I+N)^{1/4}.
\end{equation}
Thus, $\Gamma$ acts only on the reference system and satisfies
\begin{equation}
    \Gamma\ket{h_\alpha}
    =(1+|\alpha|)^{1/4}\ket{h_\alpha}.
\end{equation}
The degree-zero subspace has weight $1$, while a component of total degree $k$ has weight $(1+k)^{1/4}$. This choice allows us to control the multiplicative growth caused by a single oracle query while ensuring that substantial high-degree mass produces large final progress.

Based on the constructions of $\ket{\Psi_t}$ and $\Gamma$, we define the progress measure as follows. For a vector-valued function $F=\sum_\alpha h_\alpha v_\alpha$, let
\begin{equation*}
    \mathcal W(F)
    :=\norm{(\Gamma^{1/2}\otimes I_{\mathrm{alg}})F}^2
    =\sum_\alpha(1+|\alpha|)^{1/4}\|v_\alpha\|^2.
\end{equation*}
We then set $W_t=\mathcal W(\Psi_t)$. We interpret this expression as a quadratic form, with domain consisting of functions for which the sum is finite. For a normalized joint state, $W_t$ is the weighted average of its Hermite levels. In particular, a constant amount of squared norm at sufficiently high degree forces $W_t$ to be large.

\paragraph{Small initial progress $W_0$.}
We first evaluate the initial progress. Before making any query, the algorithm's state is independent of $G$, so $\Psi_0(G)=\ket{\psi_0}$. Since the constant function is the degree-zero Hermite polynomial $h_0(\cdot)=1$, we obtain
\begin{equation}
    W_0=1.
\end{equation}
Thus, we complete (i).

\paragraph{Bounded growth per query.}
An oracle-independent unitary $U$ maps each Hermite coefficient $v_\alpha$ to $Uv_\alpha$ and therefore preserves its norm. Consequently,
\begin{equation}
    \mathcal W\bigl((I_I\otimes U)F\bigr)
    =\sum_\alpha(1+|\alpha|)^{1/4}\|Uv_\alpha\|^2
    =\mathcal W(F).
\end{equation}
Thus, only oracle queries can change the progress measure. The following lemma bounds the multiplicative increase caused by a single query. For readability, we defer the proof to~\Cref{app: membership_lower_bound}. 

\begin{lemma}[Constant multiplicative growth per query]
\label{lem:ltf-query-multiplier}
There is a universal constant $C>1$ such that, for every $n\geq1$, every $x\in\mathbb R^n\setminus\{0\}$, and every Hilbert-space-valued $F$ with $\mathcal W(F)<\infty$,
\begin{equation}
\label{eq:ltf-query-multiplier}
    \mathcal W\!\left(
        \sgn\langle G,x\rangle F
    \right)
    \leq C\,\mathcal W(F).
\end{equation}
\end{lemma}

To see how~\Cref{lem:ltf-query-multiplier} controls a quantum membership query, recall that in the response register's $\{\ket{+},\ket{-}\}$ basis, the oracle acts either as the identity or as multiplication by $f_G(x)=\sgn\langle G,x\rangle$. Therefore, for every fixed query point $x$, the nontrivial action of the oracle on the hidden-parameter dependence is exactly the multiplication operator appearing in~\Cref{eq:ltf-query-multiplier}. Since the bound in~\Cref{eq:ltf-query-multiplier} holds uniformly in $x$, it also applies when the query register is in superposition. Hence, a single oracle query increases the progress measure by at most a constant multiplicative factor,
\begin{equation}
    W_{t+1}\leq C\,W_t.
\end{equation}
Therefore, after $T$ oracle queries,
\begin{equation}
    W_T\leq C^T W_0 = C^T.
    \label{eq:ltf-progress-upper}
\end{equation}
We complete (ii).

\paragraph{Successful learning requires high Hermite degree.}
It remains to lower bound the final progress required for successful learning. Let $P_{\leq d}$ denote the orthogonal projection onto total Hermite degree at most $d$, acting on the reference system. We first observe that a constant amount of mass above degree $d$ immediately implies large progress.

\begin{lemma}
\label{lem:ltf-degree-to-queries}
If the final state of a $T$-query algorithm satisfies $\|(I-P_{\leq d})\Psi_T\|^2\geq\frac14$, then
\begin{equation*}
   W_T\geq\frac14(d+2)^{1/4}.
\end{equation*}
\end{lemma}

\begin{proof}
Every Hermite component above degree $d$ has adversary weight $(1+|\alpha|)^{1/4}\geq (d+2)^{1/4}$. Thus
\begin{equation}
     W_T
    \geq(d+2)^{1/4}\|(I-P_{\leq d})\Psi_T\|^2
    \geq\frac14(d+2)^{1/4}.
\end{equation}
\end{proof}
We therefore established (iii). However, it remains to determine how large $d$ must be for a successful learner to localize the unknown direction $w=G/\|G\|$ to accuracy $\eps$. We reduce this to a similar question about polynomials on the sphere. We will show that the relevant scale is $d=\Omega(n/\eps)$.

For a fixed output hypothesis $h$, define its set of successful target directions by $A_h:=\{w\in \sphere^{n-1}:d_{\gamma_n}(h,f_w)\leq\eps\}$. If $u,v\in A_h$, then by~\Cref{lem:gaussian-disagreement},
\begin{equation*}
    \frac{\arccos(\inner{u}{v})}{\pi}
=
d_{\gamma_n}(f_u,f_v)
\leq
d_{\gamma_n}(f_u,h)
+
d_{\gamma_n}(h,f_v)
\leq 2\eps.
\end{equation*}
Hence every nonempty $A_h$ is contained in a spherical cap of angular radius $2\pi\eps$. This can be viewed through an $\eps$-net of $\sphere^{n-1}$. At angular resolution $\Theta(\eps)$, the sphere contains roughly $\left(\frac{1}{\eps}\right)^{n-1}$ distinguishable regions~\cite{Vershynin2018}. Thus, producing an $\eps$-accurate hypothesis amounts to localizing the unknown direction to one such $\Theta(\eps)$-sized region. Equivalently, the set of directions consistent with any fixed successful output occupies only an $O(\eps^{n-1})$ fraction of the sphere. More precisely, writing $\mu$ for normalized surface measure,
\begin{equation*}
\mu(A_h)
\leq
(2\pi\eps)^{n-1}.
\end{equation*}

We now relate this angular localization to polynomial degree. Let $\Pi_{\leq d}(\sphere^{n-1})$ denote the restrictions to the sphere of polynomials on $\mathbb R^n$ of total degree at most $d$. Its dimension is
\begin{equation*}
D_{n,d}
=
\binom{n+d-1}{n-1}
+
\binom{n+d-2}{n-1}
\leq
2
\left[
e\left(1+\frac{d}{n-1}\right)
\right]^{n-1}.
\end{equation*}
The sharp $L^2$-to-$L^\infty$ inequality for spherical polynomials~\cite{dai2020nikolskii,dai2013approximation} implies that, for every measurable $A\subseteq \sphere^{n-1}$ and every Hilbert-space-valued polynomial $q$ of degree at most $d$,
\begin{equation*}
\int_A \|q(\omega)\|^2d\mu(\omega)
\leq
D_{n,d}\mu(A)
\int_{\sphere^{n-1}}\|q(\omega)\|^2d\mu(\omega).
\end{equation*}
Applying this inequality on every sphere of fixed Gaussian radius and then integrating over the radius gives, for every degree-at-most-$d$ polynomial $p$ on $\mathbb R^n$,
\begin{equation}
\label{eq:gaussian-angular-concentration}
\int_{\{G:G/\|G\|\in A\}}
\|p(G)\|^2d\gamma_n(G)
\leq
D_{n,d}\mu(A)\|p\|^2.
\end{equation}

\Cref{eq:gaussian-angular-concentration} gives the spectral analogue of the $\eps$-net intuition. A degree-$d$ polynomial can place at most a $D_{n,d}\mu(A)$ fraction of its squared norm inside a region $A$. Hence, if the degree-$\leq d$ part were to account for a constant fraction of the successful probability inside a cap of measure $O(\eps^{n-1})$, it would be necessary that
\begin{align}
    D_{n,d}\eps^{n-1} &= \Omega(1),\\
    \left[
        \eps\left(1+\frac{d}{n}\right)
    \right]^{n-1}
    &=\Omega(1),\\
    d&=\Omega\left(\frac{n}{\eps}\right),
\end{align}
up to universal constants in the base. Thus, Hermite degree of order $n/\eps$ is the natural scale required to localize the target direction to a $\Theta(\eps)$-sized region of the sphere with constant probability.

We make this intuition precise in the following lemma.

\begin{lemma}[Successful learning requires high Hermite degree]
\label{lem:ltf-learning-requires-degree}
For $n\geq2$ and $0<\eps\leq c:=\frac{1}{128\pi e}$, let $d:=\left\lfloor\frac{c(n-1)}{\eps}\right\rfloor$. Suppose that, averaged over a uniformly random $w\in \sphere^{n-1}$, a learner outputs a classical hypothesis $h$ satisfying $d_{\gamma_n}(h,f_w)\leq\eps$ with probability at least $2/3$. Then its purified final state satisfies
\begin{equation*}
    \|(I-P_{\leq d})\Psi_T\|^2>\frac14.
\end{equation*}
\end{lemma}

\begin{proof}
Write $L:=P_{\leq d}\Psi_T$ and $H:=(I-P_{\leq d})\Psi_T$. Let $S$ denote the pointwise orthogonal projector onto successful output--target pairs: for each hidden direction $w=G/\norm{G}$, it retains exactly those classical output hypotheses $h$ for which $w\in A_h$. Therefore $\|S\Psi_T\|^2$ is precisely the learner's average success probability.

Applying~\Cref{eq:gaussian-angular-concentration} separately to each output component of the degree-$\leq d$ state $L$, and then summing over all output hypotheses, yields
\begin{align}
\|SL\|^2
&\leq
D_{n,d}(2\pi\eps)^{n-1}\|L\|^2 \nonumber\\
&\leq
2
\left[
2\pi e\eps
\left(1+\frac{d}{n-1}\right)
\right]^{n-1}\nonumber\\
&\leq
2
\left[
2\pi e(\eps+c)
\right]^{n-1}
\leq
2(1/32)^{n-1}
\leq
\frac1{16}.
\end{align}
Hence $\|SL\|\leq\frac14$. Notice that no factor depending on the number of possible hypotheses appears: the concentration bound is applied separately to each output component, after which the corresponding squared norms sum to at most $\|L\|^2$.

On the other hand, the assumed average success probability gives $\|S\Psi_T\|\geq\sqrt{\frac23}$. Since $S$ is an orthogonal projector and $\Psi_T=L+H$,
\begin{equation}
\begin{aligned}
\|H\|
&\geq
\|SH\|\\
&=
\|S\Psi_T-SL\|\\
&\geq
\|S\Psi_T\|-\|SL\|\\
&\geq
\sqrt{\frac23}-\frac14\\
&>
\frac12.
\end{aligned}
\end{equation}
Therefore
\begin{equation}
    \|(I-P_{\leq d})\Psi_T\|^2
=
\|H\|^2
>
\frac14.
\end{equation}
\end{proof}

Combining~\Cref{lem:ltf-learning-requires-degree,lem:ltf-degree-to-queries}, any successful learner must therefore satisfy
\begin{equation}
    W_T
    \geq
    \frac14(d+2)^{1/4},
    \qquad
    d=\Theta\!\left(\frac{n}{\eps}\right).
\end{equation}
Together with the per-query upper bound~\Cref{eq:ltf-progress-upper}, this completes the multiplicative adversary argument.

\begin{theorem}[Quantum membership-query lower bound]\label{thm:quantum-ltf-lower bound}
There is a universal constant $\eps_0>0$ such that, for every $n\geq2$ and $0<\eps\leq\eps_0$, any quantum learner that, for every $w\in \sphere^{n-1}$, outputs a classical hypothesis $h$ with $d_{\gamma_n}(h,f_w)\leq\eps$ with probability at least $2/3$ requires
\begin{equation}
    \Omega\!\left(\log\frac{n}{\eps}\right)
\end{equation}
membership queries. This holds even for adaptive algorithms making superposition queries at arbitrary real points and outputting improper hypotheses (i.e., an $h$ that is not itself an LTF).
\end{theorem}

\begin{proof}
A worst-case success probability of at least $2/3$ implies the same average success probability when $w$ is uniformly distributed on $\sphere^{n-1}$. Therefore, by~\Cref{lem:ltf-learning-requires-degree}, for $d=\left\lfloor c(n-1)/\eps\right\rfloor=\Theta(n/\eps)$, the final state satisfies $\|(I-P_{\leq d})\Psi_T\|^2>1/4$. Applying~\Cref{lem:ltf-degree-to-queries} gives
\begin{equation}
    W_T\geq\frac14(d+2)^{1/4}.
\end{equation}
On the other hand,~\Cref{eq:ltf-progress-upper} gives $W_T\leq C^T$. Hence
\begin{equation}
    C^T
    \geq
    \frac14(d+2)^{1/4}.
\end{equation}
Taking logarithms and replacing $d=\Theta(n/\eps)$,
\begin{equation}
    T\log C
    \geq
    \frac14\log\frac{n}{\eps}-O(1). 
\end{equation}
Since $C>1$ and $\eps_0$ is a universal constant, we obtain
\begin{equation}
    T
    =
    \Omega\!\left(\log\frac{n}{\eps}\right),
\end{equation}
which proves the theorem.
\end{proof}

\section{Learning equal-weight \texorpdfstring{$w$}{w} with a Boolean Oracle: learning Majority}
\label{sec:boolean}

We now consider a special case of the homogeneous LTF learning problem studied in the previous sections. 
Suppose that
\[
w=\frac{\mathbf 1_A}{\sqrt{|A|}}
\]
for some unknown set $A\subseteq[n]$ with the promise that $1\leq |A|\leq k$. In the regime $\eps<1/(2\sqrt{k})$, learning $w$ to Euclidean error at most $\eps$ is enough to recover $A$ \emph{exactly}. We further restrict the membership queries to $x\in\{-1,1\}^n$. We have
\[
f_w(x) =
\sgn(\langle w,x\rangle)
=
\operatorname{MAJORITY}_{|A|}(x_A)
:=\begin{cases}
    1 &\text{ if } \quad |x_A| \geq |A|/2\\
    -1 &\text{ otherwise }
\end{cases}
\]

Therefore, learning an equal-weight halfspace using Boolean membership queries is exactly the problem of learning a Majority-junta with hidden set $A$.\footnote{Note that if we choose instead $x\in\{0,1\}^n$, and set $w_i$ to be all negative with $\theta=0$, then the problem becomes equivalent to combinatorial group testing, whose quantum query complexity is $\Theta(\sqrt{k})$ as was shown in \cite{belovs2015symmetricjuntas}.}

This gives a natural \emph{discrete} version of the linear threshold learning problem using membership queries. The quantum algorithm from the previous section doesn't work here, because it queries arbitrary points in $\mathbb{R}^n$ and uses gradient estimation. This technique does not apply to the Boolean domain, so we instead investigate the query complexity of this task using the adversary bound.

The problem of junta learning has been studied classically since its introduction in
\cite{angluin1988queries}. For arbitrary $h$ applied to substring $x_A$, under the promise $|A|\leq k\leq n$, its classical query
complexity is at least $\Omega\!\left(k\log(n/k)\right)$~\cite{bshouty1996oracles}. 

We treat the special case where $h$ (as defined in \Cref{def:junta}) is the majority function. This problem, that we call the \emph{Majority-learning} problem, was studied by Belovs~\cite{belovs2015symmetricjuntas}, who gave an $O(k
^{1/4})$ upper bound under the promise $|A|=k$. To the best of our knowledge, no prior lower bound was known, leaving an exponential gap between the existing lower and upper bounds.

Surprisingly, we show here that, for $k\geq2$, learning Majority under the promise $|A|=k$ has an exponentially improved $O(\log k)$ upper bound, which we also show to be tight when $n\geq2k$, thus solving the open problem in \cite[Section~8]{belovs2015symmetricjuntas}. We then extend our results to the more general promise $1\leq|A|\leq k$, and prove a tight $\Theta(\log k)$ bound for every $n\geq k\geq2$.

\begin{theorem}[Query complexity of learning Majority]
\label{thm:majority-query-complexity}
Let $k\geq2$.
Under one of the following two promises
\begin{enumerate}[label=\textup{(\roman*)}]
\item $|A|=k$, for every $n\geq2k$
\item $1\leq |A|\leq k$, for every $n\geq k$
\end{enumerate}
the bounded-error quantum query complexity of learning
Majority is
$\Theta(\log k)$.
\end{theorem}

\subsection{Our technique}
In this section, we present the technique that led to the family of feasible solutions of \Cref{eq:junta-dual-adversary} with cost $O(\log k)$.
The dual construction was inspired by the numerical solutions to a restrictive version of the SDP 
in \Cref{eq:junta-dual-adversary}, which we call the \emph{intersection-free} restriction. 
The restriction serves two purposes:
it makes numerical exploration tractable since the matrix dimension of the SDP is exponentially smaller, 
and the restriction enforces low-rank solutions, 
which reduces the space used by the span program or transducer \cite{czekanski2023robust}.

Before introducing the intersection-free framework, we will first reduce the number of variables of the SDP 
via symmetry
reduction, also called restriction to the
invariant subspace in \cite{bachoc2010invariant}. 
Given a feasible family
$\{X_S\}_{S\subseteq[n]}$, define its group average by
\begin{equation*}
\overline X_S\llbracket A,B\rrbracket
=
\frac{1}{n!}\sum_{\pi\in S_n}
X_{\pi(S)}\llbracket \pi(A),\pi(B)\rrbracket\,.
\end{equation*}
Averaging preserves positive semidefiniteness and the adversary constraints,
and does not increase the objective value. In particular, the average of an
optimal solution is optimal and has the same objective value. Furthermore, such a solution satisfies
\begin{equation}
\overline X_{\pi(S)}\llbracket \pi(A),\pi(B)\rrbracket
=
\overline X_S\llbracket A,B\rrbracket.
\label{eq:symmetric-covariance}
\end{equation}
For fixed $S$, each $X_S$ remains a
$\binom{n}{k}\times\binom{n}{k}$ matrix, with entries indexed by pairs
$(A,B)$. The simultaneous orbit of an index triple $(S,A,B)$ is
\[
\mathcal O(S,A,B)
:=
\{(\pi(S),\pi(A),\pi(B)): \pi\in S_n\}\,.
\]
By \Cref{eq:symmetric-covariance}, the entry
$\overline X_S\llbracket A,B\rrbracket$ depends only on this orbit, which is uniquely determined by
\[
|S| \quad |A| \quad |A\cap S| \quad |B| \quad |B\cap S| \quad
|A\cap B| \quad |A\cap B\cap S|\,.
\]
Thus symmetry reduction preserves the optimum and reduces the number of
distinct entries, but not the dimensions of the matrices $X_S$. Note that for the promise $|A|=k$, the variables $|A|$ and $|B|$ are fixed. For this exact promise, we use the following restriction of the orbit reduced adversary SDP, that we call the \emph{intersection-free} restriction:

\begin{definition}
\label{def:intersection-free-witness}
The \emph{intersection-free} junta SDP restricts
\Cref{eq:junta-dual-adversary} to matrices of the form
\begin{equation}
\label{eq:intersection-free-matrices}
X_S\llbracket A,B\rrbracket
=\frac{1}{2^n}Y\llbracket |A\cap S|,|B\cap S|\rrbracket
\qquad Y\succeq0\,.
\end{equation}
\end{definition}

This restriction identifies
all orbits having the same two intersection parameters
$|A\cap S|$ and $|B\cap S|$. Equivalently, for fixed values of these two
parameters, $X_S\llbracket A,B\rrbracket$ is constant as the remaining orbit
parameters vary. As formalized in \Cref{def:intersection-free-witness}, such a
family is parametrized by a single positive semidefinite matrix
$Y\in\mathbb R^{(k+1)\times(k+1)}$. Note that this enforces the rank of the resulting matrices $\{X_S\}_{S\subseteq[n]}$ to be at most $k+1$.
Since this restricts the feasible set of a
minimization problem,
\[
\operatorname{OPT}(\text{original SDP})
=
\operatorname{OPT}(\text{symmetry-reduced SDP})
\leq
\operatorname{OPT}(\text{intersection-free SDP})\,.
\]
Thus every feasible intersection-free solution gives an upper bound on the
original SDP.

For the more general promise $1\leq |A|\leq k$, we enlarge \Cref{def:intersection-free-witness} to matrices $Y$ that also depend on the size of $A$ and $B$. Their entries are
\begin{equation}
\label{eq:size-dependent-intersection-free-matrices}
Y\llbracket
(|A|,|A\cap S|),(|B|,|B\cap S|)
\rrbracket.
\end{equation}

Both versions of the feasible intersection-free solution $Y\succeq0$
can be turned into a feasible solution $\{X_S\}_{S\subseteq[n]}$ by \Cref{eq:intersection-free-matrices}.

The intersection-free method reduces the space of feasible solutions, in a stronger way than symmetry reduction, thus optimality is, in principle, not preserved. However, given the low rank nature of the solutions, one might hope to
retrieve the analytical solutions from the numerical optimization.

\subsection{Learning Majority on \texorpdfstring{$|A|=k$}{|A| = k}}
\label{sec:majority-known-size}
In this section, for every $k\leq n$, we propose a feasible construction for $h=\text{MAJORITY}_k$, with objective value $O(\log k)$ for odd $k=2m+1$. We use the binary Krawtchouk polynomials, defined as follows.
\begin{definition}[Binary Krawtchouk polynomials]
\label{def:krawtchouk-polynomials}
Let $r\geq0$ and $0\leq t,a\leq r$. The binary Krawtchouk
polynomials are
\[
K_t^{(r)}(a)
=
\sum_{j=0}^t(-1)^j\binom{a}{j}\binom{r-a}{t-j}.
\]
In the convention of~\cite{nomura2012krawtchouk},
$K_t^{(r)}(a)=\binom rt K_t(a;1/2,r)$. We write
$K_t=K_t^{(k)}$. Useful properties of these polynomials are given in
\Cref{lem:krawtchouk-properties}.
\end{definition}

\subsubsection{Rank-one intersection-free solution}

Throughout the rest of the paper, whenever $S\subseteq[n]$ is clear from context, we write
$a:=|A\cap S|$ and $b:=|B\cap S|$, otherwise $a$ and $b$ represent integers of $[k+1]_0$. Within the intersection-free junta SDP of
Definition~\ref{def:intersection-free-witness}, we further suppose that
$Y$ is a rank-one positive-semidefinite matrix. Equivalently,
\[
Y\llbracket a,b\rrbracket
=W\llbracket a\rrbracket W\llbracket b\rrbracket
\qquad W\in\mathbb{R}^{k+1}\,.
\]

\begin{lemma}[Feasibility]
\label{lem:rank-one-intersection-free-structure}
A feasible $W$ of the above form is obtained by taking
\[
W\llbracket a\rrbracket
=\sum_{t=0}^{k}K_t(a)\bigl(p_t+(-1)^tq_t\bigr),
\]
where $p$ and $q$ are polynomials of degree $m$ whose root sets
partition $\{1,\ldots,k-1\}$, normalized so that
$p_0q_0=\frac12$. This construction is feasible for
every $n\geq k$.
\end{lemma}

\begin{proof}

For this construction, separate
$W\llbracket a\rrbracket=L(a)+U(a)$ by taking
$L = W \mathbf{1}_{a\leq m}$ and $U = W \mathbf{1}_{a > m}$.
To verify feasibility, take $a,b$ such that $h(a)\neq h(b)$. Then one of
$a,b$ lies above $m$ and the other lies below, so the $L(a)L(b)$ and
$U(a)U(b)$ terms vanish. By symmetry between $A$ and $B$, the
feasibility condition becomes
\begin{equation*}
\begin{aligned}
\sum_{S:\,f_A(S)\neq f_B(S)}X_S\llbracket A,B\rrbracket
&=\mathbb{E}_S\left[
Y\llbracket |A\cap S|,|B\cap S|\rrbracket
\mathbf{1}_{\{f_A(S)\neq f_B(S)\}}
\right]\\
&=\mathbb{E}_S\left[
W\llbracket a\rrbracket W\llbracket b\rrbracket
\mathbf{1}_{\{f_A(S)\neq f_B(S)\}}
\right]\\
&=\mathbb{E}_S[U(a)L(b)+L(a)U(b)]\\
&=2\mathbb{E}_S[U(a)L(b)]
=1.
\end{aligned}
\end{equation*}
This feasibility condition is best expressed using the Krawtchouk
polynomials introduced in \Cref{def:krawtchouk-polynomials}. Without loss of generality, we can take $L$ and $U$ in the
Krawtchouk basis and write
\[
L(a)=\sum_{t=0}^{k}K_t(a)p_t,
\qquad
U(a)=\sum_{t=0}^{k}K_t(k-a)q_t\,.
\]
By \Cref{lem:krawtchouk-properties}\textup{(i)},
$
W\llbracket a\rrbracket=L(a)+U(a)
=\sum_{t=0}^{k}K_t(a)\bigl(p_t+(-1)^tq_t\bigr)
$. Using \Cref{lem:krawtchouk-properties}\textup{(i)} and
\Cref{lem:krawtchouk-properties}\textup{(v)}, the feasibility
condition becomes
\begin{align}
\mathbb{E}_S[L(a)U(b)]
&=
\sum_{s,t=0}^{k}
p_sq_t
\mathbb{E}_S
\left[
K_s(a)K_t(k-b)
\right] \notag\\
&=
\sum_{t=0}^{k}
(-1)^t
\binom{|A\cap B|}{t}
p_tq_t\,.
\label{eq:LU-expectation}
\end{align}
We impose the following sufficient condition for every
$0\leq d\leq k-1$:
\begin{equation}
2\sum_{t=0}^{k}(-1)^t\binom{d}{t}p_tq_t=1\,.
\label{eq:rank-one-feasibility}
\end{equation}
Indeed, $|A\cap B|\leq k-1$ for any distinct
$k$-subsets, therefore $d$ is ranging over all possible intersection sizes $0\leq |A\cap B|\leq k-1$. 
\Cref{eq:LU-expectation} then shows that these conditions imply all
adversary feasibility constraints for every $n\geq k$.
Binomial inversion in
\Cref{eq:rank-one-feasibility} gives
$p_0q_0=\frac12$ and $p_1q_1=\cdots=p_{k-1}q_{k-1}=0$. When seen as a polynomial,
$\deg (pq)\leq2m=k-1$, therefore it can be written as 
\[
p_tq_t=\frac{1}{2(k-1)!}\prod_{j=1}^{k-1}(t-j)
\]
and hence $p_kq_k=\frac12$. Therefore
\begin{equation}
\label{eq:pq-vanishing}
p_tq_t=
\begin{cases}
\frac12&t\in\{0,k\}\\
0&1\leq t\leq k-1
\end{cases}\,.
\end{equation}

The threshold support of $L$ and $U$ implies
that both polynomials have degree at most $m$, by \Cref{par:krawtchouk-threshold-expansions}. Since there are
$k-1=2m$ interior points, $p$ and $q$ must each have exactly $m$
roots, and their root sets must partition $\{1,\ldots,k-1\}$.
Conversely, every root partition of this kind, with the normalization
$p_0q_0=\frac12$, gives polynomials satisfying
\Cref{eq:pq-vanishing}. 
Define
\[
L(a)=\sum_{t=0}^{k}K_t(a)p_t
\qquad
U(a)=\sum_{t=0}^{k}K_t(k-a)q_t\,.
\]
By Krawtchouk duality and polynomial orthogonality,
\Cref{lem:krawtchouk-properties}\textup{(ii),(iv)},
\begin{equation}
\label{eq:LU-support}
L(a)=0\quad(a>m)
\qquad
U(a)=0\quad(a\leq m)
\end{equation}
because $\deg p,\deg q\leq m$ and $k-a>m$ whenever $a\leq m$.
Hence the support separation used above also holds in the converse
direction, and the resulting $W=L+U$ is feasible.
\end{proof}

\begin{lemma}[Objective value]
\label{lem:rank-one-objective}
The objective value of the construction in
\Cref{lem:rank-one-intersection-free-structure} is
\[
\max_A\sum_{S\subseteq[n]}X_S\llbracket A,A\rrbracket
=\sum_{t=0}^{k}\binom{k}{t}\left(p_t^2+q_t^2\right)\,.
\]
\end{lemma}

\begin{proof}
A feasible solution from \Cref{lem:rank-one-intersection-free-structure} has objective value
\begin{align*}
\max_A\sum_{S\subseteq[n]}X_S\llbracket A,A\rrbracket &= \max_A \sum_{S\subseteq[n]} \frac{1}{2^n}W\llbracket|A\cap S|\rrbracket^2\\
&= \max_A \mathbb{E}_S\left[W\llbracket a\rrbracket ^2\right]\\
&= \mathbb{E}_S\left[W\llbracket a\rrbracket^2\right]
\end{align*}

The last equality follows from the fact that for any $A\in\binom{[n]}{k}$, the random variable
$a=|A\cap S|$ has
the same $\operatorname{Bin}(k,1/2)$ distribution (for uniformly random $S\subseteq[n]$). 
Since $L$ and $U$ have disjoint supports
$
W\llbracket a\rrbracket^2=L(a)^2+U(a)^2
$, and using \Cref{lem:krawtchouk-properties}\textup{(iii)},
\begin{align*}
\mathbb{E}_S[L(a)^2]
&=
2^{-k}\sum_{a=0}^{k}\binom{k}{a}
\left(\sum_{s=0}^{k}K_s(a)p_s\right)
\left(\sum_{t=0}^{k}K_t(a)p_t\right)\\
&=\sum_{t=0}^{k}\binom{k}{t}p_t^2\,.
\end{align*}
Similarly, $\mathbb{E}_S[U(a)^2]=\sum_{t=0}^{k}\binom{k}{t}q_t^2$.
Summing the two identities proves the claim.
\end{proof}

Now that we have reduced the problem of finding a feasible intersection-free and rank-one solution to \Cref{eq:junta-dual-adversary},
to finding polynomials $p$ and $q$ of degree $m$ that satisfy \Cref{lem:rank-one-intersection-free-structure}, we need to find good candidates
with respect to the quantity in \Cref{lem:rank-one-objective}. We define a special instance of such polynomials, that we call \emph{mod-4 polynomials}.

\begin{definition}[The mod-$4$ polynomials]
\label{def:mod4}
Let $k=2m+1$ be odd. Define
\begin{align*}
p(t)
&:=\prod_{\substack{r=0,1\pmod 4\\1\leq r\leq k-1}}
(t-r),
&
q(t)
&:=\prod_{\substack{r=2,3\pmod 4\\1\leq r\leq k-1}}
(t-r).
\end{align*}
Their normalized versions are
\[
P_k(t):=\frac{p(t)}{\sqrt2\,p(0)},
\qquad
Q_k(t):=\frac{q(t)}{\sqrt2\,q(0)}.
\]
\end{definition}

We introduce the following lemma, that will be useful for proving the $O(\log k)$ for this construction, and also the general construction in \Cref{lem:unknown-objective}.

\begin{lemma}
\label{lem:majority-asymptotic}
Let $k\geq3$ be odd, and let $P_k,Q_k$ be as defined above. There
is a universal constant
$C>0$ such that, for every $0\leq t\leq k$,
\[
\binom{k}{t}\bigl(P_k(t)^2+Q_k(t)^2\bigr)
\leq C\frac{k}{(t+1)(k-t+1)}.
\]
\end{lemma}
\begin{proof}
We give the argument for $k\equiv1\pmod4$; the case
$k\equiv3\pmod4$ is identical after exchanging the two endpoint
patterns. Write $k=4j+1$. The factors in $p$
can be paired as $(4i+1,4i+4)$, so
\[
p(t)
=\prod_{i=0}^{j-1}a_i(a_i+3),
\qquad a_i:=t-(4i+4).
\]
We use
\begin{equation}
\label{eq:four-factor-pairing}
(a(a+3))^2\leq a(a+1)(a+2)(a+3),
\end{equation}
which holds for $a\geq0$ or $a\leq-3$. If $p(t)\neq0$, the only
exceptional integer values are $a=-2,-1$, and exactly one of them can
appear. For this exceptional factor, $(a(a+3))^2=4$, while the product
of the other three absolute factors is $2$. Thus, applying
\Cref{eq:four-factor-pairing} to all remaining factors yields
\begin{equation}
\label{eq:short-p-bound}
p(t)^2
\leq4(t-1)!(k-t-1)!.
\end{equation}
For $q$, the interior roots can similarly be paired as $(4i+3,4i+6)$, leaving only the two boundary factors \(t-2\) and \(t-(k-2)\) unpaired. Applying the same argument to the paired factors gives
\begin{equation}
\label{eq:short-q-bound}
q(t)^2
\leq4(t-1)!(k-t-1)!.
\end{equation}
Indeed, the paired factors are treated exactly as above, while the two
remaining factors $t-2$ and $t-(k-2)$ satisfy
\[
(t-2)^2\leq |t-1||t-2|,
\qquad
(t-(k-2))^2\leq |t-(k-2)||t-(k-1)|
\]
whenever $t\geq3$ and $t\leq k-3$, respectively. At $t=1$ and
$t=k-1$ the corresponding remaining factor has square $1$, while at
$t=2$ or $t=k-2$ we have $q(t)=0$.

For $k\equiv3\pmod4$, write $k=4j+3$. The same argument applies, except
that $p$ has the factor $t-(k-2)$ remaining while for $q$ it is $t-2$. This gives \Cref{eq:short-p-bound,eq:short-q-bound}.

Applying the pairing at $t=0$, together with the same boundary
estimates, gives
\[
p(0)^2,q(0)^2\leq4(k-1)!.
\]
Since $|p(0)q(0)|=(k-1)!$, it follows that
\[
\frac{(k-1)!}{4}\leq p(0)^2,q(0)^2\leq4(k-1)!.
\]
Thus \Cref{eq:short-p-bound,eq:short-q-bound} imply
\[
P_k(t)^2+Q_k(t)^2
\leq
C\frac{(t-1)!(k-t-1)!}{(k-1)!},
\qquad 1\leq t\leq k-1.
\]
Multiplying by $\binom{k}{t}$ gives
\[
\binom{k}{t}\bigl(P_k(t)^2+Q_k(t)^2\bigr)
\leq C\frac{k}{t(k-t)}
\leq C'\frac{k}{(t+1)(k-t+1)}
\]
for $1\leq t\leq k-1$.
At $t=0$, the definitions give
$P_k(0)^2+Q_k(0)^2=1$. If $k\equiv1\pmod4$, then

\begin{equation}
\label{eq:PQ-0k-equality}
    P_k(k)=P_k(0)=\frac{1}{\sqrt2} \quad \text{and} \quad Q_k(k)=Q_k(0)=\frac{1}{\sqrt2},
\end{equation}
since reflection $u\mapsto k-u$ preserves the root set of each polynomial. Hence
$P_k(k)^2+Q_k(k)^2=1$. If $k\equiv3\pmod4$, reflection
$u\mapsto k-u$ exchanges the root sets of $p$ and $q$, so
$|p(k)|=|q(0)|$ and $|q(k)|=|p(0)|$. Thus the bounds on
$p(0)$ and $q(0)$ above give
$P_k(k)^2+Q_k(k)^2=O(1)$. Finally, for $t\in\{0,k\}$,
$\binom{k}{t}=1$, which proves the lemma for the endpoints as
$\frac{k}{k+1}=\Theta(1)$.
\end{proof}

\subsection{Learning Majority on \texorpdfstring{$1\leq |A|\leq k$}{1 <= |A| <= k}}
\label{sec:majority-unknown-size}

In this section, for every $k\leq n$, we generalize the previous problem for any subset $A$ of size \emph{at most}~$k$, instead of the previous promise of size \emph{exactly}~$k$. 
The dual adversary SDP of this problem is the same as
\Cref{eq:junta-dual-adversary}, with $A$ and $B$ being any subsets such that $1 \leq |A|,|B| \leq k$. 
Since the hidden subsets are of unknown size, we use the
intersection-free framework introduced in
\Cref{eq:size-dependent-intersection-free-matrices}. Now we want to use the previous construction, and tailor it to the general promise.

For each $m\geq0$, set $P_m:=P_{2m+1}$ and $Q_m:=Q_{2m+1}$,
and let $W_m$ be the vector from
\Cref{lem:rank-one-intersection-free-structure} with
$p_t=P_m(t)$ and $q_t=Q_m(t)$. For
$|A| = 2m+1$ and $|B| = 2\ell+1$, our first idea is to use
\[
Y\llbracket(|A|,|A\cap S|),(|B|,|B\cap S|)\rrbracket
=W_m\llbracket |A\cap S|\rrbracket
W_\ell\llbracket |B\cap S|\rrbracket\,.
\]
However, in the previous $|A|=|B|=k$ case, every $A\neq B$ satisfied $|A\cap B| \leq k-1$. 
Here however, for $A \subsetneq B$, we have $m<\ell$ and
$|A \cap B| = |A| = 2m+1$, therefore the cancellation in
\Cref{eq:pq-vanishing} is no longer guaranteed. Indeed, by the same argument as in \Cref{eq:LU-expectation}, the feasibility constraint becomes
\begin{equation}
\label{eq:unknown-naive-feasibility}
\sum_{t=0}^{2m+1}(-1)^t\binom{2m+1}{t}
\bigl(P_m(t)Q_\ell(t)+Q_m(t)P_\ell(t)\bigr)=1\,.
\end{equation}
The $t=0$ term is $1$. For $1\leq t\leq2m$, the special mod $4$ construction in~\Cref{def:mod4} gives $P_m(t)=P_\ell(t)=0$ when $t=0,1\pmod4$, and
$Q_m(t)=Q_\ell(t)=0$ when $t=2,3\pmod4$, making the summand equal to zero. At $t=2m+1$, however, exactly one of the two terms in the summand is nonzero: if $m$ is even, $P_\ell(2m+1)=0$, while if $m$ is odd, $Q_\ell(2m+1)=0$. Depending on $m$, the sum reduces to
$1-\delta_{m\ell}$, where
\begin{equation}
\label{eq:unknown-endpoint-defect}
\delta_{m\ell}=
\begin{cases}
P_m(2m+1)Q_\ell(2m+1)&m\text{ even}\\
Q_m(2m+1)P_\ell(2m+1)&m\text{ odd}
\end{cases}
\end{equation}
Since $\delta_{m\ell}\neq0$, feasibility fails. We propose the following way of removing this unwanted defect.

\subsubsection{Engineering new polynomials}
\label{sec:new-pol}

Fix the layers $0\leq m\leq M$, where $k=2M+1$, and define the $(M+1)$-dimensional Hilbert space
\[
\mathcal H=\operatorname{span}_{\mathbb R}
\{e_*,e_0,\ldots,e_{M-1}\}
\]
where the vectors $\{e_*,e_0,\dots,e_{M-1}\}$ form an orthonormal basis. For each
$\ell$, we define $\mathcal H$-valued polynomials
$p_\ell,q_\ell$. Their coefficients on $e_*$ are
\[
\langle p_\ell(t),e_*\rangle=P_\ell(t)
\qquad
\langle q_\ell(t),e_*\rangle=Q_\ell(t)\,.
\]
Thus the $e_*$ coordinate contains the polynomials defined in \Cref{def:mod4}, and the
feasibility \Cref{eq:unknown-naive-feasibility} gives the
defect $\delta_{m\ell}$ in
\Cref{eq:unknown-endpoint-defect}. 
We use the other coordinates to correct that defect. Fix $m<\ell$ and define
\begin{equation}
\label{eq:unknown-Cml}
C_{m\ell}(t)
=
\frac{t}{t-(2m+2)}
\begin{cases}
Q_\ell(t)&m\text{ even}\\
P_\ell(t)&m\text{ odd}
\end{cases}\,.
\end{equation}
This definition is valid because $2m+2\leq2\ell$ and the mod-$4$ root
pattern makes $2m+2$ a root of $Q_\ell$ when $m$ is even and of
$P_\ell$ when $m$ is odd. Thus $t-(2m+2)$ divides the chosen polynomial. 
Following their definition in \Cref{def:mod4}, the roots of $P_{\ell}$ and $Q_{\ell}$ all have multiplicity one. Therefore dividing by $t-(2m+2)$ removes the zero at $2m+2$\, and
multiplying by $t$ introduces a zero at $0$. 
Hence $C_{m\ell}$ has the
same degree and all the same remaining roots. This can be viewed as
\[
\begin{array}{ccc}
C_{m\ell}
&&
\begin{cases}
Q_\ell&m\text{ even}\\
P_\ell&m\text{ odd}
\end{cases}
\\[4pt]
\underbrace{0}_{\text{new root}}
&\longleftrightarrow&
\underbrace{2m+2}_{\text{old root}}
\end{array}\,.
\]
For $0\leq j<M$ and $0\leq \ell \leq M$, define the remaining coordinates by
\begin{equation}
\label{eq:unknown-vector-polynomials}
\begin{aligned}
\langle p_\ell(t),e_j\rangle
&=
\begin{cases}
\dfrac{C_{j\ell}(t)}{2j+1} & j<\ell\text{ and }j\text{ odd},\\[5pt]
P_\ell(t) & j=\ell\text{ and } j \text{ even},\\
0 & \text{otherwise},
\end{cases}
\qquad
\langle q_\ell(t),e_j\rangle
&=
\begin{cases}
\dfrac{C_{j\ell}(t)}{2j+1} & j<\ell\text{ and }j\text{ even},\\[5pt]
Q_\ell(t) & j=\ell\text{ and } j \text{ odd},\\
0 & \text{otherwise}.
\end{cases}
\end{aligned}
\end{equation}

\begin{lemma}[Feasibility]
\label{lem:unknown-feasibility}
Let $p_m,q_m$ be the polynomials in
\Cref{eq:unknown-vector-polynomials}, and set
\[
W_m(a)
=
\sum_{t=0}^{2m+1}K_t^{(2m+1)}(a)
\bigl(p_m(t)+(-1)^tq_m(t)\bigr)\,.
\]
For odd-sized hidden sets $|A|=2m+1$ and $|B|=2\ell+1$ of size at most
$k$, define the matrices
\[
X_S\llbracket A,B\rrbracket
=
\frac{1}{2^n}\left\langle
W_m(|A\cap S|),W_\ell(|B\cap S|)
\right\rangle\,.
\]
Then $\{X_S\}_{S\subseteq[n]}$ is a feasible solution of
\Cref{eq:junta-dual-adversary}.
\end{lemma}

\begin{proof}
Each $X_S$ is positive semidefinite because it is a Gram matrix.
Fix $m<\ell$. For $j<m$, the vector $p_m(t)$ can
have a nonzero $e_j$-coordinate only when $j$ is odd, whereas
$q_\ell(t)$ can have a nonzero $e_j$-coordinate only when $j$ is even.
For $j>m$, the vector $p_m(t)$ has no $e_j$-component.  Hence the only
possible common coordinates of $p_m(t)$ and $q_\ell(t)$ are $e_*$ and,
when $m$ is even, $e_m$.  Therefore
\begin{equation*}
\langle p_m(t),q_\ell(t)\rangle
=
P_m(t)Q_\ell(t)
+
\mathbf 1_{\{m\ {\rm even}\}}
P_m(t)\frac{C_{m\ell}(t)}{2m+1}\,.
\end{equation*}
Similarly, the only possible common coordinates of $q_m(t)$ and
$p_\ell(t)$ are $e_*$ and, when $m$ is odd, $e_m$, so
\begin{equation*}
\langle q_m(t),p_\ell(t)\rangle
=
Q_m(t)P_\ell(t)
+
\mathbf 1_{\{m\ {\rm odd}\}}
Q_m(t)\frac{C_{m\ell}(t)}{2m+1}\,.
\end{equation*}
Thus, writing
\begin{equation*}
c_{m\ell}(t)
:=
\langle p_m(t),q_\ell(t)\rangle
+
\langle q_m(t),p_\ell(t)\rangle
\end{equation*}
we have, for $m<\ell$,
\begin{equation}
c_{m\ell}(t)
=
P_m(t)Q_\ell(t)+Q_m(t)P_\ell(t)
+
\begin{cases}
\displaystyle
P_m(t)\frac{C_{m\ell}(t)}{2m+1}
& m\text{ even},\\[8pt]
\displaystyle
Q_m(t)\frac{C_{m\ell}(t)}{2m+1}
& m\text{ odd}.
\end{cases}\,.
\label{eq:unknown-cml-piecewise}
\end{equation}
In the case $m=\ell$, coordinates outside $e_*$ of $p_m(t)$ and $q_m(t)$
have disjoint support, and hence
\begin{equation}
\label{eq:unknown-cmm}
c_{mm}(t)=2P_m(t)Q_m(t)\,.
\end{equation}

These identities are exactly what is needed for feasibility.  First,
$C_{m\ell}(0)=0$ and
$P_r(0)=Q_r(0)=1/\sqrt{2}$ for every layer $r$, so
\begin{equation}
c_{m\ell}(0)=1\,.
\end{equation}
For $m<\ell$ and $1\leq t\leq 2m$, if $t=0,1\pmod4$, then
$P_m(t)=P_\ell(t)=0$, whereas if $t=2,3\pmod4$, then
$Q_m(t)=Q_\ell(t)=0$. Thus both $e_*$-coordinate products in
\Cref{eq:unknown-cml-piecewise} vanish. Moreover, for $1\leq t\leq 2m$,
$C_{m\ell}$ has the same roots as $Q_\ell$ when $m$ is even and as
$P_\ell$ when $m$ is odd, because its only replaced root is
$2m+2$. Hence
$P_m(t)C_{m\ell}(t)/(2m+1)=0$ when $m$ is even, and
$Q_m(t)C_{m\ell}(t)/(2m+1)=0$ when $m$ is odd.
At the remaining endpoint $t=2m+1$, the definition of $C_{m\ell}$ gives
\[
\frac{C_{m\ell}(2m+1)}{2m+1}
=
\begin{cases}
-Q_\ell(2m+1)&m\text{ even}\\
-P_\ell(2m+1)&m\text{ odd}
\end{cases}\,.
\]
The unique nonzero base term is the endpoint defect in
\Cref{eq:unknown-endpoint-defect}.  The correction term is the final
term in \Cref{eq:unknown-cml-piecewise}.  Thus, if $m$ is even, the
endpoint contribution to $c_{m\ell}$ is
\begin{align*}
P_m(2m+1)Q_\ell(2m+1)
&\quad+P_m(2m+1)\frac{C_{m\ell}(2m+1)}{2m+1}
\\
&=
P_m(2m+1)Q_\ell(2m+1)
-P_m(2m+1)Q_\ell(2m+1)
\\
&=0\,.
\end{align*}
If $m$ is odd, it is instead
\begin{align*}
Q_m(2m+1)P_\ell(2m+1)
&\quad+Q_m(2m+1)\frac{C_{m\ell}(2m+1)}{2m+1}
\\
&=
Q_m(2m+1)P_\ell(2m+1)
-Q_m(2m+1)P_\ell(2m+1)
\\
&=0\,.
\end{align*}
Hence
\begin{equation}
\label{eq:unknown-cross-vanishing}
c_{m\ell}(t)=0
\qquad
1\leq t\leq 2m+1
\qquad m<\ell.
\end{equation}
For $m=\ell$, \Cref{eq:unknown-cmm} and the identity
$P_m(t)Q_m(t)=0$ for $1\leq t\leq2m$, which is
\Cref{eq:pq-vanishing} applied to the polynomials from~\Cref{def:mod4}, give
\begin{equation}
\label{eq:cmm0}
c_{mm}(t)=0
\qquad 1\leq t\leq 2m\,.
\end{equation}
Now write
\begin{align*}
L_m(a)&=\sum_{t=0}^{2m+1}K_t^{(2m+1)}(a)p_m(t),\\
U_m(a)&=\sum_{t=0}^{2m+1}K_t^{(2m+1)}(a)(-1)^tq_m(t),\\
W_m(a)&=L_m(a)+U_m(a).
\end{align*}
The degree bound together with~\Cref{lem:krawtchouk-properties}\textup{(ii),(iv)} implies that
$L_m(a)=0$ for $a>m$. Applying the same argument to $U_m$ and using
\Cref{lem:krawtchouk-properties}\textup{(i)} gives
$U_m(a)=0$ for $a\leq m$.
Expanding $\langle W_m(a),W_\ell(b)\rangle$ gives four terms. On a
distinguishing query, $f_A(S)\ne f_B(S)$, the disjoint support makes
$\langle L_m(a),L_\ell(b)\rangle$ and
$\langle U_m(a),U_\ell(b)\rangle$ vanish. Hence only
$\langle L_m(a),U_\ell(b)\rangle$ and
$\langle U_m(a),L_\ell(b)\rangle$ survive. Therefore
\begin{align*}
&\mathbb E_S\!\left[
\langle W_m(a),W_\ell(b)\rangle
\mathbf 1_{\{f_A(S)\neq f_B(S)\}}
\right]
\notag\\
&\qquad=
\mathbb E_S\!\left[
\langle L_m(a),U_\ell(b)\rangle
+
\langle U_m(a),L_\ell(b)\rangle
\right]
\notag\\
&\qquad=
\sum_{t=0}^{|A\cap B|}
(-1)^t\binom{|A\cap B|}{t}
\left(
\langle p_m(t),q_\ell(t)\rangle
+
\langle q_m(t),p_\ell(t)\rangle
\right)
\notag\\
&\qquad=
\sum_{t=0}^{|A\cap B|}
(-1)^t\binom{|A\cap B|}{t}c_{m\ell}(t)
=1\,.
\end{align*}
Indeed, since $c_{m\ell}(t)=c_{\ell m}(t)$, when $m\neq\ell$ we may assume $m<\ell$. Then
$|A\cap B|\leq 2m+1$ and
\Cref{eq:unknown-cross-vanishing} removes every term with $t\geq1$.
If $m=\ell$, because $A\neq B$, we have
$|A\cap B|\leq2m$, and the
same conclusion follows from \Cref{eq:cmm0}.  Thus only the
$t=0$ term remains, and it equals $c_{m\ell}(0)=1$, which proves the
feasibility constraint.
\end{proof}

\begin{lemma}[Objective value]
\label{lem:unknown-objective}
The construction in \Cref{lem:unknown-feasibility} has objective value
$O(\log k)$.
\end{lemma}

\begin{proof}
Since $L_{\ell}$ and $U_{\ell}$ have disjoint support, we have $\|W_\ell(a)\|^2=\|L_\ell(a)\|^2+\|U_\ell(a)\|^2$. Together with the orthogonality identity in
\Cref{lem:krawtchouk-properties}\textup{(iii)}, the objective on
$\ell$ is
\begin{equation*}
T_\ell=\sum_{t=0}^{2\ell+1}\binom{2\ell+1}{t}
\bigl(\|p_\ell(t)\|^2+\|q_\ell(t)\|^2\bigr)\,.
\end{equation*}
Expanding the coordinates in \Cref{eq:unknown-vector-polynomials} gives
\begin{align}
T_\ell
&=
\underbrace{
\sum_{t=0}^{2\ell+1}\binom{2\ell+1}{t}
\left(P_\ell(t)^2+Q_\ell(t)^2\right)
}_{O(\log k)\text{ by Lemma~\ref{lem:majority-asymptotic}}}
+
\underbrace{
\mathbf 1_{\{\ell<M\}}
\sum_{t=0}^{2\ell+1}\binom{2\ell+1}{t}
\left(
\mathbf 1_{\{\ell\text{ even}\}}P_\ell(t)^2
+\mathbf 1_{\{\ell\text{ odd}\}}Q_\ell(t)^2
\right)
}_{O(\log k)\text{ by Lemma~\ref{lem:majority-asymptotic}}}
\notag\\
&\quad+
\underbrace{
\sum_{m<\ell}\sum_{t=0}^{2\ell+1}\binom{2\ell+1}{t}
\left(\frac{C_{m\ell}(t)}{2m+1}\right)^2
}_{(*)}\,.
\label{eq:unknown-layer-cost-decomposition}
\end{align}

For every $0\leq\ell\leq M$, the term $(*)$ in
\Cref{eq:unknown-layer-cost-decomposition} satisfies
\begin{equation}
\label{eq:unknown-correction-energy}
\sum_{m<\ell}\sum_{t=0}^{2\ell+1}\binom{2\ell+1}{t}
\left(\frac{C_{m\ell}(t)}{2m+1}\right)^2
=O(\log k)\,.
\end{equation}
This is proved in~\Cref{app:majority-unknown-upper-bound}.
Since the overall objective is the maximum of the $T_\ell$, it is
$O(\log k)$.
\end{proof}

\paragraph{Generalizing to the whole $1\leq |A| \leq k$ promise}
We have proven the $O(\log k)$ upper bound for odd support sizes $|A|$, we now generalize to even support sizes. Introduce one known auxiliary variable
$y\in \{-1,1\}$ and define
\begin{equation}
\label{eq:unknown-parity-lift}
g_A(x,y)
=
\begin{cases}
f_A(x)&y=1\\
-f_A(\bar x)&y=-1
\end{cases}\,.
\end{equation}
If $|A|$ is odd,
$-\text{MAJORITY}_{|A|}(\bar x)=\text{MAJORITY}_{|A|}(x)$, so $g_A$ is
the same Majority-junta and $y$ is irrelevant.  If $|A|=2m$ is even,
then
\[
g_A(x,y)=\mathrm{MAJORITY}_{2m+1}(x_A,y)\,.
\]
One query to $g_A$ is simulated coherently by one query to $f_A$. The set that we need to learn is
\[
B_A
=
\begin{cases}
A&|A|\text{ odd}\\
A\cup\{y\}&|A|\text{ even}
\end{cases}
\]
which always has odd size and is at most $k+1$. Applying our algorithm from the feasible SDP solution uses $O(\log k)$ queries, and recovers $A$.

\begin{theorem}[Upper bound for learning Majority]
\label{thm:majority-upper-bound}
For every $2\leq k\leq n$, under the promises $1\leq |A|\leq k$ or
$|A|=k$,
\[
Q=O(\log k)\,.
\]
The case $k=1$ has constant query complexity.
\end{theorem}

\subsection{A matching lower bound}

We assume that $k=2m$ is even. Indeed, if $k$ is odd, the promise
$|A|\leq k$ contains the subproblem $|A|\leq k-1$, giving the same $\Omega(\log(k-1)) = \Omega(\log k)$ lower bound.
We construct a feasible solution to \Cref{eq:junta-primal-adversary}, restricted to the following promise:
\[
\mathcal A:=\{A\subseteq U:|A|=m\} \subseteq \{A:1\leq |A| \leq k\}\,.
\]
Here $U\subseteq[n]$ is a known subset of size $k$. Define on this promise the following adversary matrix 
\begin{equation}
\label{eq:balanced-gamma}
\Gamma\llbracket A,B\rrbracket
=
\begin{cases}
\displaystyle
\frac{1}{|A\setminus B|\binom{m}{|A\setminus B|}^{2}}&A\ne B\\[6pt]
0&A=B
\end{cases}\,.
\end{equation}

\begin{lemma}
\label{lem:balanced-gamma-norm}
The adversary matrix in \Cref{eq:balanced-gamma} satisfies
\[
\norm{\Gamma}
=\sum_{s=1}^{m}\frac1s
=\Theta(\log k)\,.
\]
\end{lemma}

\begin{proof}[Proof of \Cref{lem:balanced-gamma-norm}]
For fixed $A\in\mathcal A$, exactly $\binom ms^2$ sets $B$ satisfy
$|A\setminus B|=s$. Hence every row sum of $\Gamma$ is
\[
\sum_{s=1}^{m}\frac1s = \Theta(\log k)\,.
\]
Indeed, multiplying $\Gamma$ by the all-ones vector gives $\Gamma \mathbf{1} = \left(\sum_{s=1}^{m}\frac1s\right)\mathbf{1}$. Therefore $\norm\Gamma \geq
\sum_{s=1}^{m}\frac1s$. Similarly, because $\Gamma$ has nonnegative entries, its infinity norm is the maximum of the sums of each row, which equals $
\norm{\Gamma}_{\infty}
 = \sum_{s=1}^{m}\frac1s
$. Because $\Gamma$ is symmetric, its induced $1$-norm is also $\norm{\Gamma}_1 = \sum_{s=1}^{m}\frac1s$. Finally, using $\norm{\Gamma} \leq \sqrt{\norm{\Gamma}_{1}\norm{\Gamma}_{\infty}}$ we have the desired result.
\end{proof}

\begin{lemma}
\label{lem:balanced-mask-norm}
For every $S\subseteq [n]$
\[
\norm{\Gamma\circ\Delta_S} = O(1)\,.
\]
\end{lemma}

The proof of \Cref{lem:balanced-mask-norm} is given in~\Cref{app:balanced-slice-lower}.

\begin{theorem}[Lower bound under the promise \(1\leq |A|\leq k\)]
\label{thm:majority-lower-bounds}
%\label{thm:leqk}
For every \(n\geq k\geq2\),
\[
Q=\Omega(\log k)\,.
\]
\end{theorem}

\begin{proof}
The adversary bound
\cite{ambainis2002quantum,HoyerLeeSpalek2007} and
\Cref{lem:balanced-gamma-norm,lem:balanced-mask-norm} give
\[
Q
=\Omega\!\left(
\frac{\norm{\Gamma}}
{\max_S\norm{\Gamma\circ\Delta_S}}
\right)
=\Omega(\log k)\,.
\]
\end{proof}

Combining~\Cref{thm:majority-upper-bound,thm:majority-lower-bounds} (upper and lower bounds)
proves part~\textup{(ii)} of~\Cref{thm:majority-query-complexity} which states that the quantum query complexity of learning majority on the promise $1\leq|A|\leq k$ is $\Theta(\log k)$ for any $n\geq k\geq2$.

\par
On the promise $|A|=k\geq 2$, where $n\geq 2k$, combining the same $O(\log k)$ upper bound with the lower bound~\Cref{thm:eqk}, proves
part~\textup{(i)}.

\subsection{Time-efficient implementation}

In this section, we propose a time-efficient implementation of the algorithm based on feasible adversary solutions constructed in~\Cref{sec:majority-known-size,sec:majority-unknown-size}. We first turn the feasible solution into a \emph{transducer}, as defined in~\cite{belovs2023taming}. The transducer is then converted into a bounded-error quantum algorithm whose query complexity, for constant error $\eps$, is the same as the objective value of the dual feasible solution (up to constant factors).

\begin{definition}[Transducers]
Let $\mathcal V$ and $\mathcal Q$ be finite-dimensional Hilbert spaces. A
transducer is a unitary $\mathsf S$ acting on their direct sum
$\mathcal V\oplus\mathcal Q$. We say that $\mathsf S$ transduces an initial
state $\ket\xi\in\mathcal V$ into an output state $\ket\tau\in\mathcal V$ if
there exists a vector $\ket v\in\mathcal Q$, called a \emph{catalyst}, such that
\begin{equation*}
\mathsf S\bigl(\ket\xi\oplus\ket v\bigr)
=\ket\tau\oplus\ket v.
\end{equation*}
\end{definition}

\subsubsection{Transducer from the feasible adversary solution.}

Let $\{X_S\}_{S\subseteq[n]}$ be a feasible solution of \Cref{eq:junta-dual-adversary}.  Since $X_S\succeq0$, we represent those as Gram matrices over a common Hilbert space $\mathcal W$,
\begin{equation*}
    X_S\llbracket A,B\rrbracket
    =\braket{w_{A,S}}{w_{B,S}},
    \qquad \ket{w_{A,S}}\in\mathcal W.
\end{equation*}
Note that its dimension can be taken as the maximum value of $\operatorname{rank}(X_S)$. When using the rank-one intersection-free solutions defined in \Cref{def:intersection-free-witness}, we associate $\ket{w_{A,S}}$ with a scalar.
Let $\mathcal Q=\operatorname{span}\{\ket S:S\subseteq[n]\}$ be the
query space, and write the phase membership oracle as
\begin{equation}
    O_A^{\pm}\ket S=f_A(S)\ket S.
    \label{eq:majority-phase-oracle}
\end{equation}
It is equivalent to the phase oracle defined in
\Cref{eq:phase-oracle}, where for simplicity we use sets $S$ instead of bit strings $x$.
The catalyst associated with $A$ is
\begin{equation}
    \ket{w_A}
    :=\frac{1}{\sqrt2}\sum_{S\subseteq[n]}
    \ket S\ket{w_{A,S}}
    \in\mathcal Q\otimes\mathcal W.
    \label{eq:majority-transducer-catalyst}
\end{equation}
Let $\mathcal V$ be an output space containing the fixed initial state
$\ket{0}$ associated with the empty set $\emptyset$, and mutually orthogonal output states
$\{\ket A:A\in\mathcal A\}$, with $\ket{0}$ orthogonal to every
$\ket A$\footnote{Remember that $\emptyset \notin \mathcal A$ for junta learning, as we do not accept constant functions.}.
For every $A,B\in\mathcal A$, feasibility of
\Cref{eq:junta-dual-adversary} gives
\begin{align}
    &\bra{w_A}
    \left[
        \left(I_{\mathcal Q}-(O_A^{\pm})^\dagger O_B^{\pm}\right)
        \otimes I_{\mathcal W}
    \right]
    \ket{w_B}
    \nonumber\\
    &\qquad=
    \sum_{\substack{S\subseteq[n]\\f_A(S)\neq f_B(S)}}
    X_S\llbracket A,B\rrbracket
    =\mathbf 1_{\{A\neq B\}}
    =1-\braket{A}{B}.
    \label{eq:majority-transducer-gram-identity}
\end{align}
Here the identity also holds for $A=B$, since then its left-hand side
vanishes.  Equation~\eqref{eq:majority-transducer-gram-identity} is equivalent to saying that
the two families
\begin{equation*}
    \ket{0}
    \oplus (O_A^{\pm}\otimes I_{\mathcal W})\ket{w_A},
    \qquad
    \ket A\oplus\ket{w_A},
\end{equation*}
have the same Gram matrix. Indeed, taking $v_A := \ket{0}
    \oplus (O_A^{\pm}\otimes I_{\mathcal W})\ket{w_A}$ and $u_A := \ket A\oplus\ket{w_A}$ and any $A,B \in \mathcal A$, we have 
    \begin{align*}
     \braket{v_A}{v_B}  &= 1 + \bra{w_A} \left[
        \left((O_A^{\pm})^\dagger O_B^{\pm}\right)
        \otimes I_{\mathcal W}
    \right]
    \ket{w_B}\\
    &= \braket{w_A}{w_B} + \braket{A}{B} \qquad \qquad \qquad \qquad  \text{(by \Cref{eq:majority-transducer-gram-identity})} \\ 
    &= \braket{u_A}{u_B}.
    \end{align*}
    
    Consequently, there exists an input-independent unitary $U$ on
$\mathcal V\oplus(\mathcal Q\otimes\mathcal W)$ satisfying
\begin{equation}
    U\left(
        \ket{0}
        \oplus (O_A^{\pm}\otimes I_{\mathcal W})\ket{w_A}
    \right)
    =\ket A\oplus\ket{w_A}
    \qquad(A\in\mathcal A).
    \label{eq:majority-transducer-unitary}
\end{equation}

Thus the following map is a transducer:
\begin{equation}
    \begin{bmatrix}
        \ket{0} \\
        \ket{w_A}
    \end{bmatrix}
    \overset{I_{\mathcal V}\oplus
        (O_A^{\pm}\otimes I_{\mathcal W})}{\longmapsto}
    \begin{bmatrix}
        \ket{0} \\
        (O_A^{\pm}\otimes I_{\mathcal W})\ket{w_A}
    \end{bmatrix}
    \overset{U}{\longmapsto}
    \begin{bmatrix}
        \ket A \\
        \ket{w_A}
    \end{bmatrix}.
    \label{eq:majority-canonical-transducer}
\end{equation}

The catalyst $\ket{w_A}$ is restored exactly. However, this transducer does not solve the junta learning problem immediately, as it requires having the input-dependent catalyst $\ket{w_A}$ as the input state. 

As shown in~\cite{belovs2023taming}, a new transducer can be constructed that achieves the same transduction as \Cref{eq:majority-canonical-transducer}, but with an arbitrarily small catalyst norm. It uses $K$ controlled applications of the same transducer, on the Hilbert space $\mathcal V$ tensored with a clock register $\mathbb{C}^K$. 

\begin{equation}
\label{eq:majority-clock-transducer}
\begin{bmatrix}
    \ket0 \\ 0 \\ \vdots \\ 0 \\ \frac{1}{\sqrt{K}}\ket{w_A}
\end{bmatrix}
\mapsto
\frac{1}{\sqrt{K}}
\begin{bmatrix}
    \ket0 \\ \ket0 \\ \vdots \\ \ket0 \\ \ket{w_A}
\end{bmatrix}
\mapsto
\frac{1}{\sqrt{K}}
\begin{bmatrix}
    \ket A \\ \ket0 \\ \vdots \\ \ket0 \\ \ket{w_A}
\end{bmatrix}
\mapsto\cdots\mapsto
\frac{1}{\sqrt{K}}
\begin{bmatrix}
    \ket A \\ \ket A \\ \vdots \\ \ket A \\ \ket{w_A}
\end{bmatrix}
\mapsto
\begin{bmatrix}
    \ket A \\ 0 \\ \vdots \\ 0 \\ \frac{1}{\sqrt{K}}\ket{w_A}
\end{bmatrix}.
\end{equation}

The first map prepares the uniform superposition over the clock register, the intermediate maps apply one copy of $\mathsf S$ controlled on each clock value, and the final map unprepares the uniform superposition. Running it as an algorithm without the reduced catalyst then introduces an arbitrarily small error. More precisely, the transducer in~\Cref{eq:majority-clock-transducer} has catalyst $\ket{w_A}/\sqrt K$. Taking
\[K=\left\lceil
        \frac{4}{\eps^2}
        \max_{A\in\mathcal A}\norm{\ket{w_A}}^2
    \right\rceil.
\]    
and then omitting the new catalyst gives an output state at a Euclidean distance at
most $\eps$ from $\ket A$.

For our feasible solution, $\max_{A\in\mathcal A}\norm{\ket{w_A}}^2=O(\log k)$, so $K=O(\log k)$ for constant error. Each of the $K$ controlled transducer applications contains one controlled application of $U$; consequently, if $U$ has gate complexity $T$ and its controlled version has complexity $O(T)$, then the resulting algorithm has non-oracle gate complexity $O(KT)=O(T\log k)$ for constant error.

\subsubsection{Gate implementation of \texorpdfstring{$U$}{U} for the promise \texorpdfstring{$|A|=k$}{|A|=k}.}
\label{sec:majority-unitary-implementation}
The goal of this section is to find an $\tilde{O}(n)$-gate implementation of $U$ as defined in~\Cref{eq:majority-canonical-transducer}. We first treat the case where the promise is $|A|=k$ with $k=2m+1$ odd. 

For this promise, we use the feasible solutions defined in \Cref{lem:rank-one-intersection-free-structure}. We associate $\ket{w_{A,S}} \in \mathcal W= \mathbb{C}$ with scalars, and the overall Hilbert space of the transducer is $\mathcal V\oplus \mathcal Q$, where $\mathcal V = \mathcal Q = \mathbb{C}^{2^n}$. Equivalently, as the query and output Hilbert spaces are equal to one another, the unitary $U$ acts on a $\mathbb{C}^2 \otimes \left(\mathbb{C}^2\right)^{\otimes n}$ Hilbert space. The first qubit then indicates the direct-sum component: $\ket0$ labels $\mathcal V$, whereas $\ket1$ labels $\mathcal Q$. 

Let $P_k$ and $Q_k$ be the normalized polynomials defined in \Cref{def:mod4}.
For any input $A \in \mathcal A = \binom{[n]}{k}$, the catalyst~\Cref{eq:majority-transducer-catalyst} before the oracle can be written as
\[
\ket{w_A}
=\frac{1}{\sqrt{2^{n+1}}}
\sum_{S\subseteq[n]}
\sum_{t=0}^{k}
K_t(|A\cap S|)
\left(P_k(t)+(-1)^tQ_k(t)\right)\ket S.
\]

Note that $\sum_{t=0}^{k}
K_t(|A\cap S|)
P_k(t) = L(|A\cap S|)$ has support in $|A\cap S|\leq m$, as shown in \Cref{eq:LU-support}. Since it corresponds to the case where $\operatorname{MAJORITY}_k(|A\cap S|) = -1$, the oracle in \Cref{eq:majority-phase-oracle} introduces a $-1$ phase to this vector, while  $\sum_{t=0}^{k}
K_t(|A\cap S|)
(-1)^t Q_k(t) = U(|A\cap S|)$ is left unchanged. Hence

\[
O_A^{\pm}\ket{w_A}
=\frac{1}{\sqrt{2^{n+1}}}
\sum_{S\subseteq[n]}
\sum_{t=0}^{k}
K_t(|A\cap S|)
\left(-P_k(t)+(-1)^tQ_k(t)\right)\ket S.
\]
We define for simplicity the coefficients before and after the oracle
\begin{equation*}
    b_t:=\frac{P_k(t)+(-1)^tQ_k(t)}{\sqrt2},
    \qquad
    a_t:=\frac{(-1)^tQ_k(t)-P_k(t)}{\sqrt2}
    \qquad (0\leq t\leq k).
\end{equation*}
The unitary $U$ of the
transducer defined in \Cref{eq:majority-transducer-unitary} then maps, for all $A\in \mathcal A$, 

\begin{equation}
\label{eq:majority_transducer_map}
U\left(\ket 0 \oplus \frac{1}{\sqrt{2^n}}
\sum_{S\subseteq[n]}
\sum_{t=0}^{k}
K_t(|A\cap S|)a_t\ket{S}\right) = \ket{A} \oplus \frac{1}{\sqrt{2^n}}
\sum_{S\subseteq[n]}
\sum_{t=0}^{k}
K_t(|A\cap S|)b_t\ket{S}.
\end{equation}

\paragraph{Hadamard transform.} We have the equivalent expression of the Krawtchouk polynomials with Fourier characters as shown in \Cref{eq:kraw-character-proof}:
\begin{equation}
    K_t(|A\cap S|)
    =\sum_{\substack{T\subseteq A\\|T|=t}}\chi_T(S)
    \label{eq:majority-krawtchouk-character-expansion}
\end{equation}

Applying the Hadamard transform $H^{\otimes n}$ maps $\frac{1}{\sqrt{2^n}}\sum_{S\subseteq[n]} \chi_T (S) \ket{S}$ to $\ket{T}$. Using this map on the uniform superposition of Krawtchouk polynomials in \Cref{eq:majority-krawtchouk-character-expansion}, and rearranging the sum, gives us

\begin{equation}
    H^{\otimes n}\left(
        \frac{1}{\sqrt{2^n}}
        \sum_{S\subseteq[n]}K_t(|A\cap S|)\ket S
    \right)
    =\sum_{\substack{T\subseteq A\\|T|=t}}\ket T := \ket{A_t}.
    \label{eq:majority-krawtchouk-hadamard}
\end{equation}

Hence, $H^{\otimes n}$ applied to the catalyst after oracle call is

\begin{equation}
\label{eq:hadamard-catalyst}
    H^{\otimes n}\left(
        O_A^{\pm}\ket{w_A}
    \right)
    =H^{\otimes n}\left(
        \frac{1}{\sqrt{2^n}}
        \sum_{S\subseteq[n]}\sum_{t=0}^{k}
        K_t(|A\cap S|)a_t\ket S
    \right)
    =\sum_{t=0}^{k}a_t\ket{A_t}.
\end{equation}

Here $\ket{A_t}$ is the uniform superposition (not normalized) of all subsets of $A$ of size $t$. Note that $\ket{A_0} = \ket{0}$ and $\ket{A_k} = \ket{A}$ are the states in $\mathcal V$ of interest in \Cref{eq:majority_transducer_map}.

The state after the Hadamard transform is easier to work with, as the problem now is to find $\widetilde{U}$ such that
\begin{equation}
\label{eq:transducer-equation-after-hadamard}
    \widetilde{U} \left( \ket{0} \oplus \sum_{t=0}^{k}a_t\ket{A_t} \right) = \ket{A} \oplus \sum_{t=0}^{k}b_t\ket{A_t}.
\end{equation}

\paragraph{Blocks controlled on the Hamming weight.}
In this section we take the $\mathcal H = \mathbb{C}^2 \otimes \left(\mathbb{C}^2\right)^{\otimes n}$ view. Denote by $D$ the first qubit register, and by $Q$ the $n$-qubit register.
Note that applying $H^{\otimes n}$ on the query space $\mathcal Q$, as was done in \Cref{eq:hadamard-catalyst}, is equivalent to applying $H^{\otimes n}$ controlled on the first qubit being $\ket{1}$. We denote this unitary as $\widehat H := (\ket{0}\bra{0}\otimes I^{\otimes n} + \ket{1}\bra{1} \otimes H^{\otimes n})$. 

The original unitary $U$ defined in~\Cref{eq:majority-transducer-unitary} is then equal to $U := \widehat{H} \widetilde{U} \widehat{H}$ where $\widetilde{U}$ satisfies \Cref{eq:transducer-equation-after-hadamard}.
This relation is best expressed when decomposed onto the Hamming weights of the $\left(\mathbb{C}^2\right)^{\otimes n}$ register.  For $0\leq t\leq n$, let
\begin{equation*}
    \Pi_t
    :=\sum_{\substack{T\subseteq[n]\\|T|=t}}\ket T\!\bra T,
\end{equation*}
be the orthogonal projector on $\mathcal S_t :=\operatorname{Span}\{\ket{T}:|T|=t\}$.  
We now study the case where $\widetilde U$ is of the form:
\begin{equation*}
    \widetilde U
    =\sum_{t=0}^{n}U_t\otimes\Pi_t,
\end{equation*}
where $U_t \in U(2)$ for each $0\leq t \leq n$. Note that $\widetilde U$ only acts non-trivially on the first qubit, and acts as the identity on each
$\mathcal S_t$. Then, the action of each $U_t$ needs to satisfy:

\begin{align}
    U_0\left(1\oplus a_0\right)
        &=0\oplus b_0,
        \label{eq:majority-U0-required-action}\\
    U_t\left(0\oplus a_t\right)
        &=0\oplus b_t
        &&(1\leq t\leq k-1),
        \label{eq:majority-Ut-required-action}\\
    U_k\left(0\oplus a_k\right)
        &=1\oplus b_k.
        \label{eq:majority-Uk-required-action}
\end{align}
We take the following non-trivial unitary maps that satisfy the relations \Cref{eq:majority-U0-required-action,eq:majority-Ut-required-action,eq:majority-Uk-required-action}:

\begin{equation}
    U_t:=
    \begin{cases}
        X,
            &t=0,\\[2mm]
        Z,
            &t\equiv2,3\pmod4,\\[2mm]
        \displaystyle
        \frac{1}{a_k}
        \begin{pmatrix}
            b_k&1\\[1mm]
            -1&b_k
        \end{pmatrix},
            &t=k,\\[6mm]
    \end{cases}
    \label{eq:majority-weight-blocks}
\end{equation}
with the understanding that $U_t = I$ for unspecified Hamming weights $t$. Indeed, $P_k(0)=Q_k(0)=1/\sqrt2$, so $a_0=0$, $b_0=1$, and the
$t=0$ block has the required action.  For $1\leq t\leq k-1$, $P_k(t)=0$ when $t\equiv0,1\pmod4$ and $Q_k(t)=0$
when $t\equiv2,3\pmod4$.  Therefore $b_t/a_t=1$ in the first case and
$b_t/a_t=-1$ in the second, so the required block is respectively $I$ or
$Z$.  Finally, $P_k(k)Q_k(k)=1/2$ and because $k$ is odd $a_k^2 - b_k^2 = 1$.
Note that for $t > k$, $\ket{A_t}=0$ by \Cref{eq:majority-krawtchouk-hadamard}, so any unitary
$U_t \in U(2)$ is valid.

It follows from \Cref{eq:majority-Ut-required-action,eq:majority-Uk-required-action}
that $U = \widehat H\widetilde{U} \widehat H$ satisfies the relation in \Cref{eq:majority_transducer_map}.

\paragraph{Hamming weight counter.}
To implement $\widetilde U$, we coherently compute the Hamming weight
$t=|T|$ of the $n$-qubit register $Q$ on a separate $\lceil\log_2(n+1)\rceil$-qubit register $W$. We use the following Hamming weight counter.

\begin{theorem}[Hamming-weight
counter~{\cite[Theorem~1]{zi2026exacthamming}}]
\label{thm:minimum-width-hamming-weight-counter}
Let $Q$ be an $n$-qubit register and let
$\ell=\lceil\log_2(n+1)\rceil$.  There is an exact unitary
$C$ satisfying
\begin{equation*}
    C
    \ket{T}_Q\ket{0^\ell}_W
    =
    \ket{T}_Q\ket{|T|}_W
    \qquad(T\subseteq[n]),
\end{equation*}
with $O(n\log n)$ gate complexity and $O(\log n)$ depth.
\end{theorem}

Let $k_{\ell} \dots k_{1}$ and $t_{\ell} \dots t_{1}$ be the binary representations of $k$ and
any $0\leq t \leq n$ respectively. We implement the three non trivial unitaries $U_t$ defined in
\Cref{eq:majority-weight-blocks} in a sequential manner. They act on the qubit $D$, and are
controlled on the register $W = W_\ell \otimes \dots \otimes W_1$ as follows:

\begin{enumerate}
    \item $X$ controlled on the state $\ket{0}_W$,
    \item $Z$ controlled on the state $\ket{1}_{W_2}$ (the second least significant qubit of $W$),
    \item $U_kZ^{k_2}$ controlled on the state $\ket k_W = \ket{k_\ell}_{W_\ell} \otimes \cdots \otimes \ket{k_1}_{W_1}$.
\end{enumerate}

If $k\equiv3\pmod4$, then the $Z^{k_2}$ gate undoes the controlled
$Z$ on the state $\ket k_W$. Note that for $k \equiv 1\pmod4$ we have $P_k(k) = Q_k(k)$ as seen in \Cref{eq:PQ-0k-equality}. Therefore $a_k = -1$ and $b_k = 0$, which simplifies the third controlled gate $U_kZ^{k_2}$ to $-iY$.

The final circuit after the uncomputation of the Hamming weight counter is shown in \Cref{fig:majority-unitary-circuit}.
\begin{figure}[htbp]
\centering
\begin{tikzpicture}[
    x=1cm,
    y=.78cm,
    wire/.style={line width=.45pt},
    gate/.style={draw,fill=white,rounded corners=1pt,
        minimum width=1.15cm,minimum height=.56cm,align=center,font=\small},
    block/.style={gate,minimum width=1.30cm}
]
    \def\xend{12.6}
    \node[anchor=east] at (0,0)  {$D:$};
    \node[anchor=east] at (0,-1) {$Q:$};
    \node[anchor=east] at (0,-2) {$W:\ket{0^\ell}$};
    \foreach \y in {0,-1,-2}
        \draw[wire] (.1,\y)--(\xend,\y);

    \fill (1.15,0) circle (1.6pt);
    \draw[wire] (1.15,0)--(1.15,-1);
    \node[gate] at (1.15,-1) {$H^{\otimes n}$};
    \node[block,minimum height=1.38cm] at (3.15,-1.5)
        {$C$};
    \node[gate] (xgate) at (4.85,0) {$X$};
    \fill (4.85,-2) circle (1.6pt);
    \draw[wire] (xgate.south)--(4.85,-2);
    \node[below,font=\scriptsize] at (4.85,-2) {$0$};
    \node[gate] (zgate) at (6.70,0) {$Z$};
    \fill (6.70,-2) circle (1.6pt);
    \draw[wire] (zgate.south)--(6.70,-2);
    \node[below,font=\scriptsize] at (6.70,-2) {$t_2$};
    \node[gate,minimum width=1.25cm] (ukgate) at (8.55,0)
        {$U_kZ^{k_2}$};
    \fill (8.55,-2) circle (1.6pt);
    \draw[wire] (ukgate.south)--(8.55,-2);
    \node[below,font=\scriptsize] at (8.55,-2) {$k$};
    \node[block,minimum height=1.38cm] at (10.10,-1.5)
        {$C^\dagger$};
    \fill (11.75,0) circle (1.6pt);
    \draw[wire] (11.75,0)--(11.75,-1);
    \node[gate] at (11.75,-1) {$H^{\otimes n}$};

    \node[anchor=west] at (\xend,-2) {$\ket{0^\ell}$};
\end{tikzpicture}
\caption{Circuit implementation of
$U$.}
\label{fig:majority-unitary-circuit}
\end{figure}

\begin{theorem}[Complexity of the algorithm for the $|A|=k$ promise]
Let $2\leq k\leq n$ be odd and suppose that $|A|=k$, and fix any constant
$\eps\in(0,1)$. The algorithm obtained from the transducer in \Cref{eq:majority-clock-transducer} has error
at most $\eps$ and uses
\begin{align*}
    O\!\left(\log k\right)
        &\quad\text{oracle calls},
        \\
    O\!\left(
        n\log n\log k
    \right)
        &\quad\text{non-oracle gates},
        \\
    O\!\left(\log n\log k\right)
        &\quad\text{non-oracle depth},
\end{align*}
assuming all-to-all connectivity and using gates acting on one or two qubits.
\end{theorem}

The gate complexity and circuit depth are proved in
\Cref{app:majority-depth}.

\subsubsection{Gate implementation of \texorpdfstring{$U$}{U} for the promise \texorpdfstring{$1 \leq |A| \leq k$}{1<=|A|<=k}.}
\label{sec:U-unknown}

For the general $1 \leq |A| \leq k$ promise, implementing $U$ efficiently is not as straightforward as in~\Cref{sec:majority-unitary-implementation}, since the adversary solution in~\Cref{sec:majority-unknown-size} is not rank-one. Consequently, the catalyst in~\Cref{eq:majority-transducer-catalyst} lives in a bigger Hilbert space $\mathcal Q \otimes \mathcal W$. In our case, the new space is the same as in~\Cref{sec:new-pol}, namely
\[
    \mathcal W
    :=\operatorname{span}_{\mathbb R}\{e_*,e_0,\ldots,e_{M-1}\}.
\]    
where $M$ denotes the largest integer such that $2M\leq k$. On this Hilbert space, we have defined the polynomials
$p_\ell$ and $q_\ell$ in~\Cref{eq:unknown-vector-polynomials},
and the feasible dual adversary matrices
\begin{equation}
    X_S\llbracket A,B\rrbracket
    :=\frac{1}{2^n}
      \left\langle W_m(|A\cap S|),W_\ell(|B\cap S|)\right\rangle,
    \qquad S\subseteq[n],
    \label{eq:unknown-implementation-adversary}
\end{equation}
where $|A|=2m+1$ and $|B|=2\ell+1$ and
\begin{equation*}
    W_m(a)
    :=\sum_{t=0}^{2m+1}K_t^{(2m+1)}(a)
       \bigl(p_m(t)+(-1)^tq_m(t)\bigr).
\end{equation*}
They are a feasible adversary
solution of objective value $O(\log k)$ as proved in~\Cref{sec:majority-unknown-size}. Similarly to the previous section, we define the output and query spaces as
$\mathcal Q=\operatorname{Span}\{\ket S:S\subseteq[n]\}$,
and $\mathcal V\cong\mathcal Q$. The catalyst associated with $A$ such that $|A|=2m+1$ is
\begin{equation*}
\begin{aligned}
    \ket{w_A}
    &=\frac{1}{\sqrt{2^{n+1}}}
      \sum_{S\subseteq[n]}\sum_{t=0}^{2m+1}
      K_t^{(2m+1)}(|A\cap S|)\ket S
      \ket{p_m(t)+(-1)^tq_m(t)}
      \in\mathcal Q\otimes\mathcal W.
\end{aligned}
\end{equation*}
We work on the direct sum
$\mathcal V\oplus(\mathcal Q\otimes\mathcal W)$. Analogously to the previous section, the unitary $U$ has to satisfy
\begin{equation*}
    U\left(
        \ket0
        \oplus(O_A^\pm\otimes I_{\mathcal W})\ket{w_A}
      \right)
    =
      \ket A\oplus\ket{w_A}
\end{equation*}
for every $A\in \mathcal A$.

\paragraph{Hadamard transform.}
Set as before
\begin{equation*}
    a_m(t):=\frac{-p_m(t)+(-1)^tq_m(t)}{\sqrt2},
    \qquad
    b_m(t):=\frac{p_m(t)+(-1)^tq_m(t)}{\sqrt2}.
\end{equation*}
The Oracle then maps from $a_m(t)$ to $b_m(t)$. Denoting the output and
query summands by $\ket0$ and $\ket1$, respectively, the controlled
Hadamard transform is
\[
    \widehat H
    :=\ket0\!\bra0\otimes I_{\mathcal V}
      +\ket1\!\bra1\otimes
       (H_{\mathcal Q}^{\otimes n}\otimes I_{\mathcal W}).
\]
It
turns the state before and after the oracle into
$\sum_t\ket{A_t}\ket{a_m(t)}$ and
$\sum_t\ket{A_t}\ket{b_m(t)}$, respectively.
Thus, writing $U=\widehat H\widetilde U\widehat H$, it is enough to
require
\begin{equation*}
\begin{aligned}
\widetilde U\left(
    \ket0
    \oplus\sum_{t=0}^{2m+1}\ket{A_t}\ket{a_m(t)}
    \right)
={}&
    \ket A
    \oplus\sum_{t=0}^{2m+1}\ket{A_t}\ket{b_m(t)}.
\end{aligned}
\end{equation*}
We choose the same construction for implementing the unitary $\widetilde U$:
\begin{equation*}
    \widetilde U
    =\sum_{t=0}^{n}\Pi_t\otimes U_t,
    \qquad U_t\in U(\mathbb C\oplus\mathcal W).
\end{equation*}
Then, for every $0\leq m\leq M$, the blocks must satisfy
\begin{align}
U_0\bigl(1\oplus\ket{a_m(0)}\bigr)
    &=0\oplus\ket{b_m(0)},
    \label{eq:unknown-U0-required-action}\\
U_t\bigl(0\oplus\ket{a_m(t)}\bigr)
    &=0\oplus\ket{b_m(t)}
    &&(1\leq t\leq2m),
    \label{eq:unknown-Ut-required-action}\\
U_{2m+1}\bigl(0\oplus\ket{a_m(2m+1)}\bigr)
    &=1\oplus\ket{b_m(2m+1)}.
    \label{eq:unknown-Uendpoint-required-action}
\end{align}

When $t=0$, the coefficient $C_{jm}(0)$ defined in~\Cref{eq:unknown-Cml} is zero. \Cref{eq:unknown-vector-polynomials} then implies that $\langle p_{m}(0),e_j\rangle \neq 0$ only for the "base" component $j=*$, and the $j=m$ component if $m$ is even. Similarly, $\langle q_{m}(0),e_j\rangle \neq 0$ only if $j=*$ and $j=m$ is odd. On those two components, the scalar product evaluates to $P_{m}(0)=Q_m(0) = 1/\sqrt2$ for both vectors. Hence, 

\[
\ket{p_m(0)} = \frac{1}{\sqrt{2}}\bigl(\ket{e_*} +\ket{e_m}\bigr) \quad \text{and} \quad \ket{q_m(0)} = \frac{1}{\sqrt{2}}\ket{e_*} \quad \text{if }m\text{ even,}
\]
as well as
\[
\ket{p_m(0)} = \frac{1}{\sqrt{2}}\ket{e_*} \quad \text{and} \quad \ket{q_m(0)} = \frac{1}{\sqrt{2}}\bigl(\ket{e_*} +\ket{e_m}\bigr)  \quad \text{if }m\text{ odd.}
\]
Therefore, $\ket{a_{m}(0)} = \frac{(-1)^{m+1}}{2}\ket{e_m}$ and $\ket{b_{m}(0)} = \ket{e_*} + \frac{1}{2}\ket{e_m}$. The expressions with $\ket{e_m}$ apply for \(0\le m<M\), since $e_M\notin \mathcal W$. For $m=M$, we instead have $\ket{p_M(0)}=\ket{q_M(0)}=\frac{1}{\sqrt2}\ket{e_*}$, $\ket{a_M(0)}=0$, and $\ket{b_M(0)}=\ket{e_*}$, which is consistent with the action of $U_0$ as:
\[
U_0\left(1 \oplus \frac{(-1)^{m+1}}{2}\ket{e_m}\right) = 0 \oplus \ket{e_*} + \frac{1}{2}\ket{e_m}, \qquad 0\le m<M,
\]
while for $m=M$, we use $U_0(1\oplus0)=0\oplus\ket{e_*}$. Hence we can take
\begin{equation}
U_0 =
\ket{e_*}\bra{0}
+
\ket{0}\bra{e_*}
+
\sum_{i=0}^{M-1}(-1)^{i+1}\ket{e_i}\bra{e_i},
\end{equation}
where $\ket{0} := 1\oplus0$ in our notation.
Notice that, among other things, $U_0$ swaps $\ket{0}$ and $\ket{e_*}$, similarly to our controlled $X$ gate in~\Cref{sec:majority-unitary-implementation}.

The remaining blocks $U_t$, $1\leq t\leq k$, are defined in
\Cref{app:unknown-size-unitary-implementation}. Note that, as in the previous section, for $t > k$, any unitary $U_t \in U(M+2)$ trivially satisfies~\Cref{eq:unknown-U0-required-action,eq:unknown-Ut-required-action,eq:unknown-Uendpoint-required-action}.
Fortunately, those unitaries can be implemented using $O(\log k)$ gates and with $O(\log k)$ depth.

As in the known-size construction, we use the same counter $C$ to compute the Hamming
weight $t$ coherently and control the corresponding block $U_t$. Finally, after uncomputing the Hamming weight counter, we apply $\widehat H$. We get the following complexity for our algorithm.

\begin{theorem}[Complexity of the algorithm for the $1\leq|A|\leq k$ promise]
\label{thm:majority-unknown-time-efficient-implementation}
Let $2\leq k\leq n$, suppose that $1\leq|A|\leq k$, and fix any
constant $\eps\in(0,1)$.  The algorithm obtained from the adversary
solution in~\Cref{eq:unknown-implementation-adversary} has error at
most $\eps$ and uses
\begin{align*}
    O(\log k)&\quad\text{oracle calls},\\
    O(n\log n\log k)&\quad\text{non-oracle gates},\\
    O(\log n\log k)&\quad\text{non-oracle depth},
\end{align*}
assuming all-to-all connectivity and using gates acting on one or two qubits.
\end{theorem}
Crucially, we apply the lift in~\Cref{eq:unknown-parity-lift} for even support sizes, which preserves the overall complexity of the algorithm.

\section{Learning homogeneous LTFs with an  Example Oracle}\label{sec:example}

We now describe our quantum algorithm for learning an unknown homogeneous halfspace: $f_w(x)=\sgn(\langle w,x\rangle)$ with $w\in\sphere^{n-1}$, from quantum examples drawn from the standard Gaussian distribution, $\gamma_n$. The learner is given coherent access to $\QEX(f_w,\gamma_n)$ and its inverse (see~\Cref{ssec:introHermitealgo}). The main idea is to convert a quantum example into a phase-encoded representation of $f_w$, apply the quantum Hermite transform of~\cite{JainIyerSommaBaoJordan2025} to expose its Hermite spectrum, and exploit the special tensor structure of the Hermite coefficients of a ridge function. A suitable high-degree component encodes many identical copies of the unknown direction state $\ket w:=\sum_{i=1}^n w_i\ket i$, which can then be recovered by pure-state tomography.

The remainder of this section turns the above idealized picture into a finite-dimensional algorithm and accounts for all approximation and implementation costs. In~\Cref{sec:discrete-qht}, we first analyze the discretized quantum Hermite transform and prove that the resulting state, restricted to the coordinatewise degree-cutoff subspace $\operatorname{Ran}(P_D^{\otimes n})$, is close to the corresponding cutoff projection of the ideal Hermite state. In~\Cref{sec:high-degree-projection}, we isolate and amplify an informative degree band \(k\in[K,2K]\). In~\Cref{sec:isometry-map}, we construct an efficient isometry that converts this band into a state containing \(K\) approximate copies of the hidden direction state \(\ket{w}\). We then recover a classical description of \(w\) using pure-state tomography in~\Cref{sec:pure-state-tomography}. Finally, in~\Cref{sec:example-putting-together}, we combine these ingredients and choose the parameters to obtain the stated query and gate complexities.
\subsection{Discrete Quantum Hermite Transform}\label{sec:discrete-qht}
For an integer $D\geq 1$, define the projector onto the computational subspace corresponding to one-dimensional Hermite degrees smaller than $D$ by
\begin{equation*}
    P_D
    :=
    \sum_{k=0}^{D-1}
    \ket{k}\!\bra{k}.
\end{equation*}
The finite-dimensional $\QHT$~\cite{JainIyerSommaBaoJordan2025} is only guaranteed to reproduce the Hermite change of basis on this subspace. Accordingly, our goal is not to approximate the full state $\ket{F_w}$, but rather to output a normalized state $\ket{\widetilde F_{w,M}}$ (for an appropriate discretization dimension $M$) satisfying
\begin{equation*}
    \left\|
        P^{\otimes n}_D\ket{\widetilde F_{w,M}}
        -
        P^{\otimes n}_D\ket{F_w}
    \right\|
    \leq \eta.
\end{equation*}
Since $\Pi_{B_K}P^{\otimes n}_D=\Pi_{B_K}$, whenever $D\geq2K+1$, this projected guarantee is sufficient for the learning algorithm.

After fixing degree $D$, preparing the Hermite state introduces two conceptually distinct sources of error. The first is the \emph{discretization error}, which stems from truncating the unbounded Gaussian domain and from replacing the continuum by a finite grid. The second is the \emph{$\QHT$ approximation error}: even on this discrete lattice, the finite quantum Hermite transform only approximates the ideal change of basis between discretized Hermite functions and occupation-number states. Accordingly, our analysis proceeds by first choosing the lattice parameters so that the discretized phase state is close to its continuum counterpart, and then choosing the parameters of the discrete $\QHT$ so that its action on all Hermite degrees in the band $[K,2K]$ is sufficiently accurate. Combining these two
bounds will show that the implemented Hermite state, after projection onto $\mathrm{Ran}(P_D^{\otimes n})$, is close to the corresponding projection of the ideal Hermite state: $P_D^{\otimes n}\ket{\widetilde F_{w,M}}\approx P_D^{\otimes n}\ket{F_w}.$

Note that the original implementation of the quantum Hermite transform ($\QHT$) in~\cite{JainIyerSommaBaoJordan2025} is formulated using the physicists' convention for Hermite polynomials and Hermite functions. In particular, their discretization uses the lattice $y_j = j\sqrt{\frac{2\pi}{M}}$ for $j\in \{-M/2,\ldots,M/2-1\}$ and $M=2^m$. Since throughout this work we use the probabilists' convention, we must rescale the spatial coordinate accordingly. Recall from~\Cref{ssec:hermitedefs} that $\operatorname{He}_k(x)=2^{-k/2}H_k\!\left(\frac{x}{\sqrt{2}}\right)$. Consequently, we use the lattice spacing
$\Delta:=\sqrt{\frac{4\pi}{M}}$. With $L:=\sqrt{\pi M}$, we define
\begin{equation}
    S_\Delta[-L, L] = \{j\Delta: j\in \Z, -L\leq j\Delta < L  \}.
\end{equation}
Importantly, this fixed rescaling does not change the required Hilbert-space dimension $M$ or the asymptotic complexity of the $\QHT$ circuit. It only changes the physical lattice spacing and truncation range.

Let us first consider the discretization error. We follow essentially the same analysis as in~\cite{JainIyerSommaBaoJordan2025}, after accounting for the rescaling introduced above. A minor additional subtlety arises in our setting because the sign function is discontinuous along the decision boundary $\{x:\langle w,x\rangle=0\}$. Consequently, discretization cells that intersect this boundary may introduce an additional error due to a change of sign across the grid. We bound this boundary contribution separately in the proposition below.
\begin{proposition}[Halfspace-boundary discretization error]
\label{prop:halfspace-boundary-error}
Let $f_w(x)=\sgn(\langle w,x\rangle)$ with $w\in\sphere^{n-1}$. Fix a finite dimension $M>0$, with grid spacing $\Delta:=\sqrt{4\pi/M}$ and $L:=\Delta M/2$. For $x\in[-L,L]^n$, define $\bar x$ coordinatewise by $\bar x_i:=\max\{z\in S_\Delta[-L,L]:z\leq x_i\}$. Then
\[
    \int_{[-L,L]^n}\abs{f_w(x)-f_w(\bar x)}^2\gamma_n(x)\,dx
    \leq4\sqrt{\frac{8n}{M}}.
\]
\end{proposition}
\begin{proof}
    If $f_w(x)\neq f_w(\bar x)$, then the two values have opposite signs, which implies
    \begin{equation}
        \abs{\inner{w}{x}}\leq\abs{\inner{w}{x-\bar x}}
        \leq\norm{w}\norm{x-\bar x}\leq\Delta\sqrt n.
    \end{equation}
    Hence,
    \begin{align}
        \Pr_{x\sim\gamma_n}[f_w(x)\neq f_w(\bar x)]
        &\leq\Pr_{x\sim\gamma_n}[\abs{\inner{w}{x}}\leq\Delta\sqrt n]
        \leq\sqrt{\frac{2}{\pi}}\,\Delta\sqrt n
        =\sqrt{\frac{8n}{M}}.
    \end{align}
    The last equality follows from our choice of $\Delta$. Since $\abs{f_w(x)-f_w(\bar x)}=2$ whenever $f_w(x)\neq f_w(\bar x)$, we obtain
    \begin{align}
        \int_{[-L,L]^n}\abs{f_w(x)-f_w(\bar x)}^2\gamma_n(x)\,dx
        &\leq4\Pr_{x\sim\gamma_n}[f_w(x)\neq f_w(\bar x)]
        \leq4\sqrt{\frac{8n}{M}}.
    \end{align}
\end{proof}
We now bound the total discretization error of the Hermite coefficients up to a finite degree $D$. For the continuous expansion, define
\begin{equation}
    \widehat f_w(\alpha):=\int_{\R^n}f_w(x)h_\alpha(x)\gamma_n(x)\,dx.
\end{equation}
For the discretized domain, we define the corresponding normalized discrete coefficients by
\begin{equation}
    \widehat f_{w,M}(\alpha)
    :=
    \frac{\Delta^n}{\sqrt{Z}}
    \sum_{\bar{x}\in S_{\Delta}[-L,L]^n}
    f_w(\bar{x})
    h_\alpha(\bar{x})
    \gamma_n(\bar{x}),
\end{equation}
where $Z:=\Delta^n\sum_{\bar{x}\in S_{\Delta}[-L,L]^n}\gamma_n(\bar{x})$ is the normalization factor of the discretized Gaussian distribution.

The following lemma quantifies the cumulative error between the continuous and discrete Hermite coefficients, restricted to degrees up to $D$, as a function of the finite discretization dimension $M$. For better readability, we defer the detailed proof of this lemma to~\Cref{app:discretization}.

\begin{lemma}\label{lem:discretization-error}
    For every finite degree cutoff $D$ and every $M\geq64n^2(D+1)$, the following inequality holds:
    \begin{equation}
        \left(
        \sum_{\alpha\in[D]_0^n}
        \left|
            \widehat f_{w,M}(\alpha)
            -
            \widehat f_w(\alpha)
        \right|^2
        \right)^{1/2}
        \leq
        24n
    \sqrt{\frac{(D+1)}{M}}
    +2\left(
        \frac{8n}{M}
    \right)^{1/4}
    +
    2\sqrt{2n}\,e^{-\pi M/4}.
    \end{equation}
\end{lemma}

The next error is the approximation error from the $\QHT$ circuit itself. On the rescaled lattice introduced above, for each Hermite degree
$k\in\{0,\ldots,D-1\}$, we define the corresponding discrete Hermite state with the grid spacing $\Delta=\sqrt{4\pi/M}$ by
\begin{equation}
    \ket{\psi_k^{(M)}}
    =
    \left(\frac{2}{M}\right)^{1/4}
    \sum_{x\in S_\Delta[-L,L]}
    e^{-x^2/4}h_k(x)\ket{x}.
    \label{eq:discrete-hermite-state-expanded}
\end{equation}
Here $\ket{x}$ denotes the computational-basis state associated with the lattice point $x=j\Delta$. Under the change of variables $x=\sqrt{2}\,y$, our convention differs from that of~\cite{JainIyerSommaBaoJordan2025} by the known phase $(-1)^k$ on degree $k$. Let
\[
    R_D\ket{k}:=(-1)^k\ket{k},\qquad 0\leq k<D.
\]
We therefore compose the $\QHT$ circuit of~\cite{JainIyerSommaBaoJordan2025} with $R_D$ on its degree input, and denote the resulting phase-corrected circuit again by $\mathsf{QHT}_{D,M}$. This adds only $O(\log D)$ gate overhead and makes the synthesis convention agree exactly with~\Cref{eq:discrete-hermite-state-expanded}. Throughout this section, we use $\QHT_{D,M}$ in the
$M$-dimensional unitary formulation of
\cite[Theorem~19]{JainIyerSommaBaoJordan2025}.
We invoke its uniform state-vector approximation guarantee
on $\operatorname{Ran}(P_D)$ and use $\QHT_{D,M}^{\dagger}$
for its unitary inverse. The errors covered by the following lemma are included in its stated accuracy parameter. The implementation of~\cite{JainIyerSommaBaoJordan2025} uses internal ancillary workspace that is returned to the all-zero state only approximately, within the stated QHT error. We suppress these ancillas in the input--output notation below; whenever we project onto the intended QHT output subspace, the corresponding projector is understood to also project the internal QHT workspace onto its all-zero state.
\begin{lemma}[Discrete Quantum Hermite transform
{\cite[Theorem~19]{JainIyerSommaBaoJordan2025}}]
\label{lemma:quantum-hermite-transform}
Let $\eta\in(0,1/2)$ and let $D\geq2$ satisfy $D>\log(1/\eta)$. For a sufficiently large universal constant $C>0$ and a power of two
\begin{equation*}
  M\geq C D^{9/4}\eta^{-13/4},
\end{equation*}
there is a quantum circuit $\mathsf{QHT}_{D,M}$ such that
\begin{equation*}
  \left\|
    \mathsf{QHT}_{D,M}
      \sum_{k=0}^{D-1}c_k\lvert k\rangle
    -\sum_{k=0}^{D-1}c_k\lvert\psi_k^{(M)}\rangle
  \right\|
  \leq\eta
\end{equation*}
for every $c\in\mathbb C^D$ with $\sum_{k=0}^{D-1}|c_k|^2=1$. The circuit uses $O(\log^3(M)\log(1/\eta))$ elementary gates.
\end{lemma}
Now, we are ready to bound the total error incurred in the discrete quantum Hermite-transform step. Under an appropriate choice of $M$, which we specify later, we have 
\begin{equation}
    (I\otimes H)\QEX_{M}(f_w, \gamma_n)\ket{0} = \frac{1}{\sqrt2}
    \left(
        \ket{\widetilde G_M}\ket0_Y
        +
        \ket{\widetilde \Phi_{w,M}}\ket1_Y
    \right) 
\end{equation}
where
\begin{equation*}
   \ket{\widetilde G_M}:= \frac{\Delta^{n/2}}{\sqrt Z}
    \sum_{\bar{x}}
    \sqrt{\gamma_n(\bar{x})}
    \ket{\bar{x}}, \qquad \ket{\widetilde\Phi_{w,M}}
    :=
    \frac{\Delta^{n/2}}{\sqrt Z}
    \sum_{\bar{x}}
    f_w(\bar{x})
    \sqrt{\gamma_n(\bar{x})}
    \ket{\bar{x}}.
\end{equation*}
Applying inverse $\QHT_{D, M}$ on the $n$ registers, we denote
\begin{equation}
    \ket{\widetilde F_{w,M}} := (\QHT^\dagger_{D, M})^{\otimes n} \ket{\widetilde \Phi_{w, M}}, \qquad \ket{\widetilde F_{0,M}} := (\QHT^\dagger_{D, M})^{\otimes n} \ket{\widetilde G_M}.
\end{equation}
By~\Cref{lemma:quantum-hermite-transform}, the forward \(\QHT\) accurately maps low-degree computational states to the corresponding discrete Hermite states. Taking adjoints, the inverse \(\QHT\) therefore accurately extracts the corresponding low-degree Hermite coefficients, as made precise below. Thus, our goal is to control its projection onto the degree-$D$ subspace:
\begin{equation*}
    \left\|
        P^{\otimes n}_D\ket{\widetilde F_{w,M}}
        -
        P^{\otimes n}_D\ket{F_w}
    \right\|
    \leq \eta.
\end{equation*} 
where $\ket{F_w}$ is the full normalized Hermite state
\begin{equation}
    \ket{F_w}
    =
    \sum_{\alpha\in\mathbb N^n}
    \widehat f_w(\alpha)\ket{\alpha}.
\end{equation}
Its restriction to the degree-$D$ subspace is the generally subnormalized vector
\begin{equation*}
    P^{\otimes n}_D\ket{F_w}
    =
    \sum_{\alpha\in[D]_0^n}
    \widehat f_w(\alpha)\ket{\alpha}.
\end{equation*}
This is presented in the following theorem.
\begin{theorem}[Finite-dimensional preparation of the Hermite component]
\label{thm:finite-hermite-state-preparation}
Let $f_w(x)=\sgn(\langle w,x\rangle)$ with $w\in\sphere^{n-1}$. Let $D\geq 2$ and $0<\eta<1/2$ satisfy $D>\log(4n/\eta)$. Let $M$ be a power of two satisfying
\begin{equation*}
    M
    \geq
    \max\left\{
        64n^2(D+1),
        \frac{20736\,n^2(D+1)}{\eta^2},
        \frac{165888\,n}{\eta^4},
        \frac{4}{\pi}
        \log\!\left(
            \frac{12\sqrt{2n}}{\eta}
        \right),
        C D^{9/4}
        \left(
            \frac{4n}{\eta}
        \right)^{13/4}
    \right\},
\end{equation*}
where $C>0$ is the universal constant from~\Cref{lemma:quantum-hermite-transform}. Then there exists a quantum circuit $A_w$, using one query to
$\QEX_M(f_w,\gamma_n)$ and $ O\!\left(n\log^3(M)\log\frac{4n}{\eta}\right)$ additional elementary gates, such that
\begin{equation*}
    A_w\ket{0}
    =
    \frac{1}{\sqrt2}
    \left(
        \ket{\widetilde F_{0,M}}\ket0_Y
        +
        \ket{\widetilde F_{w,M}}\ket1_Y
    \right),
\end{equation*}
where $\ket{\widetilde F_{w,M}}$ is normalized and satisfies
\begin{equation*}
    \left\|
        P^{\otimes n}_D\ket{\widetilde F_{w,M}}
        -
        P^{\otimes n}_D\ket{F_w}
    \right\|
    \leq
    \eta.
\end{equation*}
\end{theorem}

\begin{proof}
We separate the total approximation error into the two above-mentioned conceptually distinct contributions: the discretization error of the Hermite coefficients and the approximation error of the finite quantum Hermite transform. 

Define the finite-grid Hermite coefficient vector
\begin{equation*}
    \ket{F_{w,M}^{(D)}}:=\sum_{\alpha\in[D]_0^n}\widehat f_{w,M}(\alpha)\ket{\alpha}.
\end{equation*}
By the triangle inequality,
\begin{align}
    &
    \left\|
        P^{\otimes n}_D\ket{\widetilde F_{w,M}}
        -
        P^{\otimes n}_D\ket{F_w}
    \right\|
    \nonumber\\
    &\qquad\leq
    \underbrace{
    \left\|
        P^{\otimes n}_D\ket{\widetilde F_{w,M}}
        -
        \ket{F_{w,M}^{(D)}}
    \right\|
    }_{\mathcal E_{\mathrm{QHT}}}
    +
    \underbrace{
    \left\|
        \ket{F_{w,M}^{(D)}}
        -
        P^{\otimes n}_D\ket{F_w}
    \right\|
    }_{\mathcal E_{\mathrm{disc}}}.
    \label{eq:total-error-decomposition}
\end{align}
By~\Cref{lem:discretization-error}, whenever $M\geq64n^2(D+1)$,
\begin{align*}
    \mathcal E_{\mathrm{disc}}
    &=
    \left(
        \sum_{\alpha\in[D]_0^n}
        \left|
            \widehat f_{w,M}(\alpha)
            -
            \widehat f_w(\alpha)
        \right|^2
    \right)^{1/2}
    \nonumber\\
    &\leq
    24n\sqrt{\frac{D+1}{M}}
    +
    2\left(
        \frac{8n}{M}
    \right)^{1/4}
    +
    2\sqrt{2n}\,e^{-\pi M/4}.
\end{align*}
We allocate an error budget $\eta/6$ to each of the three terms on the right-hand side. It is therefore sufficient that
\begin{equation*}
    M
    \geq
    \max\left\{
        64n^2(D+1),
        \frac{20736\,n^2(D+1)}{\eta^2},
        \frac{165888\,n}{\eta^4},
        \frac{4}{\pi}
        \log\!\left(
            \frac{12\sqrt{2n}}{\eta}
        \right)
    \right\}.
\end{equation*}
Under these conditions,
\begin{equation}
    \mathcal E_{\mathrm{disc}}
    \leq
    \frac{\eta}{2}.
    \label{eq:disc-error-half}
\end{equation}

We now bound $\mathcal E_{\mathrm{QHT}}=\left\|P^{\otimes n}_D\ket{\widetilde F_{w,M}}-\ket{F_{w,M}^{(D)}}\right\|$. For one coordinate, define the ideal discrete Hermite synthesis map
\begin{equation*}
    S_{D,M}
    :=
    \sum_{k=0}^{D-1}
    \ket{\psi_k^{(M)}}\!\bra{k},
\end{equation*}
where
\begin{equation*}
    \ket{\psi_k^{(M)}}=\left(\frac{2}{M}\right)^{1/4}\sum_{x\in S_\Delta[-L,L]}e^{-x^2/4}h_k(x)\ket{x}.
\end{equation*}
Equivalently, since $\Delta=\sqrt{4\pi/M}$ and $\sqrt{\gamma(x)}=(2\pi)^{-1/4}e^{-x^2/4}$,
\begin{equation*}
    \ket{\psi_k^{(M)}}
    =
    \sqrt{\Delta}
    \sum_{x\in S_\Delta[-L,L]}
    h_k(x)\sqrt{\gamma(x)}\ket{x}.
\end{equation*}
Consider the actual one-dimensional quantum Hermite transform circuit, $\mathsf{QHT}_{D,M}$. If we apply~\Cref{lemma:quantum-hermite-transform} with additive error of $\frac{\eta}{4n}$, then for every normalized $\ket c\in\operatorname{Ran}(P_D)$, we have
\begin{equation*}
    \left\|
        \mathsf{QHT}_{D,M}\ket c
        -
        S_{D,M}\ket c
    \right\|
    \leq
    \frac{\eta}{4n}.
\end{equation*}
Now, for a normalized state $\ket{v}$, we have $P_D\ket{v} \in \operatorname{Ran}(P_D)$, and
\begin{align*}
    \left\|
    ( \mathsf{QHT}_{D,M}
        -
        S_{D,M}) P_D\ket{v}
    \right\| 
    &\leq
    \frac{\eta}{4n} \norm{P_D\ket{v}} \\
    &\leq \frac{\eta}{4n} \norm{\ket{v}} \\
    &\leq \frac{\eta}{4n}.
\end{align*}
Thus,
\begin{align}
    \left\|
    (\mathsf{QHT}_{D,M}
        -
        S_{D,M}) P_D
    \right\| = \left\|
    \mathsf{QHT}_{D,M}P_D
        -
        S_{D,M}
    \right\| \leq \frac{\eta}{4n}.
\end{align}
The equality comes from the fact that $S_{D,M} P_D = S_{D,M}$.

To lift this upper bound to $n$ coordinates, we use a simple hybrid argument. We consider
\begin{equation}
    \|S_{D,M}\|
    \leq
    \|\mathsf{QHT}_{D,M}P_D\|
    +
    \|\mathsf{QHT}_{D,M}P_D-S_{D,M}\|
    \leq
    1+\frac{\eta}{4n}.
\end{equation}
Thus, with $\delta:=\frac{\eta}{4n}$,
\begin{align*}
    \left\|
        (\mathsf{QHT}_{D,M}P_D)^{\otimes n}
    -
    S_{D,M}^{\otimes n}
    \right\|
    &\leq
    \delta
    \sum_{j=0}^{n-1}
    (1+\delta)^j
    \nonumber\\
    &=
    (1+\delta)^n-1
    \nonumber\\
    &\leq
    e^{n\delta}-1
    =
    e^{\eta/4}-1
    \nonumber\\
    &\leq
    \frac{\eta}{2},
\end{align*}
where the final inequality follows from $0<\eta<1/2$, and hence $e^{\eta/4}-1\leq 2(\eta/4)=\eta/2$. More precisely, in this argument \(P_D\) denotes the projector \(P_D\otimes\ket{0}\bra{0}_{\mathrm{anc}}\) on the logical degree register together with the clean internal QHT workspace. With this convention, taking adjoints does not introduce components with nonzero ancillary garbage. For convenience, we suppress the factor $\ket{0}\bra{0}_{\mathrm{anc}}$ in the subsequent notation. Taking adjoints and using invariance of the operator norm under adjunction,
\begin{equation*}
    \left\|
        (P_D\mathsf{QHT}_{D,M}^{\dagger})^{\otimes n}
        -
        S_{D,M}^{\dagger\otimes n}
    \right\|
    \leq
    \frac{\eta}{2}.
\end{equation*}
Now, we consider
\begin{align}
    \mathcal E_{\mathrm{QHT}}
    &=
    \left\|
        P^{\otimes n}_D\ket{\widetilde F_{w,M}}
        -
        \ket{F_{w,M}^{(D)}}
    \right\|
    \nonumber\\
    & = \left\|
        (P_D \mathsf{QHT}_{D,M}^{\dagger})^{\otimes n} \ket{\widetilde \Phi_{w, M}} \
        -
        S^{\dagger \otimes n}_{D, M}\ket{\widetilde \Phi_{w, M}}
    \right\| \\
    &\leq
    \frac{\eta}{2}.\label{eq:QHT-error-half}
\end{align}

By~\Cref{lemma:quantum-hermite-transform}, the choice $\delta=\eta/(4n)$ is valid provided $ M
    \geq
    C D^{9/4}
    \left(
        \frac{4n}{\eta}
    \right)^{13/4}$. Each one-dimensional $\QHT$ uses $O\!\left(
        \log^3(M)
        \log\frac{4n}{\eta}
    \right)$ elementary gates, and hence the $n$-fold transform requires $O\!\left(
        n\log^3M
        \log\frac{4n}{\eta}
    \right)$ gates.

Finally, combining
\Cref{eq:disc-error-half,eq:QHT-error-half} with~\Cref{eq:total-error-decomposition}, we obtain
\begin{align}
    \left\|
        P^{\otimes n}_D\ket{\widetilde F_{w,M}}
        -
        P^{\otimes n}_D\ket{F_w}
    \right\|
    &\leq
    \mathcal E_{\mathrm{QHT}}
    +
    \mathcal E_{\mathrm{disc}}
    \nonumber\\
    &\leq
    \frac{\eta}{2}
    +
    \frac{\eta}{2}
    =
    \eta.
\end{align}
This proves the theorem.
\end{proof}

\subsection{High-degree projection}\label{sec:high-degree-projection}

After the discrete quantum Hermite transform of~\Cref{thm:finite-hermite-state-preparation}, we obtain the normalized finite-dimensional state $\ket{\widetilde F_{w,M}}$ and guarantee that its projection onto the coordinatewise degree-$D$ subspace is close to the corresponding projection of the ideal Hermite state. We now extract from this state only the Hermite degrees that contain sufficiently many (i.e., $\Theta(n/\eps^2)$) copies of the hidden direction $w$. Fix an integer $K\ge 1$, and define the degree band
\begin{equation}
    B_K:=\{K,K+1,\ldots,2K\},
\end{equation}  
together with the corresponding projector
\begin{equation}
    \Pi_{B_K} := \sum_{k\in B_K}\Pi_k, \qquad \Pi_k := \sum_{|\alpha|=k}\ket{\alpha}\!\bra{\alpha}.    
\end{equation}
We henceforth choose $D\geq2K+1$. Since every multi-index $\alpha$ satisfying $K\leq|\alpha|\leq2K$ also satisfies $\alpha_i\leq2K<D$ for every $i$, we have
\begin{equation*}
    \Pi_{B_K}P^{\otimes n}_D
    =
    P^{\otimes n}_D\Pi_{B_K}
    =
    \Pi_{B_K}.
\end{equation*}
Consequently, the projected approximation guarantee of~\Cref{thm:finite-hermite-state-preparation} directly controls all amplitudes relevant to the degree band $B_K$. Thus, conditioned on successfully projecting onto \(B_K\), the ideal normalized target state is
\begin{equation}
\ket{F_w^{B_K}} := \frac{\Pi_{B_K}\ket{F_w}} {\|\Pi_{B_K}\ket{F_w}\|} = \frac{1}{\sqrt{p_K}} \sum_{k\in B_K} a_k\ket{C_k(w)},    
\end{equation} 
$p_K := \sum_{k\in B_K}|a_k|^2$. Starting instead from the actually prepared state \(\ket{\widetilde F_{w,M}}\), we aim to obtain a state $\ket{\widetilde F_{w,M}^{B_K}}$ that is close to $\ket{F_w^{B_K}}$. In the following theorem, we show that we can prepare $\ket{\widetilde F_{w, M}^{B_K}}$ using only $O(K^{1/4})$ applications of the Hermite-state preparation circuit and its inverse, while keeping this state close to the ideal band state $\ket{F_w^{B_K}}$.
\begin{theorem}[High-degree projection]
\label{thm:high-degree-projection}
Let $K\geq 1$ and $D\geq 2K+1$, and let $B_K$ and
$\Pi_{B_K}$ be as defined above. Define $ p_K :=\norm{\Pi_{B_K}\ket{F_w}}^2$, and $\ket{F_w^{B_K}}:=\frac{\Pi_{B_K}\ket{F_w}}{\sqrt{p_K}}$. Let $0<\xi<1/2$. Suppose there is a unitary circuit $A_w$ satisfying
\begin{equation}
    A_w\ket{0}
    =
    \frac{1}{\sqrt{2}}
    \left(
        \ket{\widetilde F_{0,M}}\ket{0}_Y
        +
        \ket{\widetilde F_{w,M}}\ket{1}_Y
    \right),
    \label{eq:coherent-preparation-Aw}
\end{equation}
where the second branch satisfies
\begin{equation}
    \left\|
        P^{\otimes n}_D\ket{\widetilde F_{w,M}}
        -
        P^{\otimes n}_D\ket{F_w}
    \right\|
    \leq
    \frac{\xi\sqrt{p_K}}{4}.
    \label{eq:high-degree-prep-error}
\end{equation}
Then there is a quantum procedure that, with probability at least $1-\xi^2/16$, prepares a state $\ket{\widetilde F_{w,M}^{B_K}}$ satisfying
\begin{equation*}
    \left\|
        \ket{\widetilde F_{w,M}^{B_K}}
        -
        \ket{F_w^{B_K}}
    \right\|
    \leq
    \xi,
\end{equation*}
using $O\!\left(K^{1/4}\log\frac{1}{\xi}\right)$ applications of $A_w$ and $A_w^\dagger$, together with $O\!\left(K^{1/4}\log\frac{1}{\xi}\,n\bigl(\log M+\log(nD)\bigr)\right)$ additional elementary gates.
\end{theorem}

\begin{proof}
We first construct a reversible circuit for recognizing the degree band. Let
\[
\ell:=\left\lceil\log\bigl(n(D-1)+1\bigr)\right\rceil.
\]
Each physical register $\alpha_i$ has $\log M$ bits. We first compare every $\alpha_i$ with $D$, at a total cost of $O(n\log M)$ gates, and proceed only if $\alpha\in[D]_0^n$. Then an $\ell$-qubit register is sufficient to store $|\alpha|=\sum_{i=1}^n \alpha_i$ for every $\alpha\in[D]_0^n$. Starting from
$\ket{\alpha}\ket{0^\ell}$, we reversibly compute the total degree via
\begin{equation*}
    U_{\mathrm{deg}}:
    \ket{\alpha}\ket{0^\ell}
    \longmapsto
    \ket{\alpha}\ket{|\alpha|}.
\end{equation*}
Using a ripple-carry adder~\cite{cuccaro2004new}, each addition into the $\ell$-bit accumulator costs $O(\ell)$ elementary reversible gates.
Applying this to the $n$ registers $\alpha_1,\ldots,\alpha_n$ therefore costs $O(n\ell)=O\!\left(n\bigl(\log M+\log(nD)\bigr)\right)$ gates.

We next reversibly check whether $K\leq |\alpha|\leq 2K$. Comparisons with the fixed integers $K$ and $2K$ can be implemented using $O(\ell)$ additional elementary gates. We then combine this band-membership bit with the label qubit $Y$ and flip a flag qubit precisely when both
conditions are satisfied. After uncomputing all temporary registers, this gives a unitary $U_{\mathrm{good}}$ satisfying
\begin{equation*}
    U_{\mathrm{good}}
    \ket{\alpha}\ket{y}_Y\ket{0}_G
    =
    \ket{\alpha}\ket{y}_Y
    \ket{
        y\cdot\mathbf 1_{\{K\leq|\alpha|\leq2K\}}
    }_G.
\end{equation*}
Equivalently, the good-subspace marker first checks whether
$\alpha\in[D]_0^n$. Conditioned on this check succeeding, it computes
$|\alpha|=\sum_i\alpha_i$ and marks the state precisely when
$K\leq|\alpha|\leq2K$ and $Y=1$. Accordingly, with the internal QHT workspace made explicit, the good projector is
\begin{equation}
    \Pi_{\mathrm{good}}
    :=
    P^{\otimes n}_D\Pi_{B_K}P^{\otimes n}_D 
    \otimes\ket1\!\bra1_Y 
    =
    \Pi_{B_K}\otimes\ket1\!\bra1_Y,
\end{equation}
where the final equality follows from $D\geq2K+1$. As above, \(P_D\) denotes the projector \(P_D\otimes\ket{0}\bra{0}_{\mathrm{anc}}\) on the logical degree register together with the clean internal QHT workspace. For convenience, we also suppress the factor $\ket{0}\bra{0}_{\mathrm{anc}}^{\otimes n}$. The total gate complexity of $U_{\mathrm{good}}$, including the computation and uncomputation of the degree register, is $O\!\left(n\bigl(\log M+\log(nD)\bigr)\right)$.

We now apply this marking circuit to the coherently prepared state in~\Cref{eq:coherent-preparation-Aw}. Define $ \widetilde p_K:=\left\|\Pi_{B_K}\ket{\widetilde F_{w,M}}\right\|^2$. Then the component marked as good is
\begin{equation*}
    \Pi_{\mathrm{good}}A_w\ket{0}
    =
    \frac{1}{\sqrt2}
    \Pi_{B_K}\ket{\widetilde F_{w,M}}
    \ket{1}_Y,
\end{equation*}
and hence the total success probability is
\begin{equation*}
    \widetilde q_K
    :=
    \norm{\Pi_{\mathrm{good}}A_w\ket{0}}^2
    =
    \frac{\widetilde p_K}{2}.
\end{equation*}
The factor $1/2$ is precisely the probability mass of the desired phase-encoded branch of the quantum example state. Since $\Pi_{B_K}$ is an orthogonal
projector,~\Cref{eq:high-degree-prep-error} gives
\begin{align}
    \left\|
        \Pi_{B_K}\ket{\widetilde F_{w,M}}
        -
        \Pi_{B_K}\ket{F_w}
    \right\|
    &=
    \left\|
        \Pi_{B_K}P^{\otimes n}_D
        \left(
            \ket{\widetilde F_{w,M}}
            -
            \ket{F_w}
        \right)
    \right\|
    \nonumber\\
    &\leq
    \left\|
        P^{\otimes n}_D
        \left(
            \ket{\widetilde F_{w,M}}
            -
            \ket{F_w}
        \right)
    \right\|
    \nonumber\\
    &\leq
    \frac{\xi\sqrt{p_K}}{4}.
    \label{eq:projected-preparation-error}
\end{align}
Therefore, by the reverse triangle inequality,
\begin{align*}
    \sqrt{\widetilde p_K}
    &=
    \norm{
        \Pi_{B_K}\ket{\widetilde F_{w,M}}
    }
    \nonumber\\
    &\geq
    \norm{
        \Pi_{B_K}\ket{F_w}
    }
    -
    \norm{
        \Pi_{B_K}\ket{\widetilde F_{w,M}}
        -
        \Pi_{B_K}\ket{F_w}
    }
    \nonumber\\
    &\geq
    \left(1-\frac{\xi}{4}\right)\sqrt{p_K}.
\end{align*}
By~\Cref{lem:sign-band},
\begin{equation*}
    p_K
    =
    \sum_{k=K}^{2K}|a_k|^2
    =
    \Theta(K^{-1/2}),
\end{equation*}
where only odd degrees contribute. Consequently,
\begin{equation*}
    \sqrt{\widetilde q_K}
    =
    \sqrt{\frac{\widetilde p_K}{2}}
    \geq
    \left(1-\frac{\xi}{4}\right)
    \sqrt{\frac{p_K}{2}}
    =
    \Omega(K^{-1/4}).
\end{equation*}
We now use this known lower bound when choosing the fixed-point amplification schedule. The normalized good component is
\begin{equation*}
    \ket{\Psi_{\mathrm{good}}}
    :=
    \frac{
        \Pi_{B_K}\ket{\widetilde F_{w,M}}
    }{
        \sqrt{\widetilde p_K}
    }
    \ket{1}_Y.
\end{equation*}
By fixed-point amplitude amplification~\cite{YoderLowChuang2014}, if the
good amplitude is at least $\delta$, a state arbitrarily close to
$\ket{\Psi_{\mathrm{good}}}$ can be prepared using
$O(\delta^{-1}\log(1/\epsilon_{\mathrm{AA}}))$ applications of the
preparation circuit and its inverse. We identify pure states up to an irrelevant global phase. Let $\delta=\Omega(K^{-1/4})$ and
$\epsilon_{\mathrm{AA}}:=\xi/4$. After fixing the irrelevant global phase of the output state, we therefore obtain a state
$\ket{\Omega_{\mathrm{out}}}$ satisfying
\begin{equation}
    \left\|
        \ket{\Omega_{\mathrm{out}}}
        -
        \ket{\Psi_{\mathrm{good}}}
    \right\|
    \leq
    \frac{\xi}{4},
    \label{eq:aa-output-error}
\end{equation}
using
\begin{equation*}
    O\!\left(
        K^{1/4}\log\frac{1}{\xi}
    \right)
\end{equation*}
applications of $A_w$ and $A_w^\dagger$.

However, \(\ket{\Omega_{\mathrm{out}}}\) still contains a garbage component associated with the \(\ket{0}_Y\) branch. We remove this unwanted contribution by measuring the projector onto the joint good subspace \(\Pi_{\mathrm{good}}\). Equivalently, we apply $U_{\mathrm{good}}$ to a fresh flag qubit and measure the flag. Since
$\Pi_{\mathrm{good}}\ket{\Psi_{\mathrm{good}}}
=\ket{\Psi_{\mathrm{good}}}$,
\Cref{eq:aa-output-error} implies
\begin{equation}
    \left\|
        (I-\Pi_{\mathrm{good}})
        \ket{\Omega_{\mathrm{out}}}
    \right\|
    \leq
    \frac{\xi}{4}.
\end{equation}
Hence the measurement succeeds with probability at least $1-\frac{\xi^2}{16}$. Conditioned on obtaining the good outcome, let
\begin{equation}
    \ket{\Omega_{\mathrm{good}}}
    :=
    \frac{
        \Pi_{\mathrm{good}}\ket{\Omega_{\mathrm{out}}}
    }{
        \norm{\Pi_{\mathrm{good}}\ket{\Omega_{\mathrm{out}}}}
    }.
\end{equation}
We obtain
\begin{align}
    \left\|
        \ket{\Omega_{\mathrm{good}}}
        -
        \ket{\Psi_{\mathrm{good}}}
    \right\|
    &\leq
    2
    \left\|
        \Pi_{\mathrm{good}}
        \left(
            \ket{\Omega_{\mathrm{out}}}
            -
            \ket{\Psi_{\mathrm{good}}}
        \right)
    \right\|
    \nonumber\\
    &\leq
    2
    \left\|
        \ket{\Omega_{\mathrm{out}}}
        -
        \ket{\Psi_{\mathrm{good}}}
    \right\|
    \leq
    \frac{\xi}{2}.
    \label{eq:postselection-error}
\end{align}

By construction,
$\ket{\Omega_{\mathrm{good}}}$ lies exactly in $\operatorname{Ran}(\Pi_{B_K})
    \otimes
    \operatorname{span}\{\ket1_Y\}$. Therefore, it can be written uniquely as
\begin{equation}
    \ket{\Omega_{\mathrm{good}}}
    =
    \ket{\widetilde F_{w,M}^{B_K}}
    \ket1_Y
\end{equation}
for some normalized
$\ket{\widetilde F_{w,M}^{B_K}}
\in\operatorname{Ran}(\Pi_{B_K})$.
At this point, we remove the label register $\ket{1}_Y$.

It remains to compare the normalized good component with the ideal band state. Define
\begin{equation*}
    \ket{\Psi_{\mathrm{target}}}
    :=
    \frac{
        \Pi_{B_K}\ket{\widetilde F_{w,M}}
    }{
        \sqrt{\widetilde p_K}
    }.
\end{equation*}
Then
$\ket{\Psi_{\mathrm{good}}}
=\ket{\Psi_{\mathrm{target}}}\ket1_Y$.
Using~\Cref{eq:projected-preparation-error} gives
\begin{align*}
    \left\|
        \ket{\Psi_{\mathrm{target}}}
        -
        \ket{F_w^{B_K}}
    \right\|
    &\leq
    \frac{
        2
        \left\|
            \Pi_{B_K}\ket{\widetilde F_{w,M}}
            -
            \Pi_{B_K}\ket{F_w}
        \right\|
    }{
        \norm{\Pi_{B_K}\ket{F_w}}
    }
    \nonumber\\
    &\leq
    \frac{
        2\cdot(\xi\sqrt{p_K}/4)
    }{
        \sqrt{p_K}
    }
    =
    \frac{\xi}{2}.
\end{align*}
Since tensoring with $\ket1_Y$ preserves Euclidean distance,
\begin{equation}
    \left\|
        \ket{\Psi_{\mathrm{good}}}
        -
        \ket{F_w^{B_K}}\ket1_Y
    \right\|
    \leq
    \frac{\xi}{2}.
\end{equation}
Combining this with~\Cref{eq:postselection-error} and applying the
triangle inequality yields
\begin{equation}
    \left\|
        \ket{\Omega_{\mathrm{good}}}
        -
        \ket{F_w^{B_K}}\ket1_Y
    \right\|
    \leq
    \xi.
\end{equation}
Removing the now exactly fixed label register, therefore, gives
\begin{equation}
    \left\|
        \ket{\widetilde F_{w,M}^{B_K}}
        -
        \ket{F_w^{B_K}}
    \right\|
    \leq
    \xi,
\end{equation}
as claimed.
\end{proof}

\subsection{Isometry map}\label{sec:isometry-map}
The ideal band state is supported on Hermite degrees \(k\in B_K\),
where each degree-\(k\) component is the state
\(\ket{C_k(w)}\). The actual amplified state is within distance
$\xi$ of it. The purpose of this subsection is to decompress this
occupation-number representation into copies of the hidden direction
state \(\ket{w}=\sum_{i=1}^n w_i\ket{i}\). 

The isometry \(V_k\) introduced in~\Cref{eq:overview-Vk} is precisely the second-to-first quantization map implemented efficiently by the circuit of~\cite{liu2025lowdepthquantumsymmetrization}. Indeed, for a
multi-index \(\alpha=(\alpha_1,\ldots,\alpha_n)\) with
\(|\alpha|=k\), the entry \(\alpha_i\) records the occupation number of
mode \(i\), and the corresponding first-quantized state is the
normalized uniform superposition over all \(k\)-tuples having these
multiplicities:
\begin{equation}
    V_k: \ket{\alpha}
\longmapsto
\sqrt{\frac{\alpha!}{k!}}
\sum_{\substack{
    i=(i_1,\ldots,i_k)\in[n]^k:\\
    |\{r:i_r=j\}|=\alpha_j\ \forall j\in[n]
}}
\ket{i_1}\otimes\cdots\otimes\ket{i_k}.
\end{equation}
When \(V_k\) is applied to \(\ket{C_k(w)}\), the multinomial
coefficients cancel, yielding exactly $V_k\ket{C_k(w)}=\ket{w}^{\otimes k}$.

To handle the coherent superposition of degrees without implementing a
separate controlled circuit for every $k\in B_K$, we introduce an
additional blank symbol $\ket{\perp}$ orthogonal to
$\ket{1},\ldots,\ket{n}$. For each degree-$k$ component, we add
$2K-k$ blanks so that every component has the same total length $2K$.
We can then apply a single symmetrization circuit $V_{2K}$ from~\cite{liu2025lowdepthquantumsymmetrization} over the
alphabet $\{1,\ldots,n,\perp\}$. The blank symbols appear in a
superposition of output positions. Measuring only whether each output
register contains $\ket{\perp}$ reveals these positions without
measuring the nonblank labels. Since $k\geq K$, every outcome leaves
at least $K$ nonblank registers containing copies of $\ket{w}$, which
are then used for tomography. We formulate the complete procedure,
including these measurements and the selection of registers, as a
quantum channel.

\begin{theorem}[Efficient band conversion]
\label{thm:band-condensate-isometry}
Let $K\geq1$, $D\geq2K+1$, and $M\geq D$. There exists a
$w$-independent quantum channel $\mathcal C_{B_K}$ with output on
$(\mathbb C^n)^{\otimes K}$ such that
\begin{equation}
    \mathcal C_{B_K}\!\left(
        \ket{F_w^{B_K}}\!\bra{F_w^{B_K}}
    \right)
    =
    \bigl(\ket{w}\!\bra{w}\bigr)^{\otimes K}.
    \label{eq:band-isometry-ideal-action}
\end{equation}
The channel can be implemented using
\[
O\!\left(n\bigl(\log M+\log(nD)\bigr)\right)
+\widetilde O\!\left(n+K\log(n+1)\right)
\]
elementary gates and $O(K)$ single-qubit measurements, with no
additional example or membership queries. Here the $\widetilde O$
notation hides factors polylogarithmic in $K+1$ and $n+1$, but not
$\log M$.
\end{theorem}

\begin{proof}
We first recall the symmetrization circuit of Liu, Childs, and
Gottesman~\cite[Sections~IV--V]{liu2025lowdepthquantumsymmetrization}.
Consider an alphabet of $m$ labels and a multiplicity vector
$\beta=(\beta_1,\ldots,\beta_m)$ satisfying $|\beta|=N$.
Let $l_\beta\in[m]^N$ be the canonical string containing exactly
$\beta_i$ copies of label $i$. Their construction implements the map
\begin{equation*}
    \ket{\beta}
    \longmapsto
    \frac{1}{\sqrt{N!\beta!}}
    \sum_{\sigma\in S_N}
        \ket{\sigma(l_\beta)},
    \qquad
    \beta!:=\prod_{i=1}^m\beta_i!.
\end{equation*}
Every distinct string having multiplicities $\beta$ appears exactly
$\beta!$ times in the sum over permutations. Hence,
\begin{align*}
    \frac{1}{\sqrt{N!\beta!}}
    \sum_{\sigma\in S_N}\ket{\sigma(l_\beta)}
    &=
    \sqrt{\frac{\beta!}{N!}}
    \sum_{\substack{
        i\in[m]^N:\\
        |\{r:i_r=j\}|=\beta_j\ \forall j\in[m]
    }}
        \ket{i_1}\cdots\ket{i_N}
    \nonumber\\
    &=V_N^{(m)}\ket{\beta},
\end{align*}
where $V_N^{(m)}$ denotes the corresponding expansion map on an
$m$-symbol alphabet. In particular, the circuit works on arbitrary
superpositions of multiplicity vectors having the same total size
$N$.

For completeness, we briefly recall the gate complexity. The first
stage converts the multiplicity representation into a sorted list of
labels using reversible prefix sums and sorting networks on
$O(m+N)$ records of
$O(\log(m+1)+\log(N+1))$ bits, requiring
$\widetilde O(m+N)$ gates. The subsequent symmetrization circuit has
polylogarithmic depth, acts on $N$ label registers, and uses
$O(N\log(N+1))$ ancillary qubits. Thus the complete conversion uses
\begin{equation}
    \widetilde O\!\left(m+N\log(m+1)\right)
    \label{eq:fixed-k-conversion-complexity}
\end{equation}
elementary gates, with all additional logarithmic factors absorbed
into the $\widetilde O$ notation.

We now make all degrees in the band have the same total size.
Starting from a basis state $\ket{\alpha}$, we first compute its total
degree,
\begin{equation*}
    \ket{\alpha}\ket{0}
    \longmapsto
    \ket{\alpha}\ket{|\alpha|}.
\end{equation*}
As before, we first check reversibly that $\alpha\in[D]_0^n$,
using $O(n\log M)$ gates. Conditioned on this check, computing
$|\alpha|$ costs $O(n\log(nD))$ gates~\cite{cuccaro2004new}.
We also check whether $K\leq|\alpha|\leq2K$. On the valid band,
each entry $\alpha_i$ fits in $\lceil\log_2(2K+1)\rceil$ bits;
the remaining high-order bits are zero and do not participate
in the conversion.

Set $k:=|\alpha|$. We introduce one additional label $\perp$ and
reversibly compute its multiplicity to be $2K-k$. This takes
$O(\log(K+1))$ further gates on the valid band and defines the
padding isometry
\begin{equation*}
    E_K:\ket{\alpha}
    \longmapsto
    \ket{k}\ket{\alpha,2K-k},
    \qquad K\leq|\alpha|=k\leq2K.
\end{equation*}
Every padded multiplicity vector now has total size $2K$.
Therefore, a single application of $V_{2K}^{(n+1)}$ handles all
degrees simultaneously. Define
\begin{equation*}
    V_{B_K}:=
    \bigl(I_{\mathrm{deg}}\otimes V_{2K}^{(n+1)}\bigr)E_K.
\end{equation*}
The temporary workspace of the conversion circuit is uncomputed,
while the degree register is retained.

We next determine the action of this map on the ideal band state.
For $T\subseteq[2K]$ with $|T|=2K-k$, let $J_T$ denote the map that
inserts the blank symbol $\ket{\perp}$ at the positions in $T$ while
preserving the order of the remaining $k$ registers. For a fixed
multi-index $\alpha$, the coefficient of each compatible output
string is
\begin{equation*}
    \sqrt{\frac{k!}{\alpha!}}\,w^\alpha
    \sqrt{\frac{\alpha!(2K-k)!}{(2K)!}}
    =
    \binom{2K}{k}^{-1/2}w^\alpha.
\end{equation*}
Summing over all nonblank strings and all possible blank positions
therefore gives
\begin{equation*}
    V_{2K}^{(n+1)}
    \sum_{|\alpha|=k}
        \sqrt{\frac{k!}{\alpha!}}\,w^\alpha
        \ket{\alpha,2K-k}
    =
    \binom{2K}{k}^{-1/2}
    \sum_{\substack{T\subseteq[2K]\\|T|=2K-k}}
        J_T\ket{w}^{\otimes k}.
\end{equation*}
The summands corresponding to distinct $T$ are mutually orthogonal,
because $\ket{\perp}$ is orthogonal to every nonblank label.
Consequently, the right-hand side is normalized.

Applying the complete map to the ideal band state gives
\begin{equation*}
    V_{B_K}\ket{F_w^{B_K}}
    =
    \frac{1}{\sqrt{p_K}}
    \sum_{k=K}^{2K}
        \frac{a_k}{\sqrt{\binom{2K}{k}}}
        \ket{k}
        \sum_{\substack{T\subseteq[2K]\\|T|=2K-k}}
            J_T\ket{w}^{\otimes k}.
\end{equation*}

We now determine which output registers are blank. On each of the
$2K$ registers, perform the two-outcome projective measurement
\begin{equation*}
    \left\{
        \ket{\perp}\!\bra{\perp},\,
        \sum_{i=1}^n\ket{i}\!\bra{i}
    \right\}.
\end{equation*}
This can be implemented by reversibly computing a blank indicator
into an ancillary qubit and measuring only that qubit. Hence, the
measurement reveals whether a register contains $\ket{\perp}$ without
revealing its nonblank label. Conditional on degree $k$, the observed
blank set $T$ is uniform among subsets of size $2K-k$. More explicitly,
on the ideal band state,
\begin{equation*}
    \Pr[T]
    =
    \frac{|a_{2K-|T|}|^2}
         {p_K\binom{2K}{|T|}},
    \qquad 0\leq|T|\leq K.
\end{equation*}
Conditioned on any observed $T$, the remaining nonblank registers
contain exactly
\[
    \ket{w}^{\otimes(2K-|T|)}.
\]

Since $|T|\leq K$, there are always at least $K$ nonblank registers.
For each observed pattern $T$, we use a classically controlled
permutation to move any $K$ nonblank registers to the first $K$
output positions and discard all remaining registers. Thus every
measurement outcome produces exactly $\ket{w}^{\otimes K}$; no
postselection is required. A permutation of $2K$ registers can be
implemented using $O(K)$ register swaps, so this step uses
$O(K\log(n+1))$ gates. The blank-indicator measurements have the same
gate bound and require $2K$ single-qubit measurements.

The resulting procedure defines the channel $\mathcal C_{B_K}$, and
the preceding calculation proves~\Cref{eq:band-isometry-ideal-action}.
To define the channel on the entire physical input space, we use the
validity check above and output a fixed state, for example
$\ket{1}^{\otimes K}$, whenever the input lies outside the valid band.
This makes $\mathcal C_{B_K}$ trace-preserving on all inputs without
affecting its action on the band. The state produced by
\Cref{thm:high-degree-projection}, conditioned on success, is supported
entirely on the valid band.

Finally, substituting $N=2K$ and $m=n+1$ into
\Cref{eq:fixed-k-conversion-complexity}, the single symmetrization
circuit costs
$\widetilde O(n+K\log(n+1))$ gates. Adding the validity check, degree
computation, padding, blank measurements, and register selection gives
the stated total complexity. All these operations are independent of
$w$ and require no additional oracle queries.
\end{proof}

\begin{corollary}[Robustness of the band conversion]
\label{cor:band-isometry-robustness}
Suppose the state produced by~\Cref{thm:high-degree-projection}
satisfies
\begin{equation}
    \left\|
        \ket{\widetilde F_{w,M}^{B_K}}
        -
        \ket{F_w^{B_K}}
    \right\|
    \leq\xi.
\end{equation}
Let $\rho_K$ be the output of $\mathcal C_{B_K}$ on this state,
averaged over all internal measurement outcomes. Then
\begin{equation*}
    \frac12
    \left\|
        \rho_K
        -
        \bigl(\ket{w}\!\bra{w}\bigr)^{\otimes K}
    \right\|_1
    \leq\xi.
\end{equation*}
\end{corollary}

\begin{proof}
By~\Cref{eq:band-isometry-ideal-action}, contractivity of trace
distance under the channel $\mathcal C_{B_K}$, and the pure-state
trace-distance bound, we have
\begin{align}
    \frac12\left\|
        \rho_K-
        \bigl(\ket{w}\!\bra{w}\bigr)^{\otimes K}
    \right\|_1
    &\leq
    \frac12\left\|
        \ket{\widetilde F_{w,M}^{B_K}}
        \!\bra{\widetilde F_{w,M}^{B_K}}
        -
        \ket{F_w^{B_K}}\!\bra{F_w^{B_K}}
    \right\|_1
    \nonumber\\
    &\leq
    \left\|
        \ket{\widetilde F_{w,M}^{B_K}}
        -
        \ket{F_w^{B_K}}
    \right\|
    \leq\xi.
\end{align}
The bound applies to the complete conversion channel, including
the measurement-dependent register selection, without conditioning
on an individual blank pattern.
\end{proof}
\subsection{Pure-state tomography}\label{sec:pure-state-tomography}
At this point, the band-conversion channel produces a state $\rho_K$
on $K$ registers, averaged over all blank-indicator measurement
outcomes. By the robustness guarantee of the previous subsection,
$\rho_K$ is close in trace distance to the ideal product state
$\ket{w}\!\bra{w}^{\otimes K}$. Our final task is therefore to recover a classical description of $w$ from these approximately prepared copies. We use the gate-efficient tomography procedure from~\cite{Haah_2023}, which shows that $K=\Theta(n/\eps^2)$ copies suffice to recover the state vector of a pure state up to Euclidean error $O(\eps)$ up to global phase with constant probability. Moreover, because the actual $K$-register input is close in trace distance to the ideal product state, the success probability of the same tomography measurement deteriorates by at most this input error. Finally, we use three additional example queries to learn the global sign.

\begin{theorem}[Robust gate-efficient pure-state tomography]
\label{thm:pure-state-tomography}
Assume $n\geq2$; the case $n=1$ is trivial. Let $w\in\mathbb R^n$ be a unit vector, and let $\rho_K$ be a state on $(\mathbb C^n)^{\otimes K}$ satisfying
\begin{equation}
    \frac12
    \left\|
        \rho_K
        -
        \ket{w}\!\bra{w}^{\otimes K}
    \right\|_1
    \leq
    \xi.
    \label{eq:tomography-input-error}
\end{equation}
There are universal constants $C,c>0$ such that, for every $\eps\in(0,1)$,
$K\geq Cn/\eps^2$, and $M\geq cn^2/\eps^2$, there is a tomography procedure
that, followed by three additional queries to $\QEX_M(f_w,\gamma_n)$,
outputs a unit vector
$\widetilde w\in\mathbb R^n$ satisfying
\begin{equation*}
    \left\|
        \widetilde w-w
    \right\|
    \leq
    \eps
\end{equation*}
with probability at least $73/96-\xi$ using $O\!\left(\frac{n}{\eps^2}\log\frac{1}{\eps}\polylog\, n\right)$ gates, together with
$\poly(n,1/\eps)$ classical post-processing time.
\end{theorem}
\begin{proof}
We first consider the ideal input $\ket{w}\!\bra{w}^{\otimes K}$. We aim to use the pure state tomography procedure from~\cite{Haah_2023}. Although~\cite[Theorem~C.1]{Haah_2023} is stated for power-of-two dimensions, we can pad $\ket{w}$ with zeros to the next power-of-two dimension and, after tomography, project and renormalize the estimate back onto the first $n$ coordinates. The output vector does not worsen the overlap with $\ket{w}$ and only changes the complexity by at most a constant factor. In particular, for every tomography accuracy parameter $\tau>0$, using $O\!\left(\frac{n}{\tau}\right)$ copies of an unknown pure state $\ket{w}$, there is a tomography algorithm that outputs a unit vector $\ket{v}$ satisfying
\begin{equation*}
    |\langle v,w\rangle|
    \geq
    1-\tau
\end{equation*}
except with probability at most $n^{-100}$, using $O\!\left(\frac{n}{\tau}\log\frac1\tau\,\polylog\ n\right)$ gates. The resulting classical post-processing time is therefore $\poly(n,1/\eps)$.

We choose $\tau:=\eps^2/1152$. For unit vectors $v,w$, we have
\begin{align*}
    \min_{|\zeta|=1}
    \|v-\zeta w\|^2
    =
    2\left(1-|\langle v,w\rangle|\right)\leq
    2\tau.
\end{align*}
Hence there exists a phase $\zeta=e^{i\phi}$ such that
\begin{equation}
    \|v-\zeta w\|
    \leq
    \sqrt{2\tau}
    =
    \frac{\eps}{24}.
    \label{eq:tomography-phase-aligned-error}
\end{equation}

The output vector $v$ is complex, but we require a real vector. Write
$v=a+ib$ for $a,b\in\mathbb R^n$. Let $r$ be whichever of $a$ and $b$ has the larger Euclidean norm. Since $\|a\|^2+\|b\|^2=1$,
\begin{equation*}
    \|r\|
    \geq
    \frac1{\sqrt2}.
\end{equation*}
If $r=a$, let $c=\cos\phi$, and if $r=b$, let $c=\sin\phi$.~\Cref{eq:tomography-phase-aligned-error} implies $\|r-cw\|\leq\frac{\eps}{24}$ and $|c|\geq\|r\|-\frac{\eps}{24}>\frac12$. Define the real unit vector
\begin{equation*}
    \bar w
    :=
    \frac{r}{\|r\|}.
\end{equation*}
Let $s=\sgn(c)$. Using the reverse triangle inequality,
\begin{align*}
    \|\bar w-sw\|
    \leq
    \frac{2\|r-cw\|}{|c|}
    \leq
    \frac{\eps}{6}.
\end{align*}
Therefore, on the ideal input state $\ket{w}\bra{w}^{\otimes K}$,
\begin{equation*}
    \Pr\!\left[
        \min_{s\in\{\pm1\}}
        \|\bar w-sw\|
        \leq
        \frac{\eps}{6}
    \right]
    \geq
    1-n^{-100}.
\end{equation*}

We next transfer the guarantee to the actual input $\rho_K$ that is $\xi$-close in trace distance to $\ket{w}\bra{w}^{\otimes K}$. By~\Cref{eq:tomography-input-error} and contractivity of trace distance under quantum channels, the total variation distance between the output distribution obtained from $\rho_K$ and that obtained from the ideal input $\ket{w}\!\bra{w}^{\otimes K}$ is at most $\xi$. Consequently,
\begin{equation}
    \Pr_{\rho_K}\!\left[
        \min_{s\in\{\pm1\}}
        \|\bar w-sw\|
        \leq
        \frac{\eps}{6}
    \right]
    \geq
    1-n^{-100}-\xi.
    \label{eq:robust-tomography-good-event}
\end{equation}

It remains to resolve the global phase $s$. We use three additional calls to $\QEX_M(f_w,\gamma_n)$ and measure the example and label registers in the computational basis, obtaining independent pairs $(\bar X_j,Y_j)$, where $\bar X_j\sim\gamma_{n,M}$ and $Y_j=f_w(\bar X_j)$. Since $\bar w$ is known classically, sample $j$ votes for $\bar w$ if $f_{\bar w}(\bar X_j)=Y_j$ and for $-\bar w$ otherwise. We define $\widetilde w$ to be the majority vote.

Condition on the event in~\Cref{eq:robust-tomography-good-event}, and let $s_\star\in\{\pm1\}$ satisfy $\|s_\star\bar w-w\|\leq\frac{\eps}{6}$.
Under the continuous Gaussian distribution, the probability that we obtain the wrong sign because of a bad sample is bounded using~\Cref{lem:gaussian-disagreement} as follows,
\begin{equation*}
    d_{\gamma_n}
    \left(
        f_{s_\star\bar w},
        f_w
    \right)
    \leq
    \frac{2}{\pi}
    \arcsin\!\left(\frac{\eps}{12}\right)
    \leq
    \frac{\eps}{12}.
\end{equation*}

To compare the discretized and continuous disagreement probabilities, define the piecewise-constant density
\[
    \widetilde\gamma_M(x)
    :=Z^{-1}\gamma_n(\bar x)\mathbf 1_{[-L,L]^n}(x).
\]
The Gaussian-amplitude estimates used in~\Cref{app:discretization} and the reverse triangle inequality give
\[
    \dtv(\widetilde\gamma_M,\gamma_n)
    \leq
    \left\|\sqrt{\widetilde\gamma_M}-\sqrt{\gamma_n}\right\|
    \leq
    \frac{16n}{\sqrt M}+2\sqrt{2n}\,e^{-\pi M/4}.
\]
Moreover, $\|x-\bar x\|\leq\Delta\sqrt n$. If $f_u(x)\neq f_u(\bar x)$, then $|\langle u,x\rangle|\leq\Delta\sqrt n$, and hence, for $x\sim\gamma_n$,
\[
    \Pr[f_u(x)\neq f_u(\bar x)]
    \leq
    \Pr_{z\sim \calN(0,1)}[|z|\leq\Delta\sqrt n]
    \leq
    \sqrt{\frac{2}{\pi}}\,\Delta\sqrt n.
\]
Applying the same bound to $v$, using $\Delta=\sqrt{4\pi/M}$, and taking a union bound yields
\begin{equation}
    d_{\gamma_{n,M}}(f_u,f_v)
    \leq
    d_{\gamma_n}(f_u,f_v)
    +\frac{16n}{\sqrt M}
    +4\sqrt{\frac{2n}{M}}
    +2\sqrt{2n}\,e^{-\pi M/4}
    \leq d_{\gamma_n}(f_u,f_v)+\frac{\eps}{12},
    \label{eq:discrete-disagreement-bound}
\end{equation}
uniformly over unit vectors $u,v$, where the last inequality holds for the chosen $M \geq 576^2\cdot\frac{n^2}{\eps^2}$.

We must also account for grid points satisfying $\langle u,\bar x\rangle=0$. For every unit vector $u$, choose $i$ such that $|u_i|\geq1/\sqrt n$. Conditioned on all coordinates of $\bar x$ except $\bar x_i$, the equation $\langle u,\bar x\rangle=0$ determines at most one possible grid value of $\bar x_i$. Every point mass of the one-dimensional discretized Gaussian has probability $O(\Delta)=O(M^{-1/2})$. Hence, uniformly over unit $u$, for sufficiently large $M\geq c\cdot n^2/\eps^2$,
\[
\Pr_{\bar x\sim\gamma_{n,M}}[\langle u,\bar x
\rangle=0]\leq\frac{\eps}{12}.
\]
Away from this tie event, an incorrect orientation vote implies $f_{s^\star\bar w}(\bar x)\neq f_w(\bar x)$. For each of the three independent samples, the corresponding orientation vote is therefore wrong with probability at most $\eps/12+\eps/12+\eps/12=\eps/4\leq1/4$. Thus, conditioned on~\Cref{eq:robust-tomography-good-event}, the majority vote is wrong with probability at most
\[
    3(1/4)^2(3/4)+(1/4)^3=5/32.
\]
Combining this with~\Cref{eq:robust-tomography-good-event}, the probability that tomography succeeds and the correct orientation is selected is at least $1-n^{-100}-\xi-5/32$. For $n\geq2$, we have $n^{-100}\leq1/12$, so this probability is at least $73/96-\xi$. On this event, $\|\widetilde w-w\|\leq\eps/6\leq\eps$. The sign-resolution step uses exactly three additional calls to $\QEX_M(f_w,\gamma_n)$, completing the proof.
\end{proof}
\begin{remark}[Global sign]
The tomography procedure above initially recovers $w$ only up to a global sign. This ambiguity is inherent in the normalized band state: since the sign function has support only on odd Hermite degrees, $ \ket{C_k(-w)}=(-1)^k\ket{C_k(w)}=-\ket{C_k(w)}$ for every contributing degree $k$, and hence $\ket{F_{-w}^{B_K}}=-\ket{F_w^{B_K}}$. Thus, the band state itself cannot distinguish $w$ from $-w$. However, before projecting onto the band, this sign is retained as a
relative phase between the $\ket{0}_Y$ and $\ket{1}_Y$ branches of the coherent state in~\Cref{eq:flagged-hermite-state}: replacing $w$ by $-w$ flips the phase of the $\ket{1}_Y$ branch while leaving the $\ket{0}_Y$ branch unchanged. In principle, this relative phase could be used to recover the orientation coherently. For simplicity,~\Cref{thm:pure-state-tomography} instead resolves the ambiguity using three additional quantum-example queries.
\end{remark}
\subsection{Putting everything together}\label{sec:example-putting-together}
\begin{theorem}[Learning homogeneous LTFs with Example Oracle]
\label{thm:quantum-example-learning}
Assume $n\geq2$; the case $n=1$ is trivial and can be handled separately. Let $f_w(x):=\sgn(\langle w,x\rangle)$ for some $w\in\mathbb \sphere^{n-1}$, and let $\eps\in(0,1/2)$. Set $K:=\left\lceil c_1 n/\eps^2\right\rceil$, $D:=2K+1$, and $M:=2^{\left\lceil\log_2\!\left(c_2n^{13/4}K^{49/16}\right)\right\rceil}$, where $c_1,c_2>0$ are sufficiently large universal constants. Given coherent access to $\QEX_M(f_w,\gamma_n)$ and its inverse, there exists a quantum algorithm that outputs a unit vector
$\widetilde w\in\mathbb R^n$ such that
\[
    \|\widetilde w-w\|\leq\eps
\]
with probability at least $2/3$. The algorithm makes
\[
    O\!\left(
        \frac{n^{1/4}}{\sqrt{\eps}}
    \right)
\]
queries to $\QEX_M(f_w,\gamma_n)$ and its inverse, using $\widetilde O\!\left(\frac{n^{5/4}}{\sqrt{\eps}}+\frac{n}{\eps^2}\right)$ additional gates, where the $\widetilde O$ notation hides factors polylogarithmic in $n$ and $1/\eps$, together with $\operatorname{poly}(n,1/\eps)$ classical post-processing time.
\end{theorem}
\begin{proof}
We combine the finite-dimensional Hermite preparation, high-degree projection, band isometry, and pure-state tomography procedures.

We first choose $M$ sufficiently large to apply~\Cref{thm:finite-hermite-state-preparation}. Let $\sgn(z) = \sum_r a_r h_r(z)$ be the Hermite expansion of the sign function. Its Hermite mass in degrees $K$ through $2K$ is $p_K:= \sum_{k =K}^{2K}\abs{a_k}^2$. By~\Cref{lem:sign-band}, we have $p_K = \Theta(K^{-1/2})$. Setting $\xi = 1/48$, we choose $\eta:=\frac{\xi\sqrt{p_K}}{4} = \frac{1}{192}\sqrt{p_K}$. Since $K = O(n/\eps^2)$ and $D=O(n/\eps^2)$, 
\begin{equation*}
    O\!\left(
        D^{9/4}
        \left(\frac{n}{\eta}\right)^{13/4}
    \right)
    =
    O\!\left(
        n^{13/4}K^{49/16}
    \right) = O(\poly(n, \frac{1}{\eps})).
\end{equation*}
Since $p_K=\Theta(K^{-1/2})$, we have $\eta^{-1}=\Theta(K^{1/4})$. Hence $\log(4n/\eta)=O(\log(n/\eps))$, whereas $D=\Theta(n/\eps^2)$. Thus, after increasing the universal constant $c_1$ if necessary, $D>\log(4n/\eta)$.
Consequently, for a sufficiently large universal constant $c_2$,
\[
    M
    =
    2^{\left\lceil
        \log_2
        \left(
            c_2n^{13/4}K^{49/16}
        \right)
    \right\rceil}
\]
satisfies all the hypotheses of~\Cref{thm:finite-hermite-state-preparation} and~\Cref{thm:pure-state-tomography}. Hence there is a circuit
$A_w$ using one query to $\QEX_M(f_w,\gamma_n)$ and $O\!\left(n\log^3M\log\frac{4n}{\eta}\right)$ additional gates such that its phase branch
$\ket{\widetilde F_{w,M}}$ satisfies
\begin{equation*}
    \left\|
        P^{\otimes n}_D\ket{\widetilde F_{w,M}}
        -
        P^{\otimes n}_D\ket{F_w}
    \right\|
    \leq
    \frac{\xi\sqrt{p_K}}{4}.
\end{equation*}
Since $D=2K+1$, we have $\Pi_{B_K}P^{\otimes n}
_D=\Pi_{B_K}$, and therefore
\begin{equation}
    \left\|
        \Pi_{B_K}\ket{\widetilde F_{w,M}}
        -
        \Pi_{B_K}\ket{F_w}
    \right\|
    \leq
    \frac{\xi\sqrt{p_K}}{4}.
\end{equation}
Thus the hypothesis of~\Cref{thm:high-degree-projection} is satisfied. 

Applying that theorem with $\xi=1/48$ prepares the normalized band state $\ket{\widetilde F_{w,M}^{B_K}}$ such that
\begin{equation}
    \left\|
        \ket{\widetilde F_{w,M}^{B_K}}
        -
        \ket{F_w^{B_K}}
    \right\|
    \leq
    \xi
    =
    \frac1{48}.
    \label{eq:final-band-state-error}
\end{equation}
Since $\xi$ is a fixed constant, the high-degree projection uses $ O(K^{1/4}) = O(\frac{n^{1/4}}{\sqrt{\eps}})$ applications of $A_w$ and $A_w^\dagger$. Each such application uses one query to $\QEX_M(f_w,\gamma_n)$ or its inverse. The gate cost at this stage is
\begin{align}
    O\!\left[
        K^{1/4}
        \left(
            n\log^3M
            \log\frac{4n}{\eta}
            +
            n\bigl(\log M+\log(nD)\bigr)
        \right)
    \right]=
    \widetilde O(nK^{1/4}) = \widetilde O\p{\frac{n^{5/4}}{\sqrt{\eps}}}.
    \label{eq:final-projection-gate-cost}
\end{align}

We next apply the $w$-independent channel $\mathcal C_{B_K}$
from~\Cref{thm:band-condensate-isometry}. In particular, let
\begin{equation*}
    \rho_K
    :=
    \mathcal C_{B_K}\!\left(
        \ket{\widetilde F_{w,M}^{B_K}}
        \!\bra{\widetilde F_{w,M}^{B_K}}
    \right).
\end{equation*}
By~\Cref{cor:band-isometry-robustness} and
\Cref{eq:final-band-state-error},
\begin{equation*}
    \frac12
    \left\|
        \rho_K
        -
        \bigl(\ket{w}\!\bra{w}\bigr)^{\otimes K}
    \right\|_1
    \leq
    \frac1{48}.
\end{equation*}
This is an unconditional bound over the conversion's internal
measurement outcomes; no additional postselection or
success-probability loss is introduced.
By~\Cref{thm:band-condensate-isometry}, this stage makes no
additional oracle queries and uses
\begin{equation}
    O\!\left(
        n\bigl(\log M+\log(nD)\bigr)
    \right)
    +
    \widetilde O\!\left(
        n+K\log(n+1)
    \right)
    =
    \widetilde O\!\left(\frac{n}{\eps^2}\right)
    \label{eq:final-isometry-gate-cost}
\end{equation}
elementary gates, where we used $K=\Theta(n/\eps^2)$ and
$\log M=O(\log(n/\eps))$.

Finally, applying~\Cref{thm:high-degree-projection} with $\xi=1/48$ succeeds with probability at least $1-\frac{\xi^2}{16}=1-\frac{1}{36864}.$ Conditioned on this event, the resulting band state is within distance $\xi=1/48$ of the ideal band state. Hence, by~\Cref{thm:pure-state-tomography}, the subsequent tomography and orientation-recovery procedure succeeds with probability at least $ \frac{73}{96}-\frac{1}{48}=\frac{71}{96}$.
Therefore, the overall success probability is at least
\begin{equation}
    \left(1-\frac{1}{36864}\right)
    \frac{71}{96}
    >
    \frac{2}{3}.
\end{equation} 
The tomography measurements use
\begin{equation}
    O\!\left(
        \frac{n}{\eps^2}
        \log\frac1\eps\,
        \polylog\, n
    \right)
    \label{eq:final-tomography-gate-cost}
\end{equation}
gates, together with $\poly(n,1/\eps)$ classical post-processing time.

The finite-dimensional Hermite preparation uses one quantum-example query per application of $A_w$ or $A_w^\dagger$. The high-degree
projection uses $O(K^{1/4})$ such applications, while the isometry and tomography steps use no further oracle queries. The final
orientation test uses three additional queries. Hence, total query complexity is
\begin{align*}
    O(K^{1/4})+3
    =
    O\!\left(
        \frac{n^{1/4}}{\sqrt{\eps}}
    \right).
\end{align*}
Summing
\Cref{eq:final-projection-gate-cost},
\Cref{eq:final-isometry-gate-cost}, and
\Cref{eq:final-tomography-gate-cost} gives the total gate complexity
\begin{align*}
    \widetilde O\!\left(
        nK^{1/4}
        +n
        +K\log(n+1)
        +\frac{n}{\eps^2}
    \right)
    &=
    \widetilde O\!\left(
        \frac{n^{5/4}}{\sqrt{\eps}}
        +\frac{n}{\eps^2}
    \right).
\end{align*}
This proves the theorem.
\end{proof}

\subsection{A matching lower bound}
We prove a lower bound for learning homogeneous halfspaces under the fixed standard Gaussian distribution by combining Assouad's hypercube method~\cite{assouad1983deux, chen2024assouad} with a Gaussian surface-area bound on noise sensitivity~\cite{KOS08}. The key idea is to construct a family of concepts indexed by hidden bit strings, such that accurate learning requires recovering most of the bits, while neighboring concepts obtained by perturbing a single coordinate remain difficult to distinguish. For Gaussian halfspaces, we bound the disagreement between neighboring concepts using Gaussian surface area and noise sensitivity~\cite{KOS08}. Combining this geometric bound with a quantum query-hybrid argument and Assouad's testing inequality yields a dimension-dependent lower bound.

We begin with the definition of Gaussian surface area and how it bounds the noise sensitivity.
\begin{definition}[Gaussian surface area]
For a Borel set $K\subseteq\mathbb R^n$, let $K_t:=\{x\in\mathbb R^n:\operatorname{dist}(x,K)\leq t\}$. Its Gaussian surface area is
\begin{equation*}
    \Gamma(K):=
    \liminf_{t\to 0}
    \frac{\gamma_n(K_t\setminus K)}{t}.
\end{equation*}
\end{definition}

\begin{definition}[Gaussian noise sensitivity]
For $\eta\in[0,1]$, let $X,G\sim\gamma_n$ be independent and set $Y=(1-\eta)X+\sqrt{2\eta-\eta^2}\,G$. The Gaussian noise sensitivity of $K$ is
\begin{equation}
\label{eq:ns-definition}
    \NS_\eta(K)
    :=\Pr[\mathbf 1_K(X)\neq\mathbf 1_K(Y)].
\end{equation}
\end{definition}
Thus noise sensitivity is precisely the probability of a label change under a correlated Gaussian perturbation of the input. In~\Cref{lem:assouad-edge}, we identify the disagreement between neighboring target halfspaces with this perturbation error for a fixed one-dimensional threshold. This identification does not introduce label noise into the learning problem.

The following bound uses the correlation convention in~\Cref{eq:ns-definition}; see \cite[Corollary~14]{KOS08}.
\begin{corollary}[Surface area controls noise sensitivity]
\label{cor:gsa-noise-sensitivity}
For every Borel set $K\subseteq\mathbb R^n$ and $\eta\in[0,1]$,
\begin{equation*}
    \NS_\eta(K)\leq\sqrt\pi\,\Gamma(K)\sqrt\eta.
\end{equation*}
\end{corollary}

To establish the lower bound, we consider a hard family of concepts indexed by unknown sign strings $z\in\{-1,+1\}^m$, where $m=n-1$.
Fix $r=4\pi\eps\leq 1/2$. For each $z=(z_1,\ldots,z_m)\in\{-1,+1\}^m$, define
\begin{equation*}
    w_z
    =
    \sqrt{1-r^2}\,e_1
    +
    \frac{r}{\sqrt m}
    \sum_{j=1}^m z_j e_{j+1}.
\end{equation*}
Then, we denote $f_z(x)=\sgn\langle w_z,x\rangle$ and $\|w_z\|=1$. We thus obtain the restricted concept class
\begin{equation}
    \widetilde{\mathcal C}
    =
    \{f_z:z\in\{-1,+1\}^m\}.\label{eqn: hard_concept_class}
\end{equation}
Any learner that succeeds for all concepts of homogeneous LTFs must also succeed on $\widetilde{\mathcal C}$.

For an unknown sign string $z$, the unitary $U_z = \QEX(f_z, \gamma_n)$ prepares the coherent labeled Gaussian example state
\begin{equation}
\label{eq:assouad-quantum-example}
    U_z\ket{0}
    =
    \ket{\psi_z}
    =
    \int_{\mathbb R^n}
    \sqrt{\gamma_n(x)}
    \ket{x}\ket{\bar f_z(x)}\,dx,
\end{equation}
where $\bar f_z(x)=\frac{1-f_z(x)}{2}\in\{0,1\}$. With our computational-basis convention, the amplitudes of \(\ket{\psi_z}\) are real and non-negative. Hence, $\langle 0|\psi_z\rangle\in[0,1]$. Note that the learner may query $U_z$ and $U^{\dagger}_z$.

The construction ensures that the geometry of the concepts reflects the Hamming distance between their hidden strings. Indeed, for any $z,z'\in\{-1,+1\}^m$,
\begin{equation}
    \langle w_z,w_{z'}\rangle
    =
    1-\frac{2r^2}{m}\dHam(z,z'),
\end{equation}
where $\dHam(z,z')$ denotes the Hamming distance between two bit strings $z$ and $z'$. Hence \Cref{lem:gaussian-disagreement} gives
\begin{equation}
\label{eq:assouad-hamming-geometry}
    d_{\gamma_n}(f_z,f_{z'})
    =
    \frac{2}{\pi}
    \arcsin\left(
        r\sqrt{\frac{\dHam(z,z')}{m}}
    \right).
\end{equation}
This relationship will allow us to reduce accurate learning of $\widetilde{\mathcal C}$ to approximate recovery of the unknown sign string $z$.

Assouad's method typically proceeds in the following steps. First, we show that any learner achieving classification error at most $\epsilon$ must recover a constant fraction of the hidden string $z$, yielding an upper bound on the expected Hamming error. Second, we bound the disagreement between neighboring concepts $f_z$ and $f_{z^{(i)}}$, which differ in exactly one coordinate. Finally, Assouad's hypercube testing inequality converts this local indistinguishability into a lower bound on the expected Hamming error. Comparing the resulting lower and upper bounds establishes the desired dimension-dependent quantum query lower bound.

We start with the following lemma showing that achieving classification error at most $\epsilon$ implies recovering most of the hidden string $z$.
\begin{lemma}[Learning implies small expected Hamming error]
\label{lem:assouad-decoding}
Suppose a procedure $\mathcal A$ outputs a classical measurable classifier $h:\mathbb R^n\to\{-1,+1\}$ satisfying $\Pr_{\mathcal A}[d_{\gamma_n}(h,f_z)\leq\eps]\geq\frac23$, for $f_z\in \widetilde \calC$. Define
\begin{equation*}
    \widehat z(h)\in
    \operatorname*{arg\,min}_{v\in\{-1,+1\}^m}
    d_{\gamma_n}(h,f_v).
\end{equation*}
For uniform $z$,
\begin{equation*}
    \E_{z,\mathcal A}[\dHam(z,\widehat z)]\leq\frac{3m}{8}.
\end{equation*}
\end{lemma}

\begin{proof}
On the event $d_{\gamma_n}(h,f_z)\leq\eps$, the definition of the
decoder and the triangle inequality imply
\[
    d_{\gamma_n}(f_z,f_{\widehat z})
    \leq d_{\gamma_n}(f_z,h)+d_{\gamma_n}(h,f_{\widehat z})
    \leq2\eps.
\]
Writing $k=\dHam(z,\widehat z)$ and applying~\Cref{eq:assouad-hamming-geometry}, we obtain $\arcsin\!\left(r\sqrt{k/m}\right)\leq\pi\eps$. Therefore,
\[
    \frac{k}{m}
    \leq\frac{\sin^2(\pi\eps)}{r^2}
    \leq\frac{\pi^2\eps^2}{(4\pi\eps)^2}
    =\frac1{16}.
\]
On an unsuccessful execution the Hamming error is at most $m$.
Since success has probability at least $2/3$, for every fixed $z$,
\[
    \E_{\mathcal A}[\dHam(z,\widehat z)]
    \leq\frac23\cdot\frac m{16}+\frac13\cdot m
    =\frac{3m}{8}.
\]
Averaging over uniform $z$ proves the claim.
\end{proof}

Now we come to the key lemma of Assouad's method~\cite{yu1997festschrift}. For any $z\in\{-1,+1\}^m$ and $i\in[m]$, let $z^{(i)}$
denote the neighboring sign vector obtained by flipping the $i$-th coordinate of $z$, i.e.,
\[
    z^{(i)}_j =
    \begin{cases}
        -z_i, & j=i,\\
        z_j, & j\neq i.
    \end{cases}
\]
For probability measures $P,Q$, write $\dtv(P,Q)=\sup_A|P(A)-Q(A)|$.

\begin{lemma}[Assouad's Lemma]
\label{lem:assouad-testing}
Let $P_z$ be the distribution of an estimator
$\widehat z\in\{-1,+1\}^m$ when the hidden index is
$z\in\{-1,+1\}^m$. If
\[
    \max_{z\in\{-1,+1\}^m,\ i\in[m]}
    \dtv(P_z,P_{z^{(i)}})\leq\alpha,
\]
then, for uniform $z$,
\begin{equation*}
    \E[\dHam(z,\widehat z)]
    \geq\frac m2(1-\alpha).
\end{equation*}
\end{lemma}

\begin{proof}
Fix $i\in[m]$ and let
$A_i=\{\widehat z:\widehat z_i=+1\}$.
Pair every $z$ with $z_i=+1$ with its neighbor $z^{(i)}$.
For each such pair, the sum of the two coordinate-error probabilities is
\begin{align*}
    P_z(A_i^c)+P_{z^{(i)}}(A_i)
    &=1-\bigl(P_z(A_i)-P_{z^{(i)}}(A_i)\bigr)\\
    &\geq1-\dtv(P_z,P_{z^{(i)}})\\
    &\geq1-\alpha.
\end{align*}
There are $2^{m-1}$ pairs, and each vertex occurs in exactly one pair.
Averaging over uniform $z$ gives
\[
    \Pr[\widehat z_i\neq z_i]\geq\frac{1-\alpha}{2}.
\]
Summing over coordinates proves
\[
    \E[\dHam(z,\widehat z)]
    =\sum_{i=1}^m\Pr[\widehat z_i\neq z_i]
    \geq\frac m2(1-\alpha).
\]
\end{proof}

To apply this lemma, we bound the distinguishability between the estimator's output distributions \(P_z\) and \(P_{z^{(i)}}\), corresponding to two hidden sign strings that differ in exactly one coordinate. Here we employ the relation between Gaussian surface area and noise sensitivity in~\Cref{cor:gsa-noise-sensitivity}.
\begin{lemma}[Disagreement under a one-coordinate perturbation]
\label{lem:assouad-edge}
For every $z\in\{-1,+1\}^m$ and $i\in[m]$,
\begin{equation*}
    d_{\gamma_n}(f_z,f_{z^{(i)}})
    \leq\frac{4\pi\eps}{\sqrt m}.
\end{equation*}
\end{lemma}

\begin{proof}
Let $x\sim\gamma_n$ and define $u=\langle w_z,x\rangle$, $v=\langle w_{z^{(i)}},x\rangle$. We have
\[
    \E[uv]=\langle w_z,w_{z^{(i)}}\rangle
    =1-\frac{2r^2}{m}.
\]
Set $\eta=2r^2/m\in(0,1/2]$. For independent $g,g'\sim\mathcal N(0,1)$, we therefore have
\[
    (u,v)\overset{\mathrm d}=
    \bigl(g,(1-\eta)g+\sqrt{2\eta-\eta^2}\,g'\bigr),
\]
where $\overset{\mathrm d}=$ denotes equality in distribution. Taking $K_+=[0,\infty)\subseteq\mathbb R$, this gives the exact identity
\begin{equation*}
    d_{\gamma_n}(f_z,f_{z^{(i)}})
    =\Pr[\operatorname{sign}(u)\neq\operatorname{sign}(v)]
    =\NS_\eta(K_+).
\end{equation*}
From~\Cref{cor:gsa-noise-sensitivity}, we know that $\NS_\eta(K_+)$ is controlled by Gaussian surface area $\Gamma(K_+)$. Using the existing result~\cite{KOS08}, the Gaussian surface area of the halfspace $K_+$ is $\frac{1}{\sqrt{2\pi}}$. Thus, we have
\[
    d_{\gamma_n}(f_z,f_{z^{(i)}})
    \leq\sqrt\pi\,\frac1{\sqrt{2\pi}}
          \sqrt{\frac{2r^2}{m}}
    =\frac r{\sqrt m}
    =\frac{4\pi\eps}{\sqrt m}.
\]
\end{proof}
Recall that in our QEX model, the oracle $U_z$ is specified only through its action on the initial state, $U_z\ket{0}=\ket{\psi_z}$, while its action on the orthogonal complement remains unspecified. Consequently, the distance between two quantum example states does not, in general, control the operator-norm distance between arbitrary unitary extensions. This distinction is important for the quantum hybrid argument, which requires bounding $\|U_z-U_{z^{(i)}}\|$ for every pair of neighboring vertices in the Assouad hypercube. However, for our lower-bound argument, it suffices to construct a family of oracle unitaries where the bound holds. The following lemma establishes the existence of such a family using the standard Householder reflections.
\begin{lemma}\label{lem:canonical-unitary-extensions}
    Let \(\{|\psi_z\rangle\}_{z\in\{\pm 1\}^m}\) (see~\Cref{eq:assouad-quantum-example}) be normalized quantum states satisfying $\langle0|\psi_z\rangle\in[0,1]$ for all $z\in\{\pm 1\}^m$. Then there exists a family of unitaries \(\{U_z\}_{z\in\{\pm 1\}^m}\) satisfying $U_z|0\rangle=|\psi_z\rangle$ such that for every $z, z'\in \{\pm 1\}^m$:
    $$
    \|U_z-U_{z'}\| \le 2\sqrt2\,\||\psi_z\rangle-|\psi_{z'}\rangle\|.
    $$
\end{lemma}
\begin{proof}
    For each $z\in\{\pm 1\}^m$, define $|v_z\rangle
        :=
        \frac{|0\rangle+|\psi_z\rangle}
        {\||0\rangle+|\psi_z\rangle\|},$ then the family of $U_z$ is defined as 
    \[
        U_z:=2|v_z\rangle\langle v_z|-I.
    \]
    Since $|v_z\rangle$ is normalized, $U_z$ is Hermitian and satisfies $U_z^2=I$, hence is unitary. It is easy to see that
    \begin{align*}
        U_z|0\rangle=2|v_z\rangle\langle v_z|0\rangle-|0\rangle =|\psi_z\rangle.
    \end{align*}

    It remains to bound the operator-norm distance. Since $\||0\rangle+|\psi_z\rangle\|^2=2+2\langle0|\psi_z\rangle\ge 2$, it follows that
    \begin{align*}
        \|\,|v_z\rangle-|v_{z'}\rangle\,\|\le
        \frac{2\||\psi_z\rangle-|\psi_{z'}\rangle\|}
        {\sqrt{2}}=\sqrt{2}\,
        \||\psi_z\rangle-|\psi_{z'}\rangle\|.
    \end{align*}
    Here, we use $\left\|
        \frac{a}{\|a\|}-\frac{b}{\|b\|}
        \right\|
        \le
        \frac{2\|a-b\|}
        {\min\{\|a\|,\|b\|\}},$ for any nonzero vectors $a,b$. Therefore,
    \begin{align*}
        \|U_z-U_{z'}\|
        &=2\bigl\|
        |v_z\rangle\langle v_z|
        -|v_{z'}\rangle\langle v_{z'}|
        \bigr\|\\
        &\le 2\|\,|v_z\rangle-|v_{z'}\rangle\,\|\\
        &\le 2\sqrt{2}\,
        \||\psi_z\rangle-|\psi_{z'}\rangle\|.
    \end{align*}
\end{proof}
Combining all the above components, we establish the following lower bound: 
\begin{theorem}[Quantum example-oracle lower bound]
\label{thm:gaussian-assouad-lower-bound}
For every $n\geq2$, $0<\eps\leq1/(8\pi)$, any quantum learner that for every $w\in \sphere^{n-1}$, outputs a classical measurable classifier $h$ satisfying
$d_{\gamma_n}(h,f_w)\leq\eps$ with probability at least $2/3$ requires
\begin{equation*}
    \Omega\!\left(\frac{n^{1/4}}{\sqrt\eps}\right)
\end{equation*}
queries to $\QEX(f_w, \gamma_n)$ (and/or its inverse). This holds even for adaptive algorithms that output improper hypotheses (i.e., an $h$ that is not itself a homogeneous LTF).
\end{theorem}

\begin{proof}
We restrict to learning the concept class $\widetilde{\mathcal C}$ from~\Cref{eqn: hard_concept_class}. Any learner that succeeds for every $w\in\sphere^{n-1}$
must also succeed on $\widetilde{\mathcal C}$.

Since the $\QEX$ oracle is specified only through its action on the initial state, $U_z\ket{0}=\ket{\psi_z}$, we may fix a particular family of oracle constructions for our lower-bound analysis. In particular, by~\Cref{lem:canonical-unitary-extensions}, we can choose a single family of unitaries
$\{U_z\}_{z\in\{\pm1\}^m}$ satisfying $U_z\ket{0}=\ket{\psi_z}$ such that, for any
$z,z'\in\{\pm1\}^m$,
\begin{equation}
\label{eq:completion-lipschitz}
    \|U_z-U_{z'}\|
    \le 2\sqrt{2}\,
    \|\ket{\psi_z}-\ket{\psi_{z'}}\|.
\end{equation}
Moreover, by the definition of the quantum example states,
\begin{align}
    \langle\psi_z|\psi_{z'}\rangle
    &=\Pr_{x\sim\gamma_n}[f_z(x)=f_{z'}(x)]
    \nonumber\\
    &=1-d_{\gamma_n}(f_z,f_{z'}),
\end{align}
which implies
\begin{equation}
    \|\ket{\psi_z}-\ket{\psi_{z'}}\|^2
    =2d_{\gamma_n}(f_z,f_{z'}).
\end{equation}
Substituting into~\Cref{eq:completion-lipschitz}, we obtain
\begin{equation}
\label{eq:completion-lipschitz-disagreement}
    \|U_z-U_{z'}\|
    \le 4\sqrt{d_{\gamma_n}(f_z,f_{z'})}.
\end{equation}
The same bound holds for inverse oracles because $\|U_z^\dagger-U_{z'}^\dagger\|=\|U_z-U_{z'}\|$.

For a bit string $z$, let $z^{(i)}$ denote the neighboring sign vector obtained by flipping its $i$-th coordinate. By~\Cref{lem:assouad-edge} and~\Cref{eq:completion-lipschitz-disagreement},
\begin{equation}
\label{eq:assouad-oracle-edge}
\begin{aligned}
    \Delta
    &:=\max_{z\in\{\pm1\}^m,\;i\in[m]}
    \|U_z-U_{z^{(i)}}\|\\
    &\le 4
    \sqrt{d_{\gamma_n}(f_z,f_{z^{(i)}})}\\
    &\leq 8\sqrt{\frac{\pi\eps}{\sqrt{m}}}.
\end{aligned}
\end{equation}
On the other hand, let ${\cal A}(U)$ be an arbitrary quantum learner making at most $T$ queries to the example-preparation oracle and its inverse. Denote by $\rho_z^{(T)}$ its final output state when the oracle is $U_z$ and let $P_z$ be the distribution of the decoded output from ${\cal A}(U_z)$. By the
data-processing inequality for trace distance and the standard quantum hybrid bound, we can show 
\begin{equation*}
    \dtv(P_z,P_{z^{(i)}})\leq d_{\mathrm{tr}}\left(
        \rho_z^{(T)},
        \rho_{z^{(i)}}^{(T)}
    \right) \leq T\lVert U_z-U_{z^{(i)}}\rVert = T\Delta.
\end{equation*}

Now,~\Cref{lem:assouad-testing} and~\Cref{lem:assouad-decoding} imply
\[
    \frac m2(1-T\Delta) \leq \E[\dHam(z,\widehat z)] \leq \frac{3m}{8}.
\]
Using~\Cref{eq:assouad-oracle-edge} and $m=n-1$, we conclude
\[
    T\geq  \frac{m^{1/4}}{32\sqrt{\pi\eps}} =  \Omega\p{\frac{m^{1/4}}{\sqrt{\eps}}} = \Omega\p{\frac{n^{1/4}}{\sqrt{\eps}}} 
\]
as claimed.
\end{proof}

\begin{remark}[Under discretized quantum example oracle]\label{rem: lower_bound_discrete}
The quantum query lower bound of~\Cref{thm:gaussian-assouad-lower-bound} also holds for the finite-dimensional oracle $\QEX_M(f_w,\gamma_n)$ with
the grid size $M$ chosen in~\Cref{thm:quantum-example-learning}. Indeed, for any neighboring pair $z,z^{(i)}$, the corresponding discretized example states satisfy
$\bigl\|\ket{\psi_{z,M}}-\ket{\psi_{z^{(i)},M}}\bigr\|_2^2=2d_{\gamma_{n,M}}(f_z,f_{z^{(i)}})$. By the first inequality in~\Cref{eq:discrete-disagreement-bound}, the chosen value of $M$ makes the additive discretization error $o(\eps/\sqrt m)$ and ensures that $d_{\gamma_{n,M}}(f_z,f_{z^{(i)}}) \leq O\!\left(\frac{\eps}{\sqrt m}\right)$. Consequently, the same construction of full-unitary
oracle completions gives
\[
\|U_{z,M}-U_{z^{(i)},M}\|
=O\!\left(\frac{\sqrt{\eps}}{m^{1/4}}\right).
\]
The hybrid and hypercube-testing argument of~\Cref{thm:gaussian-assouad-lower-bound} therefore applies unchanged, yielding the same lower bound
$
T=\Omega\!\left(\frac{n^{1/4}}{\sqrt{\eps}}\right).
$
\end{remark}

\section{Conclusion and open problems}

We have studied the quantum query complexity of learning linear threshold functions under three natural access models, obtaining exponential quantum speedups for learning general LTFs from real-domain membership queries and for the special case of homogeneous LTFs corresponding to Majority-juntas from Boolean membership queries, together with a quartic speedup in the dimension for learning homogeneous LTFs from quantum examples under the Gaussian distribution.

Several concrete questions remain open.

\begin{enumerate}

\item Our \(O(\log k)\)-query algorithm for Majority-juntas improves substantially over the \(\widetilde O(k^{1/4})\) Fourier sampling and amplitude amplification approach of~\cite{Montanaro2022QuantumAF}, but it is not clear how to apply it in other learning contexts. It would be interesting to see whether the circuit in~\Cref{fig:majority-unitary-circuit} can be used as a new quantum primitive for learning problems. 

\item Can the intersection-free framework give improved algorithms for
learning other symmetric juntas? More broadly, can similar restrictions be
useful for learning problems invariant under other symmetry groups?

% \item What is the optimal quantum-example complexity of learning homogeneous LTFs? In particular, can one prove a lower bound matching, or approaching, the
% upper bound obtained from the quantum Hermite transform? We showed that the $1/\sqrt{\eps}$-dependence of our upper bound is needed, but proved nothing about the $n$-dependence. 

\item Can the Hermite-transform approach for the example oracle be extended from homogeneous halfspaces to general LTFs with a nonzero threshold~$\theta$?

\item Can we generalize some of our results from Linear Threshold Functions to \emph{Polynomial} Threshold Functions, where $f(x)$ is the sign of a low-degree real $n$-variate polynomial in~$x$ rather than the sign of an affine linear function in~$x$?
\end{enumerate}

\paragraph{Acknowledgments.}
The authors thank Francisco Escudero Gutiérrez, Sander Gribling, Arjan Cornelissen, Jop Briët, and Aleksandrs Belovs for very helpful discussions and guidance. The authors also thank Sid Jain for discussions on using the quantum Hermite transform to learn LTFs.

\paragraph{AI-use disclosure.}
GPT-5.6-Sol and GPT-6-Astra were used to assist with proof development and literature research. In Section~3, the connection between learning LTFs and earlier work on quantum convex optimization, and the potential of using Jordan's gradient-estimation algorithm, were identified by the authors; based on this connection and guidance from the authors, the AI was used to help develop and refine the proof structure for the upper bound. It also strengthened our query-complexity lower bound from $\log(1/\eps)$ to the optimal $\log (n/\eps)$. In Section~4, the rank-one intersection-free SDP was proposed as a numerical optimization by the authors. Once the surprising $\log k$ pattern emerged, AI was used to find the analytical solutions based on the numerical results. It then found the characterization with Krawtchouk polynomials (which were already used in \cite{belovs2015symmetricjuntas} for the primal adversary SDP), generalized this construction to the $1\leq |A|\leq k$ promise, and found the matching lower bound. The authors first tried to use the QFT construction in~\cite{Ambainis2016Efficient} for the circuit implementation of $U$, however the AI found a substantially simpler construction with the Hamming weight counter. In Section~5, the authors developed the overall algorithmic framework, while AI was used to reduce the non-oracle gate complexity. An initial implementation of the required isometry map used a Schur-transform circuit~\cite{Krovi2019}, resulting in an $O(n^{10})$ gate complexity. The AI suggested replacing this construction with a substantially simpler second-to-first quantization circuit for identical bosons~\cite{liu2025lowdepthquantumsymmetrization}, reducing the gate complexity to $O(n^{5/4})$. Note that this modification affects only the non-oracle gate complexity and does not contribute to the claimed improvement in membership-query complexity. For the matching lower bound of $\Omega(n^{1/4}/\sqrt{\eps})$, the authors initially established the $\Omega(1/\sqrt{\eps})$ dependence using a simple two-point argument. To investigate the dimension dependence, the AI suggested exploring Assouad's lemma. Building on this suggestion, the authors identified a novel connection between Assouad's hypercube method and Gaussian surface area, developed the necessary mathematical arguments, and obtained an initial dimension-dependent lower bound using a general Gaussian surface area bound for convex sets. The AI tool was subsequently used to refine this argument by exploiting the sharper, dimension-independent bound on the Gaussian surface area of halfspaces directly, rather than the general bound for convex sets. This refinement led to the matching lower bound of $\Omega(n^{1/4}/\sqrt{\eps})$. Finally, GPT-6-Astra and Claude Fable 5.1 were used for help with proofreading. The authors reviewed and verified the resulting arguments, citations, and text, and take full responsibility for the contents of the paper. 

\bibliographystyle{alpha}
\bibliography{ref}

\clearpage

\appendix
\crefalias{section}{appendix}
\crefname{appendix}{Appendix}{Appendices}
\Crefname{appendix}{Appendix}{Appendices}
\section{Supplementary proofs of classical bounds}

\subsection{Classical upper bound for learning Majority}

\begin{theorem}[Classical upper bound]
\label{prop:majority-classical-upper-bound}
The Majority-junta problem with the promise $1\leq |A|\leq k$ can be learned classically using $O(k\log(en/k))$ membership queries for all $n\geq k$. If $\liminf_{k\to\infty}n/k>1$, this gives the matching $O(k\log(n/k))$ classical upper bound.
\end{theorem}
\begin{proof}
Let $m=\lceil |A|/2\rceil$ and let $p$ be the $m$-th smallest element of $A$. We define the prefixes $P_j=[j]$, for which
$f_A(P_j)=1$ if and only if $|A\cap P_j|\geq m$. Thus, $p$ can be found by
binary search using $O(\log n)$ queries. Put $L=[p-1]$ and
$R=\{p+1,\ldots,n\}$, so that $|A\cap L|=m-1$. For every $Q\subseteq R$,
\[
f_A(L\cup Q)=1
\quad\Longleftrightarrow\quad
A\cap Q\neq\emptyset,
\]
whereas for every $Q\subseteq L$,
\[
f_A((L\setminus Q)\cup\{p\})=-1
\quad\Longleftrightarrow\quad
A\cap Q\neq\emptyset.
\]
Hence, after finding $p$, membership queries to $f_A$ simulate group-testing queries on $L$ and $R$. Du and
Park~\cite{DuPark94} gave an adaptive group-testing algorithm that does not
require knowing in advance the number $d$ of elements of $A$ among $N$ items and uses at most
$d\log_2(N/d)+4d$ queries. If $d=0$, there are no marked elements on that side and no group-testing
queries are required; hence, we apply the bound below only when $d\geq1$. In our case, $d=|A\cap L|$ when applying the algorithm to $L$ and $d=|A\cap R|$ when applying it to $R$, and these satisfy
$|A\cap L|+|A\cap R|=|A|-1$. Applying their algorithm to the two
sides identifies $A$ using
\[
O\!\left(\log n+|A|\log\frac{en}{|A|}\right)
=O\!\left(k\log\frac{en}{k}\right)
\]
membership queries, where the last inequality uses $|A|\leq k$. If $\liminf_{k\to\infty}n/k>1$, then $\log(n/k)$ is bounded below by a positive constant for all sufficiently large $k$. Consequently,
\[
    k\log\frac{en}{k}
    =k\left(1+\log\frac nk\right)
    =O\!\left(k\log\frac nk\right),
\]
which proves the claimed matching upper bound in this regime.
\end{proof}

\subsection{Classical lower bounds for membership queries}
We first establish an information-theoretic lower bound on the classical membership-query complexity of learning homogeneous halfspaces under the
standard Gaussian distribution. The bound applies to arbitrary adaptive randomized learners and allows improper output hypotheses  (i.e., hypotheses that are not LTFs). The argument is
based on the metric packing characterization of membership-query learning from~\cite{kulkarni1993active}.
Let $\calC_n:=\{f_w:w\in\sphere^{n-1}\}$ denote this concept class.

\begin{definition}[Packing Number]
Let $\calC$ be a concept class over an instance space $\X$, and let $\D$ be a distribution over $\X$. For $\eps>0$, an $\eps$-packing of $\calC$ with respect to $\D$ is a subset $\mathcal K\subseteq\calC$ such that $d_{\D}(f,g)>\eps$ for every pair of distinct $f,g\in\mathcal K$.  The $\eps$-packing
number of $\calC$ under $\D$ is
\[
    \mathcal P(\calC,\D,\eps)
    :=
    \sup\left\{
        |\mathcal K|:
        \mathcal K\subseteq\calC
        \text{ is an $\eps$-packing}
    \right\}.
\]
\end{definition}
The query complexity depends on the packing number as follows.
\begin{lemma}\label{lem:classical-lb}
Let $\calC$ be a concept class and let $\D$ be a distribution over the instance space. Any randomized membership-query learner that, with probability at least $1-\delta$, outputs a hypothesis with error at most $\eps$ requires at least
\[
\left\lceil\log_2\!\left((1-\delta)\,\mathcal P(\calC,\D,2\eps)\right)\right\rceil
\]
queries.
\end{lemma}
\begin{proof}
Let $\mathcal K$ be any finite $2\eps$-packing of size $M$, and choose the target uniformly from $\mathcal K$. After fixing its internal randomness, a $q$-query learner is a deterministic binary decision tree with at most $2^q$ transcripts. The hypothesis associated with any transcript can be $\eps$-close to at most one member of $\mathcal K$, since otherwise the triangle inequality would place two distinct members at distance at most $2\eps$. Thus the deterministic learner succeeds on at most $2^q$ targets, and its average success probability is at most $2^q/M$. Averaging over the learner's randomness preserves this bound. Since the learner succeeds with probability at least $1-\delta$ on every target, $1-\delta\leq2^q/M$, and hence $q\geq\log_2((1-\delta)M)$. Taking the supremum over all finite $2\eps$-packings proves the claim.
\end{proof}
We now specialize this general bound to homogeneous halfspaces under the standard Gaussian distribution.
\begin{theorem}[Classical randomized MQ lower bound]
\label{thm:classical-main}
Let $n\geq2$ and let $0<\eps\leq\eps_0$, where $\eps_0>0$ is a sufficiently small universal constant. Any randomized classical membership-query learner that, for every \(w\in\sphere^{n-1}\), outputs with probability at least \(2/3\) a classifier \(h\) satisfying $d_{\gamma_n}(h,f_w)\leq\eps$ requires
\begin{equation}
  q=\Omega\!\left(n\log\frac1\eps\right)
  \label{eq:classical-main}
\end{equation}
queries.
\end{theorem}
\begin{proof}
From~\Cref{lem:classical-lb}, with $\delta=1/3$, any such learner must make at least
\[
\left\lceil
\log_2\!\left(\frac{2}{3}\,\mathcal{P}(\calC_n,\gamma_n,2\eps)\right)
\right\rceil
\]
queries. Thus, it remains to lower bound $\mathcal{P}(\calC_n,\gamma_n,2\eps)$. By~\Cref{lem:gaussian-disagreement}, for homogeneous halfspaces $d_{\gamma_n}(f_u,f_w)=\theta(u,w)/\pi$, where $\theta(u,w)=\arccos(\langle u,w\rangle)$ is the angle between $u$ and $w$. Consequently, a $2\eps$-packing of $\calC_n$ under $d_{\gamma_n}$ is equivalent to a packing of $\sphere^{n-1}$ whose distinct
points have angular separation greater than $2\pi\eps$. Standard spherical packing bounds imply that there exists a universal constant $c>0$ such that $\mathcal{P}(\calC_n,\gamma_n,2\eps)\geq\left(\frac{c}{\eps}\right)^{n-1}$~\cite{Vershynin2018}. Taking logarithms proves~\Cref{eq:classical-main}.
\end{proof}

\subsection{Classical lower bounds for example queries}

\begin{theorem}[Classical example-query lower bound~{\cite[Theorem~13]{BalcanLong2013}}]
\label{thm:classical-example-main}
Let $n\geq 3$, $0<\eps\leq 1/8$, and $0<\delta\leq 1/4$.
Suppose that, for every unknown target $f_w\in\calC_n$, a randomized
classical learner given access to $\EX(f_w,\gamma_n)$ outputs a
hypothesis $h$ satisfying $d_{\gamma_n}(h,f_w)\leq\eps$ with probability
at least $1-\delta$. Then the learner requires
\begin{equation*}
    q
    =
    \Omega\!\left(
        \frac{n}{\eps}
        +
        \frac{1}{\eps}\log\frac1\delta
    \right)
\end{equation*}
example queries.
\end{theorem}

\begin{proof}
The standard Gaussian distribution $\gamma_n$ is log-concave and has
full-rank covariance. The result therefore follows directly from
\cite[Theorem~13]{BalcanLong2013}, which gives a lower bound of
\[
    \Omega\!\left(
        \frac{n}{\eps}
        +
        \frac{1}{\eps}\log\frac1\delta
    \right)
\]
for passive learning of homogeneous linear separators under any
full-rank log-concave distribution.
\end{proof}
\begin{remark}[Continuous versus discretized examples]\label{rem: continuous-vs-discretized}
The classical lower bound in~\Cref{thm:classical-example-main}
is stated for examples drawn from the continuous Gaussian distribution,
whereas~\Cref{thm:quantum-example-learning} uses the finite discretization
underlying $\QEX_M$. This does not affect the comparison. Indeed, the
choice of $M$ in~\Cref{thm:quantum-example-learning} implies that rounding a continuous Gaussian example to the grid while retaining its label differs from a labeled example obtained from the discretized distribution by at most $o(\eps/n)$ in total variation distance. Consequently, over $O(n/\eps)$ examples, the joint distributions differ by only $o(1)$. Hence, an asymptotically better classical learner for the discretized example model would also give one for the continuous Gaussian model, contradicting~\Cref{thm:classical-example-main}. Thus, the $\Omega(n/\eps)$ classical lower bound, and hence our quartic
improvement in terms of the dimension, is unaffected by the discretization.
\end{remark}
\section{Supplementary proofs for Quantum Membership Query Lower Bound}\label{app: membership_lower_bound}
Here we present the missing proof for~\Cref{lem:ltf-query-multiplier} in the quantum membership queries lower bound. Before going into the detailed proof, we first present the following weighted Schur test lemma~\cite{DymKatsnelson2007} showing that an infinite matrix is a bounded operator, without having to calculate its eigenvalues or operator norm directly.
\begin{lemma}[Weighted Schur test~\cite{DymKatsnelson2007}]
\label{lem:weighted-schur}
Let $(K_m)_{m\geq 0}$ be a family of possibly infinite
matrices $K_m=((K_m)_{jk})_{j,k\geq 0}$ with complex
entries. Suppose there exist strictly positive weights
$(u_j)_{j\geq 0}$ and universal constants $C_1, C_2>0$ such
that, for every $m\geq 0$,
\[
    \sum_{k\geq 0}|(K_m)_{jk}|u_k\leq C_1 u_j,
    \qquad
    \sum_{j\geq 0}|(K_m)_{jk}|u_j\leq C_2 u_k,
\]
for all $j,k\geq 0$. Then each $K_m$ defines a bounded linear operator on
$\ell^2(\mathbb N_0)$, satisfying
\[
    \sup_{m\geq 0}\|K_m\|\leq \sqrt{C_1C_2} < \infty.
\]
\end{lemma}

\begingroup
\renewcommand{\thelemma}{\ref{lem:ltf-query-multiplier}}
\begin{lemma}[Constant multiplicative growth per query]
There is a universal constant $C>1$ such that, for every $n\geq1$,
every $x\in\mathbb R^n\setminus\{0\}$, and every Hilbert-space-valued
$F$ with $\mathcal W(F)<\infty$,
\begin{equation}
    \mathcal W\!\left(
        \sgn\langle G,x\rangle F(G)
    \right)
    \leq C\,\mathcal W(F).
\end{equation}
\end{lemma}
\endgroup

\begin{proof}
We first analyze multiplication by $\sgn(z)$ in one
dimension, then lift the estimate to arbitrary dimensions.

Let $h_j$ be the normalized probabilists' Hermite polynomials,
orthonormal in $L^2(\gamma_1)$, and define
\[
    M_{jk}:=\mathbb E_{Z\sim\gamma_1}
        [\sgn(Z)h_j(Z)h_k(Z)].
\]
The operator $M$ is unitary, since multiplication by
$\sgn(z)$ preserves the $L^2(\gamma_1)$ norm.
Parity implies $M_{jk}=0$ whenever $j$ and $k$ have the same parity.

For completeness, let $\varphi(z)=(2\pi)^{-1/2}e^{-z^2/2}$.
The identities $h_j'=\sqrt j\,h_{j-1}$ and
$(\varphi h_j')'=-j\varphi h_j$ give
\[
    (k-j)\int_0^\infty h_j(z)h_k(z)\varphi(z)\,dz
    =\varphi(0)\bigl(h_j(0)h_k'(0)-h_j'(0)h_k(0)\bigr).
\]
Taking $j=2r$ and $k=2s+1$ therefore yields
\begin{equation*}
    M_{2r,2s+1}
    =\sqrt{\frac2\pi}\,
      \frac{h_{2r}(0)\sqrt{2s+1}\,h_{2s}(0)}{2s+1-2r}.
\end{equation*}
Moreover,
\[
    h_{2r}(0)=(-1)^r\frac{\sqrt{(2r)!}}{2^r r!},
    \qquad |h_{2r}(0)|\lesssim(r+1)^{-1/4}.
\]
Here and throughout the proof, $A\lesssim B$ means that $A\leq CB$ for some universal constant $C>0$,
independent of all indices and parameters. The central-binomial-coefficient bound, together with the symmetry of
$M$, consequently implies that, for $j\neq k$,
\begin{equation}
\label{eq:ltf-sign-entry-bound}
    |M_{jk}|
    \lesssim
    \frac{\left(\dfrac{\max\{j+1,k+1\}}
                           {\min\{j+1,k+1\}}\right)^{1/4}}
         {|j-k|}.
\end{equation}
All implicit constants below are universal.

For each integer $m\geq0$, define
\[
    W_m:=\operatorname{diag}\bigl((j+m+1)^{1/8}\bigr)_{j\geq0}.
\]
We claim that there is a universal $B<\infty$ such that
\begin{equation}
\label{eq:ltf-weighted-sign-bound}
    \sup_{m\geq0}\|W_mMW_m^{-1}\|\leq B.
\end{equation}
Write $W_mMW_m^{-1}=M+K_m$. With $a=j+1$ and $b=k+1$,
\[
    (K_m)_{jk}
    =\left[\left(\frac{a+m}{b+m}\right)^{1/8}-1\right]M_{jk}.
\]
Adding $m$ moves the ratio towards $1$, so the absolute value of the
bracket is bounded by its value at $m=0$. Using~\Cref{eq:ltf-sign-entry-bound}, we obtain
\begin{equation*}
    |(K_m)_{jk}|\lesssim
    \begin{cases}
        a^{-5/8}b^{-3/8}, & a\geq2b,\\
        a^{-1/4}b^{-3/4}, & b\geq2a,\\
        a^{-1}, & a/2<b<2a.
    \end{cases}
\end{equation*}
Indeed, in the first region the bracket is at most $(a/b)^{1/8}$;
in the second it is at most $1$. In the comparable-index region,
the mean-value theorem gives
\[
    \left|\left(\frac{a+m}{b+m}\right)^{1/8}-1\right|
    \lesssim\frac{|a-b|}{a},
\]
which cancels the factor $1/|j-k|$ in~\Cref{eq:ltf-sign-entry-bound}. The diagonal entries of $K_m$ vanish.

We apply~\Cref{lem:weighted-schur} with $u_j=(j+1)^{-1/2}$. To verify its hypotheses, we bound the weighted row and column sums in each region as follows. For $a\geq 2b$, we have
\[
    a^{-5/8}\sum_{b\leq a/2}b^{-7/8}\lesssim a^{-1/2},
    \qquad
    b^{-3/8}\sum_{a\geq2b}a^{-9/8}\lesssim b^{-1/2}.
\]
For $b\geq2a$, the corresponding estimates are
\[
    a^{-1/4}\sum_{b\geq2a}b^{-5/4}\lesssim a^{-1/2},
    \qquad
    b^{-3/4}\sum_{a\leq b/2}a^{-3/4}\lesssim b^{-1/2}.
\]
In the comparable-index region there are $O(a)$ indices in each row,
each contributing $O(a^{-3/2})$, and the column estimate is analogous.
Thus, uniformly in $m$,
\[
    \sum_k|(K_m)_{jk}|u_k\lesssim u_j,
    \qquad
    \sum_j|(K_m)_{jk}|u_j\lesssim u_k.
\]
\Cref{lem:weighted-schur} guarantees $\sup_m\|K_m\|<\infty$.
Since $\|M\|=1$, this proves~\Cref{eq:ltf-weighted-sign-bound}.
The conjugated operators, initially interpreted on finitely supported
sequences, therefore extend to bounded operators on $\ell^2$.

In $n$ dimensions, write $\alpha=(j,\beta)$ and $m=|\beta|$.
Multiplication by $\sgn(G_1)$ preserves $\beta$.
On this block, $(I+N)^{1/8}$ acts as $W_m$, so~\Cref{eq:ltf-weighted-sign-bound} gives
\[
    \|(I+N)^{1/8}[\sgn(G_1)F]\|
    \leq B\|(I+N)^{1/8}F\|.
\]
The same bound holds for Hilbert-space-valued $F$ by tensoring with
the identity on the value space.

Finally, orthogonal coordinate changes preserve Gaussian measure and
each total Hermite-degree subspace. Rotating $x/\|x\|$ to the first
coordinate therefore gives, for every $x\neq0$,
\[
    \|(I+N)^{1/8}[\sgn\langle G,x\rangle F]\|
    \leq B\|(I+N)^{1/8}F\|.
\]
These estimates hold first for finite Hermite expansions. Density in
the weighted norm, together with the ordinary $L^2$ continuity of the
sign multiplier, extends them to every $F$ with
$\mathcal W(F)<\infty$.
Squaring proves the lemma with $C=\max\{2,B^2\}$.
\end{proof}
\section{Supplementary proofs for the case of Boolean queries}

\subsection{Krawtchouk polynomial properties}
    \begin{lemma}
    \label{lem:krawtchouk-properties}
    The following properties hold.

    \begin{enumerate}[label=\roman*.]

        \item
        For every $r\geq0$ and $0\leq t,a\leq r$,
        \begin{equation*}
        K_t^{(r)}(r-a)
        =
        (-1)^t K_t^{(r)}(a)
        \end{equation*}

        \item
        For every $r\geq0$ and $0\leq t,a\leq r$,
        \begin{equation*}
        \binom{r}{a}K_t^{(r)}(a)
        =
        \binom{r}{t}K_a^{(r)}(t)
        \end{equation*}

        \item
        For every $r\geq0$ and $0\leq u,v\leq r$,
        \begin{equation*}
        \sum_{a=0}^{r}
        \binom{r}{a}
        K_u^{(r)}(a)K_v^{(r)}(a)
        =
        2^r\binom{r}{u}\mathbf{1}_{\{u=v\}}
        \end{equation*}

        \item
        If $P$ is a polynomial of degree strictly less than $u$, then
        \begin{equation*}
        \sum_{a=0}^{r}
        \binom{r}{a}
        K_u^{(r)}(a)P(a)
        =0
        \end{equation*}

        \item
        Let $A,B\subseteq[n]$ satisfy $|A|=r$ and $|B|=s$, and let
        $S\subseteq[n]$ be uniformly random.  Then for every
        $0\leq u\leq r$ and $0\leq v\leq s$,
        \begin{equation*}
        \mathbb{E}_S
        \left[
            K_u^{(r)}(|A\cap S|)
            K_v^{(s)}(|B\cap S|)
        \right]
        =
        \mathbf{1}_{\{u=v\}}
        \binom{|A\cap B|}{u}
        \end{equation*}

    \end{enumerate}
    \end{lemma}
\begin{proof}[Proof of \Cref{lem:krawtchouk-properties}]
Statements \textup{(i)}--\textup{(iii)} are the standard reflection,
reciprocity, and orthogonality identities for binary Krawtchouk
polynomials. In the notation of~\cite{nomura2012krawtchouk}, our
polynomial is $K_t^{(r)}(a)=\binom rt K_t(a;1/2,r)$; orthogonality and
the generating function are given there in Theorems~3.22 and~3.25.
For \textup{(iv)}, the polynomials
$K_0^{(r)},\ldots,K_{u-1}^{(r)}$ form a basis of the polynomials of
 degree at most $u-1$.  Thus, if $\deg P<u$, write
$P=\sum_{j=0}^{u-1}c_jK_j^{(r)}$.  By \textup{(iii)}
\[
\sum_{a=0}^{r}\binom{r}{a}K_u^{(r)}(a)P(a)
=
\sum_{j=0}^{u-1}c_j
\sum_{a=0}^{r}\binom{r}{a}
K_u^{(r)}(a)K_j^{(r)}(a)
=0
\]

For \textup{(v)}, define
$\chi_T(S)=(-1)^{|T\cap S|}$.  If $|A|=r$, then grouping
$T\subseteq A$, $|T|=u$, according to
$j=|T\cap S|$ gives
\begin{equation}
\label{eq:kraw-character-proof}
K_u^{(r)}(|A\cap S|)
=
\sum_{\substack{T\subseteq A\\|T|=u}}\chi_T(S)
\end{equation}
Similarly
\[
K_v^{(s)}(|B\cap S|)
=
\sum_{\substack{U\subseteq B\\|U|=v}}\chi_U(S)
\]
The Boolean-cube characters are orthonormal under uniform $S$, hence
\begin{align*}
&\mathbb{E}_S
\left[
K_u^{(r)}(|A\cap S|)
K_v^{(s)}(|B\cap S|)
\right]\\
&\qquad=
\sum_{\substack{T\subseteq A,\ |T|=u}}
\sum_{\substack{U\subseteq B,\ |U|=v}}
\mathbb{E}_S[\chi_T(S)\chi_U(S)]\\
&\qquad=
\mathbf{1}_{\{u=v\}}
\#\{T\subseteq A\cap B:|T|=u\}\\
&\qquad=
\mathbf{1}_{\{u=v\}}
\binom{|A\cap B|}{u}
\end{align*}
\end{proof}
\paragraph{Krawtchouk expansions of threshold-supported vectors.}
\label{par:krawtchouk-threshold-expansions}
By \Cref{lem:krawtchouk-properties}\textup{(iii)}, the normalized
Krawtchouk polynomials
$K_t/\sqrt{\binom{k}{t}}$, for $0\leq t\leq k$, form an orthonormal basis
of $\mathbb{R}^{k+1}$ with respect to the binomial distribution
$\operatorname{Bin}(k,1/2)$. Hence any vectors $L$ and $U$ can be written
without loss of generality as
\[
L(a)=\sum_{t=0}^{k}K_t(a)p_t
\qquad
U(a)=\sum_{t=0}^{k}K_t(k-a)q_t\,.
\]
If $L$ is supported on $\{0,\ldots,m\}$,
\Cref{lem:krawtchouk-properties}\textup{(ii)} gives
\[
p_t=2^{-k}\sum_{a=0}^{m}K_a(t)L(a)
\]
so $t\mapsto p_t$ is the restriction of a polynomial of degree at most
$m$. The same argument applied to $a\mapsto U(k-a)$ shows that
$t\mapsto q_t$ is also the restriction of a polynomial of degree at most
$m$.

\subsection{Upper bound under the promise \texorpdfstring{$1\leq |A|\leq k$}{1 <= |A| <= k}}
\label{app:majority-unknown-upper-bound}
\begin{lemma}
Let $r=2\ell+1$, and let $R$ be either
$P_\ell$ or $Q_\ell$ from
\Cref{def:mod4}. There is a constant
$C$ such that, for every $0\leq t\leq r$,
\begin{equation}
\label{eq:unknown-normalized-root-pairing}
\binom rt R(t)^2
\leq C\frac{r}{(t+1)(r-t+1)}.
\end{equation}
Moreover, if $s$ is a root of $R$, then
\begin{equation}
\label{eq:unknown-normalized-root-derivative}
\binom rs R'(s)^2
\leq C\frac{r}{(s+1)(r-s+1)}.
\end{equation}
\end{lemma}

\begin{proof}
By \Cref{lem:majority-asymptotic}, the first bound holds with
$R(t)^2$ replaced by $P_\ell(t)^2+Q_\ell(t)^2$. Since both terms are
nonnegative, the same bound follows immediately for $R(t)^2$ alone. For the second estimate, write $R(t)=c\prod_{u\in Z}(t-u)$, where $Z$ denotes the set of roots of the associated polynomial. If
$s\in Z$, then
$R'(s)=c\prod_{u\in Z\setminus\{s\}}(s-u)$.
Thus, it is the same proof as
\Cref{lem:majority-asymptotic}, but the factor $t-s$ is omitted. More
precisely, all root pairings not containing $s$ are unchanged. For the
pair containing $s$, if it comes from a four-point block, the remaining
factor has square $9$, while the product of the other three distances is
$6$, so the same comparison holds up to a factor $2$. If $s$ is a root
of one of the endpoint factors, the corresponding remaining factor is
$1$. Hence the same pairing and normalization argument gives
\[
R'(s)^2
\leq C\frac{(s-1)!(r-s-1)!}{(r-1)!}.
\]
Since $1\leq s\leq r-1$, multiplying by $\binom rs$ gives
\[
\binom rs R'(s)^2
\leq C\frac{r}{s(r-s)}
\leq 4C\frac{r}{(s+1)(r-s+1)},
\]
which proves the second bound.
\end{proof}

\begin{lemma}
Let $\ell\geq1$, $m<\ell$, and $r=2\ell+1$. For every
$0\leq t\leq r$,
\begin{equation}
\label{eq:unknown-uniform-correction}
\binom rt
\left(\frac{C_{m\ell}(t)}{2m+1}\right)^2
\leq
C\frac{r}{(t+1)(r-t+1)}
\frac{t^2}
{(2m+1)^2(1+|t-2m-2|)^2}.
\end{equation}
\end{lemma}

\begin{proof}
Let $R=Q_\ell$ when $m$ is even and $R=P_\ell$ when $m$ is odd.
By the mod-$4$ root pattern, $2m+2$ is a root of $R$. Thus
\[
D(t):=\frac{R(t)}{t-(2m+2)}
\]
is a polynomial, including at $t=2m+2$, and
\Cref{eq:unknown-Cml} becomes
\begin{equation}
\label{eq:unknown-Cml-quotient}
C_{m\ell}(t)=tD(t).
\end{equation}
For $t\ne2m+2$, \Cref{eq:unknown-normalized-root-pairing} gives
\[
\binom rtD(t)^2
\leq C\frac{r}{(t+1)(r-t+1)}
\frac{4}{(1+|t-2m-2|)^2},
\]
where we used $1/d^2\leq4/(1+|d|)^2$ for every nonzero integer $d$.
At $t=2m+2$, we have $D(2m+2)=R'(2m+2)$, so
\Cref{eq:unknown-normalized-root-derivative} gives the same bound.
Multiplying by $t^2/(2m+1)^2$, absorbing the factor $4$ into $C$,
and using \Cref{eq:unknown-Cml-quotient} proves the claim.
\end{proof}

\begin{lemma}
\label{lem:unknown-uniform-kernel}
For every integer $t\geq0$,
\begin{equation}
\label{eq:unknown-uniform-kernel}
\sum_{m\geq0}
\frac{t^2}
{(2m+1)^2(1+|t-2m-2|)^2}
\leq C\,.
\end{equation}
\end{lemma}

\begin{proof}
The case $t\leq2$ is immediate, since $t^2\leq4$ and
$\sum_{m\geq0}(2m+1)^{-2}<\infty$. For $t\geq3$, take $N=t-1$
and extend the sum to all positive integers:
\[
\sum_{m\geq0}
\frac1{(2m+1)^2(1+|t-2m-2|)^2}
\leq
\sum_{a\geq1}\frac1{a^2(1+|N-a|)^2}\,.
\]
For $a\leq N/2$, we have $1+|N-a|\geq N/2$, while for $a>N/2$,
we have $a\geq N/2$. Hence
\[
\sum_{a\geq1}\frac1{a^2(1+|N-a|)^2}
\leq
\frac{C}{N^2}
\left(
\sum_{a\geq1}\frac1{a^2}
+
\sum_{j\in\mathbb Z}\frac1{(1+|j|)^2}
\right)
\leq
\frac{C}{t^2}\,,
\]
where in the last step we used $N=t-1\geq t/2$. Multiplying by $t^2$
proves \Cref{eq:unknown-uniform-kernel}.
\end{proof}

\begin{proof}[Proof of \Cref{eq:unknown-correction-energy}]
The claim is trivial for $\ell=0$, so assume $\ell\geq1$ and put
$r=2\ell+1$. Summing \Cref{eq:unknown-uniform-correction} over
$m<\ell$ and applying \Cref{lem:unknown-uniform-kernel} gives
\[
\sum_{m<\ell}\binom rt
\left(\frac{C_{m\ell}(t)}{2m+1}\right)^2
\leq
C\frac{r}{(t+1)(r-t+1)}\,.
\]
Therefore
\begin{align*}
\sum_{m<\ell}\sum_{t=0}^{r}\binom rt
\left(\frac{C_{m\ell}(t)}{2m+1}\right)^2
&\leq
C\sum_{t=0}^{r}\frac{r}{(t+1)(r-t+1)}\\
&=
\frac{2Cr}{r+2}\sum_{j=1}^{r+1}\frac1j\\
&=
O(\log(r+1))
=
O(\log k)\,.
\end{align*}
This proves \Cref{eq:unknown-correction-energy}.
\end{proof}

\subsection{Lower bound under the promise \texorpdfstring{$1\leq|A|\leq k$}{1 <= |A| <= k}}
\label{app:balanced-slice-lower}
\begin{proof}[Proof of \Cref{lem:balanced-mask-norm}]
Fix $S\subseteq[n]$ and put $r=\lceil m/2\rceil$, where $|A|=m$ on
$\mathcal A$ and $k=2m$. Define
$\mathcal L_a:=\{A\in\mathcal A:|A\cap S|=a\}\,$.
We first show that, for $A\in\mathcal L_a$ and $b\ne a$,
\begin{equation}
\label{eq:balanced-row-layer}
\sum_{B\in\mathcal L_b}\Gamma\llbracket A,B\rrbracket\leq\frac1{|b-a|}\,.
\end{equation}
Choose independent uniform orderings $a_1,\ldots,a_m$ of $A$ and
$b_1,\ldots,b_m$ of $U\setminus A$, and set
\[
A_s=(A\setminus\{a_1,\ldots,a_s\})\cup\{b_1,\ldots,b_s\},
\qquad
X_s=|A_s\cap S|\,.
\]
Since $A_s$ is uniform among the $\binom ms^2$ sets at distance $s$
from $A$,
\begin{equation}
\label{eq:balanced-row-probability}
\sum_{B\in\mathcal L_b}\Gamma\llbracket A,B\rrbracket
=\sum_{s=1}^{m}\frac1s\Pr[X_s=b]\,.
\end{equation}

We use the following standard form of the cycle lemma
\cite{DvoretzkyMotzkin1947,Raney1960}: if $z_1,\ldots,z_s$ are integers
with $z_i\leq1$ and positive total sum $q=\sum_{i=1}^s z_i$,
then exactly $q$ of the $s$ cyclic shifts, counted by their starting
positions, have all partial sums strictly positive. Let $\tau_b:=\min\{t\geq1:X_t=b\}$
and put $q=b-a$. Suppose first that $b>a$. Conditional on $X_s=b$, define the
reversed increments
\[
z_i:=X_{s+1-i}-X_{s-i},
\qquad 1\leq i\leq s.
\]
Then $z_i\in\{-1,0,1\}$ and $\sum_{i=1}^s z_i=q$. Moreover, for every $1\leq j\leq s$, $\sum_{i=1}^j z_i=X_s-X_{s-j}$.
Since $X_t$ changes by at most one at each step, all the partial sums
are positive exactly when $X_s$ first reaches $b$ at time $s$, that is,
when $\tau_b=s$.

Conditional on $X_s=b$, the increments are exchangeable, so the
conditional distribution of $(z_1,\ldots,z_s)$ is invariant under cyclic
shifts. The cycle lemma therefore gives
\[
\Pr[\tau_b=s\mid X_s=b]=\frac{q}{s}.
\]
The case $b<a$ is identical after replacing the increments by their
negatives, and gives the same bound. Thus, in both cases, $
\Pr[\tau_b=s\mid X_s=b]=\frac{|b-a|}{s}$, and hence $\Pr[\tau_b=s]=\frac{|b-a|}{s}\Pr[X_s=b]$.
Consequently, \Cref{eq:balanced-row-probability} is at most
\[
\frac1{|b-a|}
\sum_{s=1}^m\Pr[\tau_b=s]
\leq
\frac1{|b-a|},
\]
which proves \Cref{eq:balanced-row-layer}.

Put $M=\Gamma\circ\Delta_S$. If the query is constant on
$\mathcal A$, then $M=0$. Otherwise define
\[
w(A)=
\begin{cases}
(r-|A\cap S|)^{-1/2},&|A\cap S|<r,\\[4pt]
(|A\cap S|-r+1)^{-1/2},&|A\cap S|\geq r.
\end{cases}
\]
For $A\in\mathcal L_a$ with $a<r$, put $i=r-a$. By
\Cref{eq:balanced-row-layer},
\[
(Mw)(A)
\leq\sum_{j\geq1}\frac{1}{(i+j-1)\sqrt j}
\leq\frac1i\sum_{j\leq i}\frac1{\sqrt j}
   +\sum_{j>i}\frac1{j^{3/2}}
\leq\frac4{\sqrt i}=4w(A)\,.
\]
If $A\in\mathcal L_a$ with $a\geq r$, then put $i=a-r+1$ and write the
layers below the threshold as $b=r-j$, $j\geq1$. Then
$a-b=i+j-1$, so the same calculation gives
$(Mw)(A)\leq4w(A)$. Therefore $Mw\leq4w$ entrywise.
Since $M$ is symmetric and its entries are nonnegative, the weighted Schur
test~\cite{DymKatsnelson2007} and $Mw\leq4w$ imply
\[
\norm{\Gamma\circ\Delta_S}
=
\norm{M}_{2\to2}
\leq4 = O(1).
\]
\end{proof}
\subsection{Lower bound under the promise \texorpdfstring{$|A|=k$}{|A| = k}}

\begin{theorem}
\label{thm:eqk}
Under the promise $|A|=k$, and $1\leq k \leq n-1$
\[
Q=\Omega\!\left(1+\log(\min\{k,n-k\})\right)\,.
\]
In particular, if $n\geq2k$ and $k\geq 2$, then $Q=\Omega(\log k)$.
\end{theorem}

\begin{proof}
Put $m=\min\{k,n-k\}$ and $c=k-m$. The claim is trivial for $m=0$.
Otherwise choose disjoint known sets $C,U\subseteq[n]$ with
$|C|=c$ and $|U|=2m$, and restrict to hidden sets $H_A=C\cup A,
A\in\binom Um
$.
Index $\Gamma$ by the sets $A,B\in\binom Um$ and define it by
\[
\Gamma\llbracket A,B\rrbracket
=
\begin{cases}
\displaystyle
\frac{1}{|A\setminus B|\binom{m}{|A\setminus B|}^{2}}&A\ne B\\[6pt]
0&A=B
\end{cases}\,.
\]
After relabeling $U$ as $[2m]$, the proof of
\Cref{lem:balanced-gamma-norm} gives
$\norm{\Gamma}=\Theta(\log m)$. For a query $S$, put
$T=S\cap U$. Its answer is determined by
\[
|A\cap T|\geq
\left\lceil\frac{k}{2}\right\rceil-|C\cap S|\,.
\]
Thus the proof of \Cref{lem:balanced-mask-norm}, with this shifted
threshold, gives
$\norm{\Gamma\circ\Delta_S}=O(1)$. The positive-weight adversary
bound~\cite{ambainis2002quantum,HoyerLeeSpalek2007} now yields
\[
Q=\Omega(1+\log m)
=\Omega\!\left(\log(\min\{k,n-k\})\right)\,.
\]
When $n\geq2k$, we have $m=k$.
\end{proof}

\subsection{Gate complexity and depth analysis}
\label{app:majority-depth}

We analyze our algorithm in the same circuit model: all-to-all
connectivity with arbitrary unitaries acting on one or two qubits. We also assume that we have access to clean ancillas.  We denote by
$G(\cdot)$ and $D(\cdot)$ gate complexity and depth. For both promises,  the unitary is of the form
\begin{equation*}
    U=\widehat H C^\dagger\widehat U C\widehat H,
    \qquad
    \widehat U=\sum_{t=0}^{n}\ket t\!\bra t_{W}\otimes U_t,
\end{equation*}
where $C$ is the counter from
\Cref{thm:minimum-width-hamming-weight-counter} that computes the Hamming weight of the register $Q$ onto $W$, and 
$\widehat U$ is controlled on $W$.  Thus
\begin{equation}
    G(C)=O(n\log n),\qquad D(C)=O(\log n).
    \label{eq:majority-counter-cost}
\end{equation}
The controlled Hadamard $\widehat H$ uses $O(n)$ gates and $O(\log n)$
depth. Indeed, we can fan out its control with a balanced CNOT tree, and apply the controlled
Hadamards in parallel, then reverse the tree. 

When implementing the controlled version of $U$, we only need to implement the controlled version of $\widehat U$. Indeed, if $\widehat U$ is not applied, then $\widehat H C^\dagger C\widehat H$ collapses to the identity. 
Hence, it is sufficient to analyze the gate complexity and depth of blocks $U_t$ controlled on $\ket{t}_W$ and on the clock state $\ket c\in\mathbb{C}^K$. Then, the overall complexity becomes the following. 

Suppose that every relevant $U_t$ can be implemented with at most $g$ gates and depth at most $d$. There are at most $k+2$ relevant values of $t$. The controlled circuits cannot be applied in parallel on the original target register, since they act on the same qubits. We therefore use coherent routing. Allocate one target register for each relevant value of $t$, together with a one-qubit activity flag for each register. A balanced binary network of controlled-SWAP gates routes the original target register and the activity flag to the leaf selected by $W$. At each level, the controlled-SWAP gates act on disjoint registers and are applied in parallel after fanning out the corresponding address bit. At the leaves, the activity flags and the condition that the clock equals $c$ control the circuits for $U_t$. These circuits now act on disjoint target registers and can be applied in parallel. Finally, reverse the routing network and uncompute all fanout ancillas.

The target register uses $O(\log k)$ qubits in the bounded-size construction and one qubit in the known-size construction. Thus the routing network uses $O(k\log k)$ gates, $O(\log k)$ depth, and $O(k\log k)$ clean workspace qubits. Computing the range and clock conditions costs $O(\log n)$ gates and depth. Fanning out their conjunction to the gates of each block costs $O(kg)$ gates and $O(\log(g+1))$ depth. Adding one control to a one- or two-qubit gate requires only a constant number of one- and two-qubit gates in our circuit model. Consequently,
\begin{equation*}
    G(c\text{-}\widehat U)=O\bigl(k(\log n+g)\bigr),
    \qquad
    D(c\text{-}\widehat U)
    =O\bigl(\log n+d+\log(g+1)\bigr).
\end{equation*}
Together with~\Cref{eq:majority-counter-cost} and the controlled Hadamards, we obtain
\begin{equation*}
    G(c\text{-}U)=O\bigl(n\log n+k(\log n+g)\bigr),
    \qquad
    D(c\text{-}U)=O\bigl(\log n+d+\log(g+1)\bigr).
\end{equation*}
For the known-size construction, the bounds below have $g,d=O(1)$. Since $k\leq n$, the controlled version of $U$ then has $O(n\log n)$ gates and $O(\log n)$ depth. The $K=O(\log k)$ sequential controlled applications in the clocked construction therefore give $O(n\log n\log k)$ non-oracle gates and $O(\log n\log k)$ non-oracle depth. For the bounded-size construction, the resulting bounds are stated after the analysis of the endpoint rotation.

\subsubsection{Gate complexity and depth of \texorpdfstring{$U_t$}{Ut} for known support size: \texorpdfstring{$|A|=k$}{|A|=k}}

We analyze the complexity of the blocks defined in~
\Cref{eq:majority-weight-blocks}. Each $U_t$ is either $I$, $X$, $Z$, or the
one-qubit unitary $U_k$. For consistency with the circuit in
\Cref{fig:majority-unitary-circuit}, for $t>k$ we use the extension that
applies $Z$ when the second least significant bit of $t$ is one and applies
$I$ otherwise. This is allowed because $\ket{A_t}=0$ for $t>k$. Hence
\begin{equation*}
    G(U_t)=O(1),\qquad
    D(U_t)=O(1).
\end{equation*}

\subsubsection{Gate complexity and depth of \texorpdfstring{$U_t$}{Ut} for bounded support size: \texorpdfstring{$1\leq|A|\leq k$}{1 <= |A| <= k}}
\label{app:unknown-size-unitary-implementation}

We first consider odd support sizes with cutoff $2M+1$ on $n$ variables and
use the space
\begin{equation*}
    \mathcal W=\operatorname{span}\{e_*,e_0,\ldots,e_{M-1}\}
\end{equation*}
from~\Cref{sec:majority-unknown-size}. We now give the blocks explicitly.
The block $U_0$ swaps $1\oplus0$ with $0\oplus\ket{e_*}$ and multiplies
$0\oplus\ket{e_i}$ by $(-1)^{i+1}$ for $0\leq i<M$. The coordinate formulas
in~\Cref{eq:unknown-vector-polynomials} at $t=0$ show that this satisfies
every weight-zero relation, including the top layer, where the $e_M$
coordinate is absent. For $t\geq1$, put $h=\lfloor t/2\rfloor$ and define the
diagonal operator on $\mathcal W$ by
\begin{equation*}
    D_t
    :=(-1)^h\ket{e_*}\!\bra{e_*}
      +\sum_{i=0}^{M-1}(-1)^{\min\{i,h\}}
        \ket{e_i}\!\bra{e_i}.
\end{equation*}
The root pattern in~\Cref{eq:unknown-vector-polynomials} gives
\begin{equation*}
    (I_{\mathbb C}\oplus D_t)
    \bigl(0\oplus\ket{a_m(t)}\bigr)
    =0\oplus\ket{b_m(t)}
    \qquad(1\leq t\leq2m).
\end{equation*}
Indeed, if $t=2h+1$, then $P_m(t)=0$ for even $h$ and $Q_m(t)=0$ for odd
$h$. The removed root in $C_{im}$ is the even number $2i+2$, so the same
vanishing holds in every correction coordinate. Thus only $q_m(t)$ is
nonzero when $h$ is even and only $p_m(t)$ is nonzero when $h$ is odd;
$D_t$ has sign $+1$, respectively $-1$, on its support. If $t=2h$, one of
$P_m,Q_m$ vanishes at $t$. In the correction coordinates built from that
polynomial, the only possible exception is $e_{h-1}$, because its root $2h$
was removed. On $e_*$, the sign $(-1)^h$ is $-1$ when $P_m(t)$ is nonzero
and $+1$ when $Q_m(t)$ is nonzero. A nonzero correction coordinate of
$p_m$ has odd index $i$ and therefore sign $-1$ both for $i<h$ and, since
then $h$ is odd, for $i\geq h$; its even-indexed base coordinate, when
present, has the same sign because $h\leq m$ and $h\neq m$. Similarly, every nonzero
correction coordinate of $q_m$ has even index and sign $+1$, and its
odd-indexed base coordinate, when present, has the same sign because
$h\leq m$ and $h\neq m$. The exceptional
coordinate $e_{h-1}$ has the required sign as well.
Since $a_m(t)=(-p_m(t)+(-1)^tq_m(t))/\sqrt2$ and
$b_m(t)=(p_m(t)+(-1)^tq_m(t))/\sqrt2$, this proves all interior relations.

It remains to treat an odd endpoint $t=2j+1$. First suppose $j<M$. For every
larger layer $m>j$,
the correction formula gives
\begin{equation*}
    \frac{C_{jm}(2j+1)}{2j+1}
    =-Q_m(2j+1)\quad\text{if $j$ is even},
    \qquad
    \frac{C_{jm}(2j+1)}{2j+1}
    =-P_m(2j+1)\quad\text{if $j$ is odd}.
\end{equation*}
Consequently, after $D_t$, the $e_*$- and $e_j$-coordinates of every
interior vector lie in the antisymmetric direction
$0\oplus(\ket{e_*}-\ket{e_j})/\sqrt2$. At the endpoint layer $m=j$, all
coordinates outside $e_*,e_j$ already agree between the two sides, and the
antisymmetric coordinate also agrees. Only the plane with ordered basis
\begin{equation*}
    1\oplus0,
    \qquad
    0\oplus\frac{\ket{e_*}+\ket{e_j}}{\sqrt2}
\end{equation*}
remains.

Since $2P_j(2j+1)Q_j(2j+1)=1$, define the nonzero real number
\begin{equation*}
    \rho_j:=
    \begin{cases}
        2P_j(2j+1),&j\text{ even},\\
        -2Q_j(2j+1),&j\text{ odd},
    \end{cases}
    \qquad
    A_j:=-\frac{\rho_j+\rho_j^{-1}}2,
    \qquad
    B_j:=\frac{\rho_j-\rho_j^{-1}}2.
\end{equation*}
In the displayed ordered basis, the remaining endpoint condition is exactly
\begin{equation*}
    \binom{0}{A_j}\longmapsto\binom{1}{B_j}.
\end{equation*}
Because $A_j^2=1+B_j^2$, it is implemented by the real plane rotation
\begin{equation*}
    R_t:=\frac1{A_j}
    \begin{pmatrix}B_j&1\\-1&B_j\end{pmatrix}
\end{equation*}
on this plane, with $R_t$ equal to the identity on the antisymmetric direction
and on all other coordinates.

At the top endpoint $t=2M+1$, the coordinate $e_M$ is absent. All coordinates
$e_i$, $i<M$, already agree after $D_t$, so only the ordered plane
$1\oplus0,0\oplus\ket{e_*}$ remains. Put
\begin{equation*}
    \alpha_M:=\frac{(-1)^{M+1}}{\sqrt2}
        \bigl(P_M(2M+1)+Q_M(2M+1)\bigr),
    \qquad
    \beta_M:=\frac{P_M(2M+1)-Q_M(2M+1)}{\sqrt2}.
\end{equation*}
The remaining endpoint condition is
$\binom{0}{\alpha_M}\mapsto\binom{1}{\beta_M}$. Since
$2P_M(2M+1)Q_M(2M+1)=1$, we have
$\alpha_M^2=1+\beta_M^2$, and hence this condition is implemented on that
plane by
\begin{equation*}
    R_{2M+1}:=\frac1{\alpha_M}
    \begin{pmatrix}\beta_M&1\\-1&\beta_M\end{pmatrix}.
\end{equation*}
The matrix is real orthogonal and fixes every other coordinate. Thus, for
odd $t$, take $U_t=R_t(I_{\mathbb C}\oplus D_t)$; for positive even $t$,
take $U_t=I_{\mathbb C}\oplus D_t$; and for $t>2M+1$, take $U_t=I$.

The gate and depth bounds now follow from three elementary facts.
\begin{enumerate}
    \item The space $\mathbb C\oplus\mathcal W$ has dimension
    $M+2=O(k)$ and is encoded by $q=O(\log k)$ qubits. Distinguishing the
    public basis state from the private labels is an equality test on $q$
    bits.

    \item The signs in $U_0$ and $D_t$ depend only on parity and on the
    comparison of the auxiliary label $i$ with $h$. They therefore use
    $O(\log k)$ gates and $O(\log k)$ depth.

    \item The swap in $U_0$ and the change from $e_*,e_j$ to their symmetric
    and antisymmetric combinations are two-level gates between known
    computational-basis labels. After this change, $R_t$ is a two-level
    rotation between the public label and the label of the symmetric
    direction; reversing the change of basis gives the required rotation in
    the original encoding. The top rotation is directly a two-level gate
    between the public and $e_*$ labels. For any such operation, a reversible
    affine relabeling using $O(q)$ $X$ and CNOT gates first makes the two
    labels differ in one bit. An equality test on the other $q-1$ bits,
    computed into clean ancillas, controls the required one-qubit gate; the
    test and the relabeling are then uncomputed. Thus each two-level operation
    uses $O(q)=O(\log k)$ gates and depth.
\end{enumerate}
Hence every block satisfies
\begin{equation*}
    G(U_t)=O(\log k),\qquad D(U_t)=O(\log k).
\end{equation*}
Substituting these bounds into the coherent-routing construction above and
using $k\leq n$ gives
\begin{equation*}
    G(c\text{-}U)=O(n\log n),\qquad
    D(c\text{-}U)=O(\log n).
\end{equation*}
Finally, the $K=O(\log k)$ sequential controlled applications in the clocked
construction give
\begin{equation*}
    O(n\log n\log k)\ \text{non-oracle gates},
    \qquad
    O(\log n\log k)\ \text{non-oracle depth}.
\end{equation*}
For even support sizes, the lift in~\Cref{eq:unknown-parity-lift} replaces
$n$ by $n+1$ and the odd cutoff by at most $k+1$, so these asymptotic bounds
are unchanged for the full promise $1\leq|A|\leq k$.

\section{Supplementary proofs for the example-query results}\label{app:discretization}
\begingroup
\renewcommand{\thelemma}{\ref{lem:discretization-error}}
\begin{lemma}
    For every finite degree cutoff $D$ and every $M\geq64n^2(D+1)$, the following inequality holds:
    \begin{equation}
        \left(
        \sum_{\alpha\in[D]_0^n}
        \left|
            \widehat f_{w,M}(\alpha)
            -
            \widehat f_w(\alpha)
        \right|^2
        \right)^{1/2}
        \leq
        24n
    \sqrt{\frac{(D+1)}{M}}
    +2\left(
        \frac{8n}{M}
    \right)^{1/4}
    +
    2\sqrt{2n}\,e^{-\pi M/4}.
    \end{equation}
\end{lemma}
\endgroup
\begin{proof}
For every collection of coefficients $c=(c_\alpha)_{\alpha\in[D]_0^n}$ satisfying $\sum_{\alpha\in[D]_0^n}|c_\alpha|^2=1$, we have
\begin{equation}
    \left(
        \sum_{\alpha\in[D]_0^n}
        \left|
            \widehat f_{w,M}(\alpha)
            -
            \widehat f_w(\alpha)
        \right|^2
        \right)^{1/2}
        \leq
        \sup_c
        \abs{
            \sum_{\alpha\in[D]_0^n}
            c_\alpha
            \left(
                \widehat f_{w,M}(\alpha)-\widehat f_w(\alpha)
            \right)
        }.
        \label{target_error}
\end{equation}
For convenience, define
\[
\Psi_c(x):=\sum_{\alpha\in[D]_0^n}c_\alpha h_\alpha(x)\sqrt{\gamma_n(x)}.
\]
Since the normalized Hermite polynomials are orthonormal with respect to $\gamma_n$, we have $\|\Psi_c\|^2=\sum_{\alpha\in[D]_0^n}|c_\alpha|^2=1$. Define also
\[
\Phi_w(x):=f_w(x)\sqrt{\gamma_n(x)},
\qquad
\Phi_{w,M}(x):=Z^{-1/2}f_w(\bar{x})\sqrt{\gamma_n(\bar{x})}\mathbf 1_{[-L,L]^n}(x).
\]
By the definition of $Z$, $\|\Phi_{w,M}\|^2=\|\Phi_w\|^2=1$. Thus,
\begin{equation*}
    \widehat f_w(\alpha)
    =
    \int_{\R^n}
    h_\alpha(x)\sqrt{\gamma_n(x)}
    \Phi_w(x)\,dx,
\end{equation*}
and, since each discretization cell has volume $\Delta^n$,
\begin{equation*}
    \widehat f_{w,M}(\alpha)
    =
    \int_{[-L,L]^n}
    h_\alpha(\bar{x})
    \sqrt{\gamma_n(\bar{x})}
    \Phi_{w,M}(x)\,dx.
\end{equation*}
Hence, we decompose
\begin{align}
    \abs{\sum_{\alpha\in[D]_0^n}
    {c_\alpha}
    \left(
        \widehat f_{w,M}(\alpha)
        -
        \widehat f_w(\alpha)
    \right)}
    &=
    \abs{\int_{[-L,L]^n}
    {\Psi_c(\bar{x})}
    \Phi_{w,M}(x)\,dx
    -
    \int_{\R^n}
    {\Psi_c(x)}
    \Phi_w(x)\,dx}
    \nonumber\\
    &\leq \abs{\int_{[-L,L]^n}
    \left(
        {\Psi_c(\bar{x})}
        -
        {\Psi_c(x)}
    \right)
    \Phi_{w,M}(x)\,dx}
    \nonumber\\
    &\quad+
    \abs{\int_{\R^n}
    {\Psi_c(x)}
    \left(
        \Phi_{w,M}(x)-\Phi_w(x)
    \right)dx} \\
    &\leq \left(
        \int_{[-L,L]^n}
        \left|
            \Psi_c(\bar{x})-\Psi_c(x)
        \right|^2dx
        \right)^{1/2}
    \nonumber\\
    &\quad+
    \left(
        \int_{\R^n}
        \abs{\Phi_{w,M}(x)-\Phi_w(x)}^2dx
    \right)^{1/2}.
    \label{eq:main-error-decomposition}
\end{align}
We now bound these two terms separately.

We consider the first term. Let us start with a one-dimensional normalized Hermite decomposition such that for $t\in \R$, we define $v(t):=\sum_{k=0}^D a_k h_k(t)\sqrt{\gamma(t)}$. For each grid cell $I_j=[z_j,z_j+\Delta)$, we have $\bar t=z_j$ for every $t\in I_j$. Applying the one-sided Poincaré inequality to $u(t):=v(t)-v(z_j)$, which satisfies $u(z_j)=0$, gives
\begin{equation*}
    \int_{I_j}
    |v(t)-v(z_j)|^2\,dt
    \leq
    \frac{4\Delta^2}{\pi^2}
    \int_{I_j}|v'(t)|^2\,dt
    \leq
    \frac{4\Delta^2}{\pi}
    \int_{I_j}|v'(t)|^2\,dt,
\end{equation*}
where we used $4/\pi^2\leq4/\pi$. Summing over all grid cells yields
\begin{equation}
    \int_{-L}^{L}
    |v(\bar t)-v(t)|^2\,dt
    \leq
    \frac{4\Delta^2}{\pi}
    \int_{-L}^{L}|v'(t)|^2\,dt
    \leq
    \frac{4\Delta^2}{\pi}\|v'\|^2.
    \label{eq:sum-cell-poincare}
\end{equation}
We next bound $\|v'\|$. Writing $\phi_k(t) := h_k(t)\sqrt{\gamma(t)}$ (with $\gamma := \gamma_1$) for the normalized probabilists' Hermite functions, we have
\begin{equation*}
    \phi_k'(t)
    =
    \frac12
    \left(
        \sqrt{k}\,\phi_{k-1}(t)
        -
        \sqrt{k+1}\,\phi_{k+1}(t)
    \right).
\end{equation*}
Therefore,
\begin{align}
    \|v'\|
    &\leq
    \frac12
    \left(
        \sum_{k=0}^D
        k|a_k|^2
    \right)^{1/2}
    +
    \frac12
    \left(
        \sum_{k=0}^D
        (k+1)|a_k|^2
    \right)^{1/2}
    \nonumber\\
    &\leq
    \sqrt{D+1}
    \left(
        \sum_{k=0}^D|a_k|^2
    \right)^{1/2},
    \label{eq:v-derivative-bound}
\end{align}
where we used the orthonormality of $\{\phi_k\}_{k\geq0}$ in $L^2(\R)$. Combining~\Cref{eq:sum-cell-poincare} and~\Cref{eq:v-derivative-bound}, and using $\Delta=\sqrt{4\pi/M}$, gives
\begin{equation}
    \left(
        \int_{-L}^{L}
        |v(\bar t)-v(t)|^2\,dt
    \right)^{1/2}
    \leq
    \sqrt{\frac{16(D+1)}{M}}
    \left(
        \sum_{k=0}^D|a_k|^2
    \right)^{1/2},
    \label{eq:one-dimensional-grid-error}
\end{equation}
We now lift this one-dimensional estimate to $n$ dimensions using the same hybrid argument as in~\cite[Proposition~29]{JainIyerSommaBaoJordan2025}.
For $r=1,\ldots,n$, define
\begin{equation*}
    \Psi_c^{(r)}(x)
    :=
    \mathbf 1_{[-L,L]^n}(x)
    \sum_{\alpha\in[D]_0^n}
    c_\alpha
    \left(
        \prod_{i=1}^{r}
        \phi_{\alpha_i}(\bar x_i)
    \right)
    \left(
        \prod_{i=r+1}^{n}
        \phi_{\alpha_i}(x_i)
    \right).
\end{equation*}
Thus, $\Psi_c^{(0)}(x)=\mathbf 1_{[-L,L]^n}(x)\Psi_c(x)$, and $\Psi_c^{(n)}(x)=\mathbf 1_{[-L,L]^n}(x)\Psi_c(\bar x)$. The $r$th hybrid replaces only the $r$th coordinate by its discretized version. By~\Cref{eq:one-dimensional-grid-error}, this replacement has operator norm at most $\zeta:= \sqrt{{16(D+1)}/{M}}$. Moreover, the continuous one-dimensional Hermite synthesis map has norm at most $1$, while its discretized version has norm at most $1+\zeta$. Hence,
\begin{align*}
    \|
        \Psi_c^{(r)}-\Psi_c^{(r-1)}
    \|
    &\leq
    (1+\zeta)^{r-1}
    \zeta
    \left(
        \sum_{\alpha\in[D]_0^n}
        |c_\alpha|^2
    \right)^{1/2}
    \nonumber\\
    &=
    \zeta(1+\zeta)^{r-1},
\end{align*}
where in the last equality we used $\sum_\alpha|c_\alpha|^2=1$. Applying the triangle inequality over the $n$ hybrids, we obtain
\begin{align*}
    \left(
        \int_{[-L,L]^n}
        \left|
            \Psi_c(\bar x)-\Psi_c(x)
        \right|^2dx
    \right)^{1/2}
    &=
    \|
        \Psi_c^{(n)}-\Psi_c^{(0)}
    \|
    \nonumber\\
    &\leq
    \sum_{r=1}^n
    \|
        \Psi_c^{(r)}-\Psi_c^{(r-1)}
    \|
    \nonumber\\
    &\leq
    \zeta
    \sum_{r=0}^{n-1}
    (1+\zeta)^r
    \nonumber\\
    &=
    (1+\zeta)^n-1.
\end{align*}
The assumption $M\geq64n^2(D+1)$ implies $n\zeta\leq\frac12$. Hence, $(1+\zeta)^n-1\leq e^{n\zeta}-1\leq2n\zeta$, and therefore
\begin{equation}
    \left(
        \int_{[-L,L]^n}
        |
            \Psi_c(\bar{x})-\Psi_c(x)
        |^2dx
    \right)^{1/2}
    \leq
    2n
    \sqrt{\frac{16(D+1)}{M}}.
    \label{eq:smooth-hermite-grid-error}
\end{equation}

We now consider the second term of~\Cref{eq:main-error-decomposition} and decompose the difference $\Phi_{w,M}-\Phi_w$ into four contributions:
\begin{align*}
    \Phi_{w,M}(x)-\Phi_w(x)
    &=
    \underbrace{
        \left(
            \frac1{\sqrt Z}-1
        \right)
        f_w(\bar{x})
        \sqrt{\gamma_n(\bar{x})}
        \mathbf 1_{[-L,L]^n}(x)
    }_{\mathcal E_{\mathrm{norm}}(x)}
    \nonumber\\
    &\quad+
    \underbrace{
        f_w(\bar{x})
        \left(
            \sqrt{\gamma_n(\bar{x})}
            -
            \sqrt{\gamma_n(x)}
        \right)
        \mathbf 1_{[-L,L]^n}(x)
    }_{\mathcal E_{\mathrm{Gauss}}(x)}
    \nonumber\\
    &\quad+
    \underbrace{
        \left(
            f_w(\bar{x})-f_w(x)
        \right)
        \sqrt{\gamma_n(x)}
        \mathbf 1_{[-L,L]^n}(x)
    }_{\mathcal E_{\mathrm{bdry}}(x)}
    \nonumber\\
    &\quad-
    \underbrace{
        f_w(x)\sqrt{\gamma_n(x)}
        \mathbf 1_{\R^n\setminus[-L,L]^n}(x)
    }_{\mathcal E_{\mathrm{tail}}(x)}.
\end{align*}
Therefore,
\begin{equation*}
    \|\Phi_{w,M}-\Phi_w\|
    \leq
    \|\mathcal E_{\mathrm{norm}}\|
    +
    \|\mathcal E_{\mathrm{Gauss}}\|
    +
    \|\mathcal E_{\mathrm{bdry}}\|
    +
    \|\mathcal E_{\mathrm{tail}}\|.
\end{equation*}

We bound these four terms in turn.

\begin{itemize}
    \item Applying the same discretization estimate~\Cref{eq:one-dimensional-grid-error} to the degree-zero Hermite function $\sqrt{\gamma_n(x)}$ and same hybrid argument analysis gives
        \begin{equation}
    \|\mathcal E_{\mathrm{Gauss}}\|
    =
    \left(
        \int_{[-L,L]^n}
        \left|
            \sqrt{\gamma_n(\bar{x})}
            -
            \sqrt{\gamma_n(x)}
        \right|^2dx
    \right)^{1/2}
    \leq
    2n\sqrt{\frac{16}{M}}.
    \label{eq:bound-Gauss-term}
    \end{equation}
    \item By~\Cref{prop:halfspace-boundary-error},
\begin{equation}
    \|\mathcal E_{\mathrm{bdry}}\|
    \leq
    2\left(
        \frac{8n}{M}
    \right)^{1/4}.
    \label{eq:bound-boundary-term}
\end{equation}
    \item Since $|f_w(x)|=1$,
\begin{equation}
    \|\mathcal E_{\mathrm{tail}}\|^2
    =
    \int_{\R^n\setminus[-L,L]^n}
    \gamma_n(x)\,dx = \Pr[x\notin[-L,L]^n]\leq
    2n e^{-L^2/2}
    =
    2n e^{-\pi M/2}.\label{eq:bound-tail-term}
\end{equation}
    \item Let
    \[
        A(x):=\sqrt{\gamma_n(\bar x)}\mathbf 1_{[-L,L]^n}(x),
        \qquad
        B(x):=\sqrt{\gamma_n(x)}.
    \]
    By the reverse triangle inequality,
    \begin{align}
        \|\mathcal E_{\mathrm{norm}}\|
        =|\sqrt Z-1|
        &=\bigl|\|A\|-\|B\|\bigr|\\
        &\leq \|A-B\|\\
        &\leq \|\mathcal E_{\mathrm{Gauss}}\|+\|\mathcal E_{\mathrm{tail}}\|\\
        &\leq \frac{8n}{\sqrt M}+\sqrt{2n}\,e^{-\pi M/4}.
        \label{eq:bound-normalization-term}
    \end{align}
\end{itemize}
Combining~\Cref{eq:bound-Gauss-term},~\Cref{eq:bound-boundary-term},~\Cref{eq:bound-tail-term}, and~\Cref{eq:bound-normalization-term}, we obtain
\begin{equation}
    \|\Phi_{w,M}-\Phi_w\|
    \leq
    \frac{16n}{\sqrt M}
    +2\left(
        \frac{8n}{M}
    \right)^{1/4}
    +2\sqrt{2n}\,e^{-\pi M/4}.
    \label{eq:phase-total-error}
\end{equation}
Substituting~\Cref{eq:smooth-hermite-grid-error} and~\Cref{eq:phase-total-error} into~\Cref{eq:main-error-decomposition} and finally~\Cref{target_error} gives
\begin{align}
    \left(
        \sum_{\alpha\in[D]_0^n}
        \left|
            \widehat f_{w,M}(\alpha)
            -
            \widehat f_w(\alpha)
        \right|^2
        \right)^{1/2} 
    &\leq
    24n
    \sqrt{\frac{(D+1)}{M}}
    +2\left(
        \frac{8n}{M}
    \right)^{1/4}
    +
    2\sqrt{2n}\,e^{-\pi M/4}.
\end{align}
This proves the lemma.
\end{proof}

\end{document}